\documentclass[12pt,a4paper,english,draftcls,onecolumn]{IEEEtran}
\usepackage[T1]{fontenc}
\usepackage[latin9]{inputenc}
\usepackage{verbatim}
\usepackage{textcomp}
\usepackage{amsmath}
\usepackage{amsthm}
\usepackage{amssymb}
\usepackage{graphicx}

\makeatletter

\pdfpageheight\paperheight
\pdfpagewidth\paperwidth

\theoremstyle{plain}
\newtheorem{thm}{\protect\theoremname}
\theoremstyle{definition}
\newtheorem{example}[thm]{\protect\examplename}
\theoremstyle{plain}
\newtheorem{cor}[thm]{\protect\corollaryname}
\theoremstyle{plain}
\newtheorem{prop}[thm]{\protect\propositionname}

\@ifundefined{date}{}{\date{}}
\usepackage{babel}
\usepackage{tikz}
\usetikzlibrary{arrows,decorations.pathmorphing,backgrounds,positioning,fit,petri,shadows}

\usepackage{babel}
\providecommand{\corollaryname}{Corollary}
\providecommand{\examplename}{Example}
\providecommand{\propositionname}{Proposition}
\providecommand{\theoremname}{Theorem}

\makeatother

\usepackage{babel}
\providecommand{\corollaryname}{Corollary}
\providecommand{\examplename}{Example}
\providecommand{\propositionname}{Proposition}
\providecommand{\theoremname}{Theorem}

\begin{document}
\title{Network Coding with Random Packet-Index Assignment for Data Collection Networks}
\author{\IEEEauthorblockN{Cedric Adjih\IEEEauthorrefmark{1}, Michel Kieffer\IEEEauthorrefmark{2},
and Claudio Greco\IEEEauthorrefmark{1}\IEEEauthorrefmark{2}\\
 } \IEEEauthorblockA{\IEEEauthorrefmark{1}INRIA Saclay, INFINE Project Team, 1 rue Honoré
d'Estienne d'Orves, 91120 Palaiseau, France\\
} \IEEEauthorblockA{\IEEEauthorrefmark{2}L2S, CNRS-Supélec-Univ Paris-Sud, 3 rue Joliot-Curie,
91192 Gif-sur-Yvette, France\\
 } }
\maketitle
\begin{abstract}
This paper considers data collection using a network of uncoordinated,
heterogeneous, and possibly mobile devices. Using medium
and short-range radio technologies, multi-hop communication is required
to deliver data to some sink. While numerous techniques from managed
networks can be adapted, one of the most efficient (from the energy
and spectrum use perspective) is \emph{network coding} (NC). NC is
well suited to networks with mobility and unreliability, however,
practical NC requires a precise identification of individual packets
that have been mixed together. In a purely decentralized system, this
requires either conveying identifiers in headers along with coded
information as in COPE, or integrating a more complex protocol in
order to efficiently identify the sources (participants) and their
payloads.

A novel solution, Network Coding with Random Packet
Index Assignment (NeCoRPIA), is presented where packet indices in NC headers are
selected in a decentralized way, by choosing them randomly.
Traditional network decoding can be applied when all original
packets have different indices. When this is not the case, \emph{i.e.},
in case of collisions of indices, a specific decoding algorithm is
proposed. A theoretical analysis of its performance in terms of complexity
and decoding error probability is described. Simulation results match
well the theoretical results. 
Comparisons of NeCoRPIA header lengths
with those of a COPE-based NC protocol are also provided.
\end{abstract}


\section{Introduction\label{sec:Intro}}

In smart cities, smart factories, and more generally in modern Internet
of Things (IoT) systems, efficient data collection networks (DCN)
are growing in importance to gather various information related to
the environment and the human activity \cite{GomezSensors2015}. This
trend is accelerated by the development of a wide variety of objects
with advanced sensing and connectivity capabilities, which offer the
possibility to collect information of more diversified nature. This
also leads to the emergence of new DCN modalities such as participatory
sensing or crowdsensing applications \cite{GuoACM2015},
which contrast with classical DCN architectures, where nodes are owned
and fully controlled by one managing authority.

To transmit data in DCN, various communication protocols may be considered,
depending on the radio technology integrated in the sensing devices.
Long- (4G, NB-IoT, 5G, LoRa, SigFox), medium- (WiFi), and short-range
(ZigBee, Bluetooth) radio technologies and protocols have all their
advantages and shortcomings \cite{AlFuqahaCST2015}. 

One of the challenges with new DCN modalities is the distributed nature
of the network with unpredictable stability, high churn rate, and
node mobility. Here, one focuses on DCN scenarios where measurements
from some area are collected using medium- or short-range radio technologies,
which require multi-hop communication \cite{Doppler+2009,ShabbirACM2013,Liu2016,Jung2016,Ding2017,JiangTMC2018}.
For such scenarios, a suitable communication technique is Network
Coding (NC)~\cite{Weiwei2009,Keller+2013,PriorTSG2014,Paramanathan2015,NistorComLet2015,KwonTMC2018}. NC \cite{Ahlswede+2000+IEEE_J_IT}
is a transmission paradigm for multi-hop networks, in which, rather
than merely relaying packets, the intermediate nodes may mix the packets
that they receive. %
In wireless networks, the inherent broadcast capacity of the channel
improves the communication efficiency
\cite{Fragouli+2008,Katti+2008,Chachulski+2007}.

In practical NC (PNC) protocols, the mixing of the packets is achieved
through linear combinations. The corresponding coding coefficients
are generally included in each mixed packet as an \emph{encoding vector}~\cite{Chou+2003+A3C}
(see Figure~\ref{fig:encodingChou}). In this way, a coefficient
can be uniquely associated with the corresponding packet: the downside
is the requirement for a global indexing of all packets in the network.
When a single source generates and broadcasts network-coded packets
(intra-flow NC) as in \cite{Chachulski+2007,SundararajanPIEEE2011}, such indexing is easily performed, since the source
controls the initial NC headers of all packets. When the packets of
several sources are network coded together (inter-flow NC \cite{XieCN2015}), a global
packet indexing is difficult to perform in a purely distributed system,
where sources may appear, move, and disappear. This is also true when intra and inter-flow NC is performed \cite{SeferogluInfocom2011}.

Moreover, even in a static network of $N$ sensor nodes, each potentially
generating a single packet, assuming that all packets may be NC together
in some Galois field $\mathbb{F}_{q}$, headers of $N\log_{2}q$ bits would be required.
In practice, only a subset of packets are NC together, either due
to the topology of the network, or to constraints imposed on the way
packets are network-coded \cite{Jafari+2009+ISIT,Keller+2013,Feizi2014}. This
property allows NC headers to be compressed~\cite{Jafari+2009+ISIT,Thomos+2012+IEEE_J_COML,GligoroskiICC2015},
but does not avoid a global packet indexing. This indexing issue has
been considered in COPE \cite{Katti+2008}, where each packet to be
network coded is identified by a 32-bit hash of the IP source address
and IP sequence number. Such solution is efficient when few packets
are coded, but leads to large NC headers when the number of coded
packets increases.

The major contribution of this paper is Network Coding with Random Packet-Index
Assignment (NeCoRPIA), an alternative approach to COPE, addressing
global packet indexing, while keeping relatively
compact NC headers. With NeCoRPIA, packet indices in NC headers
are selected in a decentralized way, by simply choosing them randomly.
In dynamic networks, NeCoRPIA does not require an agreement among
the nodes  on a global packet indexing. In a DCN, when packets generated
by a small proportion of nodes have to be network coded, with NeCoRPIA,
NC headers of a length proportional to the number of nodes generating
data are obtained.

This paper reviews some related work in Section~\ref{sec:related-work}.
Section~\ref{sec:architecture} presents the architecture and protocol
of NeCoRPIA, our version of PNC dedicated to data collection. Section~\ref{sec:decoding}
describes network decoding techniques within the NeCoRPIA framework.
Section~\ref{sec:Complexity} analyzes the complexity of the proposed
approach. Section~\ref{sec:PerfEval} evaluates the performance of
NeCoRPIA in terms of decoding error. Section~\ref{sec:Comparison-of-packet}
provides simulations of a DCN to compare the average
header length of NeCoRPIA with plain NC and with a COPE-inspired approach.
Finally, Section~\ref{sec:conclusion} provides some conclusions
and future research directions.

\section{Related Work\label{sec:related-work}}

The generic problem of efficiently collecting information from multiple
sources to one or several sinks in multi-hop networks has been extensively
studied in the literature: for instance, for static deployments of
wireless sensor networks, a routing protocol such as the Collection
Tree Protocol (CTP) \cite{Gnawali+2009} is typical. Focusing on the
considered data collection applications \cite{GuoACM2015}, performance
can be improved by the use of NC, as exemplified by \cite{Fragouli+2008}
where NC is shown to outperform routing in a scenario with multiple
sources (all-to-all broadcast) or by \cite{Keller+2013,PriorTSG2014,NistorComLet2015}, where NC
is employed for data collection in a sensor network, see also \cite{ShabbirACM2013,Liu2016,Ding2017}. Combinations of NC and opportunistic routing have also been considered, see \cite{KafaeCST2018} and the references therein.

NeCoRPIA addresses two issues related to the use of NC in DCN, namely
the need for a global packet indexing and the overhead related to
NC headers.

An overview of NC header compression techniques has been provided
in \cite{GligoroskiICC2015}. For example, \cite{Heide+2011} explores
the trade-off between field size and generation (hence encoding vector)
size. In \cite{Jafari+2009+ISIT,LiCL2010}, encoding vectors are compressed
using parity-check matrices of channel codes yielding gains when only
limited subsets of linear combinations are possible. In \cite{Thomos+2012+IEEE_J_COML}
a special coding scheme permits to represent NC headers with one single
symbol at the expense of limited generation sizes. In Tunable Sparse NC (TNSC) 
\cite{Feizi2014,GarridoTCOM2017}, COPE-inspired headers are used and the sparsity level of NC headers 
is optimized by controlling the number of network-coded packets. Similar results may also 
be obtained by a dynamic control of the network topology as proposed in \cite{KwonTMC2018}.
In Fulcrum NC \cite{LucaniACCESS2018}, $g$ packets generated by a source are first encoded using 
a channel code over some extension field, \emph{e.g.}, $\mathbb{F}_{2^8}$, using a (systematic) Reed-Solomon code to generate $r$ redundancy packets. The $g+r$ packets are then network coded over $\mathbb{F}_2$. Powerful receivers may retrieve the $g$ original packet from $g$ independent linear combinations seen as mixtures over the extension field, whereas limited receivers require $g+r$ independent combinations to be decoded over $\mathbb{F}_2$. NC headers of $g+r$ bits are thus necessary. Nevertheless, this approach does not address the global packet indexing issue and is better suited to intra-flow NC.

Going further, the encoding vectors may be entirely removed from the
header of packets. Network decoding may then be seen as a \emph{source
separation} problem using only information about the content of packets.
Classical source separation aims at recovering source vectors defined
over a field (\emph{e.g.}, $\mathbb{R}$ or $\mathbb{C}$) from observations
linearly combined through an \emph{unknown} encoding matrix~\cite{Comon+2010}.
In $\mathbb{R}$, one classical approach to source separation is \emph{Independent
Component Analysis} (ICA)~\cite{Comon+2010}, 
which estimates the sources as the set of linear combinations that
minimizes the joint entropy and the mutual similarity among vectors.
In previous work, we have proposed different techniques that exploit
ICA over finite fields \cite{Yeredor+2011+IEEE_J_IT} 
in the context of NC. In~\cite{Nemoianu+2014+IEEE_J_MM}, network-coded
packets without encoding vectors are decoded using entropy minimization
jointly with channel encoding, while in~\cite{Greco+2014+NETCOD},
we exploit the redundancy introduced by communication protocols to
assist the receiver in decoding.
The price to pay is a significantly larger
decoding complexity compared to simple Gaussian elimination when considering
classical NC, which prohibits considering large generation sizes.

When packets
from several sensor nodes have to be network coded, employing the
format of NC vectors proposed by~\cite{Chou+2003+A3C}  requires
a coordination among nodes. This is necessary to avoid two nodes using
the same encoding vector when transmitting packets. Figure~\ref{fig:encodingChou}
(left) shows four sensor nodes, each generating a packet supplemented
by a NC header, which may be seen as a vector with entries in the
field $\mathbb{F}$ in which NC operations are performed. A different
base vector is chosen for each packet to ensure that packets can be
unmixed at receiver side. The coordination among nodes is required
to properly choose the number of entries in $\mathbb{F}$ that will
build up the NC header and the allocation of base vectors among nodes.
\begin{figure}
\centering \includegraphics[width=0.33\textwidth]{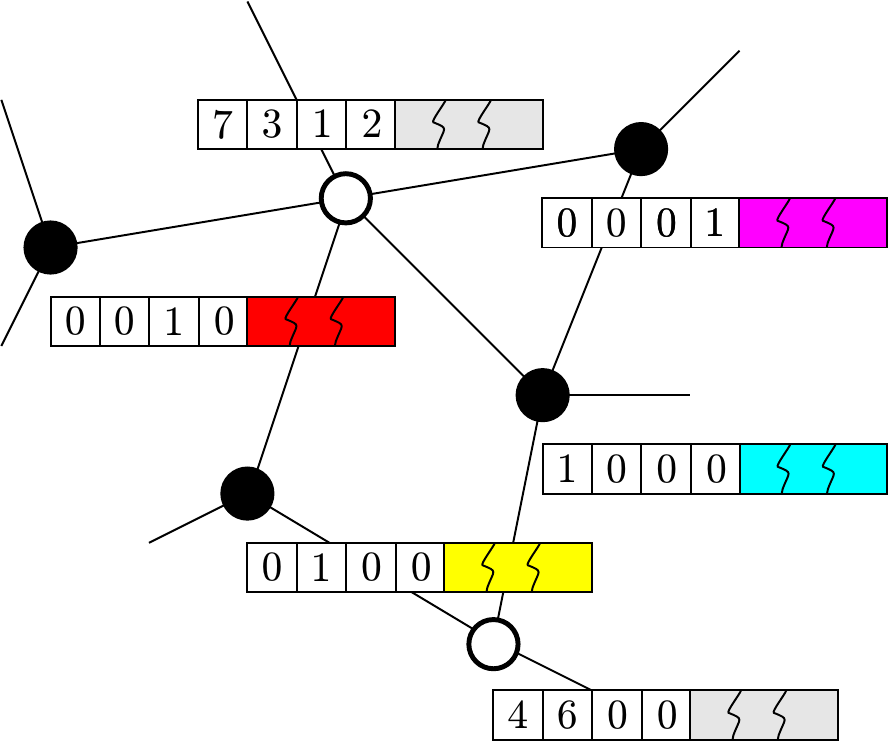}
\includegraphics[width=0.33\textwidth]{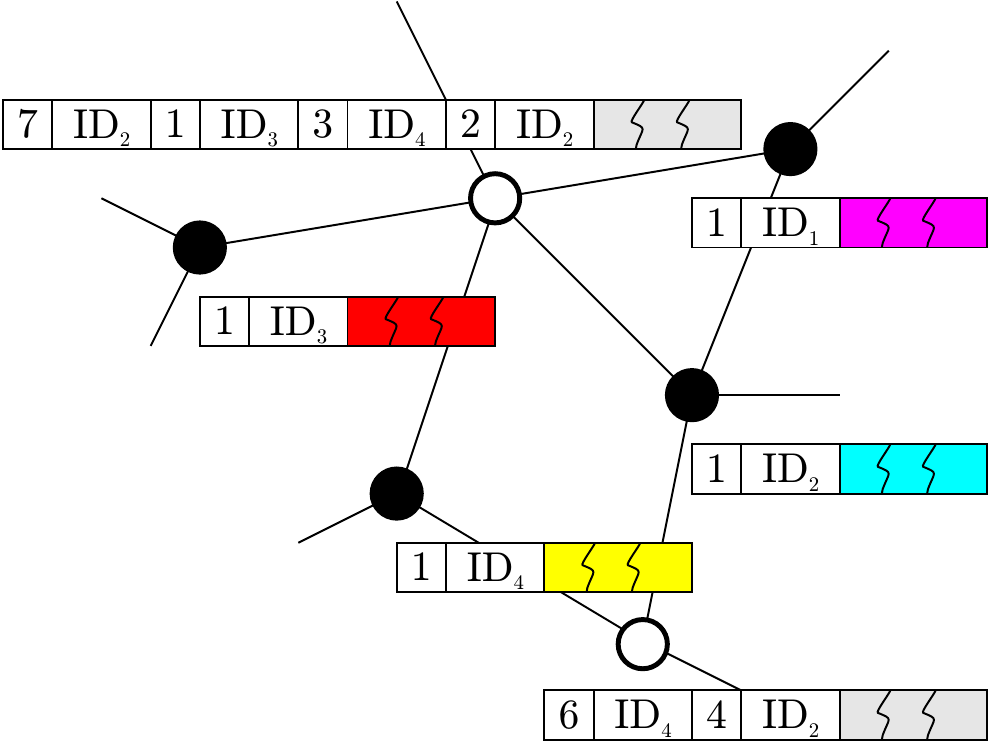}\includegraphics[width=0.33\textwidth]{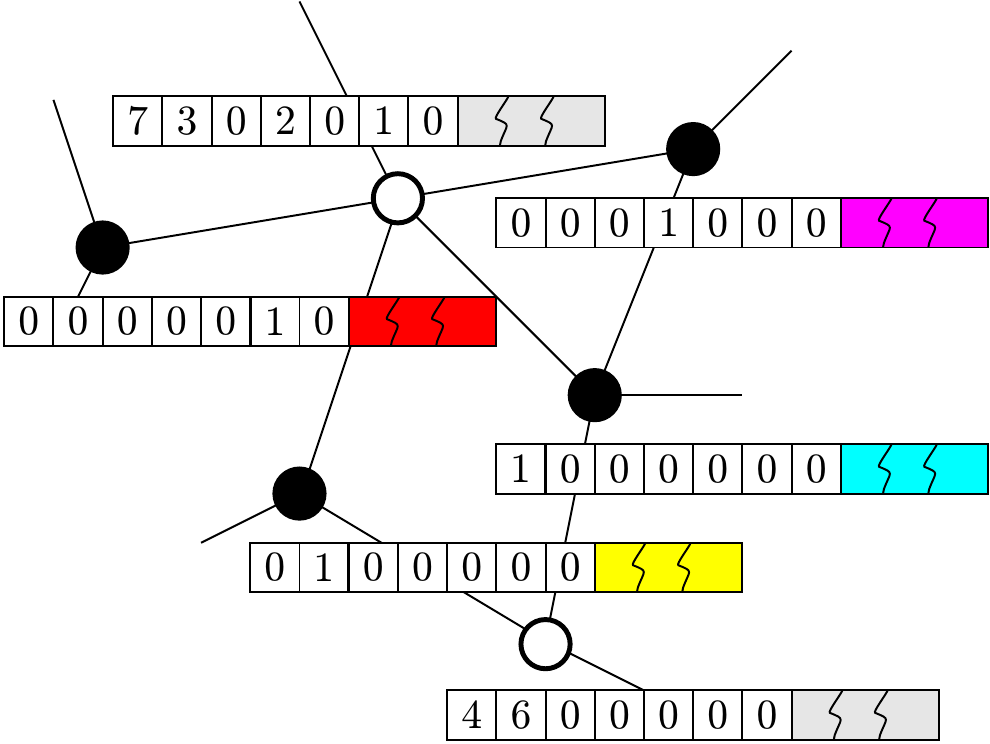}
\caption{NC headers as proposed by \cite{Chou+2003+A3C} (left), by \cite{Katti+2008}
(middle) and NeCoRPIA NC headers (with $n_{\text{v}}=1$) (right)
in the case of inter-session coding of packets generated by different
sensor nodes; nodes generating data packets are in black and nodes
relaying packets are in white\label{fig:encodingChou}}
\end{figure}
COPE \cite{Katti+2008} or TSNC \cite{Feizi2014} employ a header format including one identifier (32-bit
hash for COPE) for each of the coded packets, see Figure~\ref{fig:encodingChou}
(middle). When several packets are network coded together, their identifier
and the related NC coefficient are concatenated to form the NC header.
There is no more need for coordination among nodes and the size of
the NC header is commensurate with the number of mixed packets. Nevertheless,
the price to be paid is a large increase of the NC header length compared
to a coordinated approach such as that proposed in~\cite{Chou+2003+A3C}.

To the best of our knowledge, NeCoRPIA represents the only alternative
to COPE to perform a global packet indexing in a distributed way,
allowing packets generated by several uncoordinated sources to be
efficiently network coded. A preliminary version of NeCoRPIA was first
presented in \cite{GrecoICC2015}, where a simple random NC
vector was considered, see Figure~\ref{fig:encodingChou} (right).
Such random assignment is simple, fast, and fully distributed but
presents the possibility of \emph{collisions}, that is, two packets
being assigned to the same index by different nodes. Here, we extend
the idea in \cite{GrecoICC2015}, by considering a random assignment
of several indexes to each packet. This significantly reduces the
probability of collision compared to a single index, even when represented
on the same number of bits as several indexes. Since collisions cannot
be totally avoided, 
an algorithm to decode the received
packets in spite of possible collisions is also proposed.
We additionally detail how this approach is encompassed in a practical
data collection protocol.

\section{NeCoRPIA Architecture and Protocol\label{sec:architecture}}

\subsection{Objective of the data collection network\label{sub:formulation}}

A data collection architecture is considered with a set $\mathcal{S}=\{S_{1},\ldots,S_{N}\}$
of $N$ static nodes gathering measurements performed by possibly
mobile sensing nodes. Each node $S_{i}$, located in $\boldsymbol{\theta}_{i}$,
in some reference frame $\mathcal{F}$, acts as a data collection
point for all mobile nodes located in its assigned area $\mathcal{R}_{i}$
(\emph{e.g.}, its Voronoi cell) as in Figure~\ref{fig:scenario}.
The data consist, for instance, in a set of measurements of some physical
quantity $D(\boldsymbol{\xi})$, \emph{e.g.}, temperature, associated
with the vector $\boldsymbol{\xi}$, representing the experimental
conditions under which the measurements were taken (location, time
instant, regressor vector in case of model linear in some unknown
parameters). Note that for crowdsensing applications, the identity
of the node that took the measurements is often secondary, provided
that there are no node polluting the set of measurements with outliers.

\begin{figure}
\centering \includegraphics[width=0.5\columnwidth]{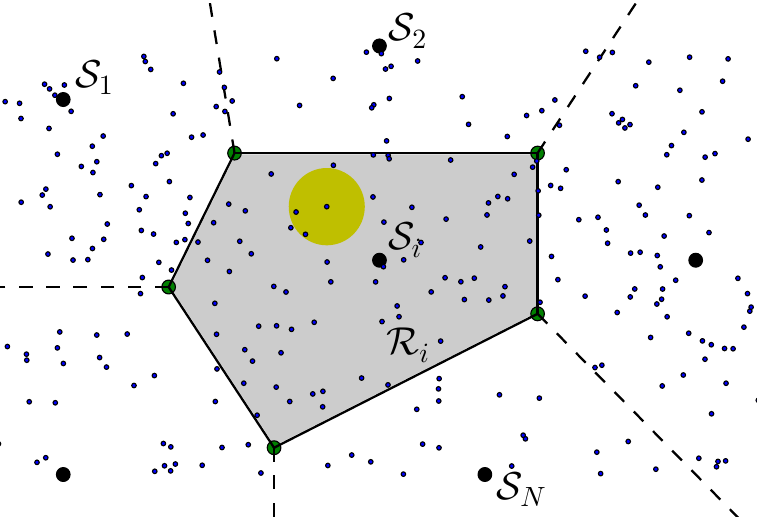}

\protect\protect\caption{Data collection scenario: the static node $S_{i}$ gathers measurements
from all mobile nodes (with limited communication range) in its assigned
data collection area $\mathcal{R}_{i}$ in grey\label{fig:scenario}}
\end{figure}
Typically, a mobile node measures periodically $D$ under the experimental
conditions $\boldsymbol{\xi}$. The data collection network objective
is to collect the tuples $(\boldsymbol{\xi},\mathbf{d})$, where $\mathbf{d}=D(\boldsymbol{\xi})$,
to the appropriate collection point $S_{i}$ (responsible for the
area $\mathcal{R}_{i}$ where the mobile node finds itself).

\subsection{General Architecture}

For fulfilling the objectives of the previous section, a communication
architecture based on NC is designed where information is propagated
to the closest sink through dissemination of coded packets.

\begin{figure}
\centering \includegraphics[width=1\columnwidth]{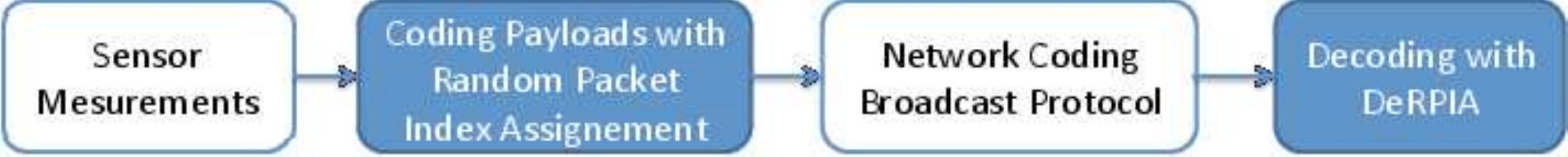}

\caption{Modules of NeCoRPIA architecture\label{fig:architecture-overview}}
\end{figure}
Figure~\ref{fig:architecture-overview} represents the modules involved
in NeCoRPIA. The sensing module is in charge of collecting local sensor
information with associated experimental conditions, \emph{i.e.},
the tuple $\left(\boldsymbol{\xi},\mathbf{d}\right)$. The encoding
module takes as input $\left(\boldsymbol{\xi},\mathbf{d}\right)$
and creates packets containing the data payload, and our specific
NeCoRPIA header. The NC protocol aims at ensuring that all (coded)
packets reach the collection points, by transmitting at intermediate
nodes (re)combinations of received packets. The decoding module at
the collection points applies the algorithm from Section~\ref{sec:decoding}
to recover the experimental data collected in the considered area
$\mathcal{R}_{i}$.

\subsection{Network Encoding Format}

\label{sub:NetworkEncodingFormat}

Figure~\ref{fig:general-format} represents the general packet format
used in NeCoRPIA: it includes a control header used by the NC dissemination
protocol, followed by an encoded content, considered as symbols from
$\mathbb{F}_{q}$, the Galois field with $q$ elements. 
\begin{figure}
\begin{centering}
\includegraphics[width=0.5\textwidth]{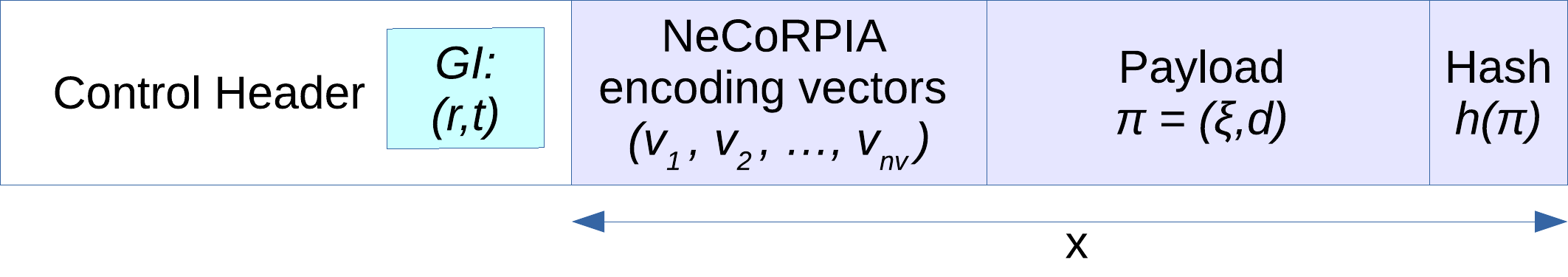} 
\par\end{centering}
\protect\protect\caption{NeCoRPIA packet format\label{fig:general-format}}
\end{figure}
The control header itself includes, as generation identifier (GI),
a \emph{spatio-temporal slot} (STS) $\left(r,t\right)$, where $r$
is the index of the sink $\mathcal{S}_{r}$ and $t$ is the index
of the considered time slot, within which the data has been collected.
Only packets with the same GI are combined together by the protocol.
The rest of the packet can be formally written as a vector $\mathbf{x}$
of $L_{\text{x}}$ entries in $\mathbb{F}_{q}$ as 
\begin{equation}
\mathbf{x}=\left(\mathbf{v}_{1},\dots,\mathbf{v}_{n_{\text{v}}},\mathbf{\boldsymbol{\pi}},\mathbf{h}\right),\label{eq:PacketFormat}
\end{equation}
where $\mathbf{v}_{\ell}\in\mathbb{F}_{q}^{L_{\ell}}$, $\ell=1,\dots,n_{\text{v}}$
represent the encoding subvectors, $\mathbf{\boldsymbol{\pi}}=(\boldsymbol{\xi},\mathbf{d})$
is the payload, where $\boldsymbol{\xi}$ and $\mathbf{d}$ are represented
with finite precision on a fixed number $L_{\mathbf{\boldsymbol{\pi}}}$
of symbols in $\mathbb{F}_{q}$, and $\mathbf{h}=h(\mathbf{\boldsymbol{\pi}})\in\mathbb{F}_{q}^{L_{\text{h}}}$
is the hash of $\mathbf{\boldsymbol{\pi}}$. In classical NC, $n_{\text{v}}=1$,
and $\mathbf{v}_{1}$ corresponds to one of the canonical vectors
of $\mathbb{F}_{q}^{L_{1}}.$ The choice of the canonical vector requires
an \emph{agreement} among mobile sensing nodes, to avoid the same
vector being selected by two or more mobile nodes in the same STS.

To avoid this resource-consuming agreement step, when a source generates
a packet, NeCoRPIA assigns \emph{random} canonical vectors $\mathbf{e}_{i}\in\mathbb{F}_{q}^{L_{\ell}}$
to each $\mathbf{v}_{\ell}$, $\ell=1,\dots,n_{\text{v}}$. For each
new payload $\mathbf{\boldsymbol{\pi}}$, the random NC vector is
then represented by $\left(\mathbf{v}_{1},\dots,\mathbf{v}_{n_{\text{v}}}\right)$.
One may choose $n_{\text{v}}=1$, but in this case, $L_{1}$ should
be quite long to avoid collisions, even for a moderate number of packets
in each STS (this is reminiscent to the birthday paradox \cite{Holst1986}),
see Section~\ref{sub:Nvv1}. This results in an encoding vector with
the format represented in Figure~\ref{fig:general-format} for
the case $n_{\text{v}}=1$.

The hash $\mathbf{h}$ is included to assist the decoding process
in case of collisions.

\subsection{Network Coding Protocol}

\label{subsec:Protocol}

The NC protocol is in charge of ensuring that the coded packets are
properly reaching the data collection points for later decoding. Since
the main contribution of our method lies in other parts, we only provide
the sketch of a basic protocol. It operates by broadcasting measurements
to all nodes within each area $\mathcal{R}_{i}$, with NC (in the
spirit of \cite{Fragouli+2008}): with the effect that the collection
point $S_{i}$ will gather the information as well. It is a multi-hop
protocol relying on the \emph{control header}, shown in Figure~\ref{fig:general-format},
to propagate control information to the entire network (as DRAGONCAST
\cite{dragoncast-draft} does for instance). The control headers are
generated by the data collection points and copied in each encoded
packet by the nodes.

The baseline functioning is as follows: at the beginning of each STS,
$S_{i}$ initiates data collection by generating packets with an empty
payload, and with a source control header holding various parameters
such as: number of encoding subvectors $n_{\text{v}}$ and size of
each encoding vector $L_{\ell}$, $\ell=1,\dots,n_{\text{v}}$, the
buffer size $G_{B}$, the STS $(r,t)$, along with a compact description
of its area $\mathcal{R}_{i}$, sensing parameters, \emph{etc.} Upon
receiving packets from a data collection point or from other nodes,
and as long its current position matches $\mathcal{R}_{i}$, a node
periodically\footnote{or immediately as in \cite{Fragouli+2008} }
retransmits (coded) packets with the most up-to-date control header.
When several measurements are taken within the same STS, packets with
different random encoding vectors should be generated. Furthermore,
each node maintains a buffer of (at most) $G_{B}$ coded vectors:
when a packet associated to a given STS is received, it is (linearly)
combined with all coded packets associated with the same STS in the
buffer. Likewise, when a packet is generated in some STS, the node
computes a linear combination of all coded packets belonging to the
same STS and stored in the buffer. $S_{i}$ (indirectly) instructs
nodes to stop recoding of packets of a STS (and to switch to the next
one) through proper indication in the control header. Note that many
improvements of this scheme exist or can be designed.

\section{Estimation of the transmitted packets\label{sec:decoding}}

Assume that within a STS $\left(r,s\right)$, mobile sensing nodes
have generated a set of packets $\mathbf{x}_{1},\ldots,\mathbf{x}_{g}$,
which may be stacked in a matrix $\mathbf{X}$. Assume that $g'\geqslant g$
linear combinations of the packets $\mathbf{x}_{1},\ldots,\mathbf{x}_{g}$
are received by the data collection point $\mathcal{S}_{r}$ and collected
in a matrix $\mathbf{Y}'$ such that 
\begin{equation}
\mathbf{Y}'=\mathbf{A}'\mathbf{X}=\mathbf{A}'\left(\mathbf{V}_{1},\dots,\mathbf{V}_{n_{\text{v}}},\mathbf{P}\right),\label{eq:RecPackets}
\end{equation}
where $\mathbf{A}'$ represents the NC operations that have been performed
on the packets $\mathbf{x}_{1},\ldots,\mathbf{x}_{g}$. $\mathbf{V}_{1},\dots,\mathbf{V}_{n_{\text{v}}}$,
and $\mathbf{P}$ are matrices which rows are the corresponding vectors
$\mathbf{v}_{1,i},\dots,\mathbf{v}_{n_{\text{v}},i}$, and $\mathbf{p}_{i}=\left(\boldsymbol{\pi}_{i},\mathbf{h}_{i}\right)\in\mathbb{F}_{q}^{L_{\mathbf{p}}}$
of the packets $\mathbf{x}_{i}$, $i=1,\dots,g$, with $L_{\text{p}}=L_{\mathbf{\boldsymbol{\pi}}}+L_{\text{h}}$.
If enough linearly independent packets have been received, a full-rank
$g$ matrix \textbf{$\mathbf{Y}$} may be extracted by appropriately
selecting\footnote{We assume that even if packet index collisions have occurred (which
means $\text{rank}\left(\mathbf{V}\right)<g$) the measurement process
is sufficiently random to ensure that $\mathbf{X}$ and thus $\mathbf{Y}$
have full rank $g$.} $g$ rows of $\mathbf{Y}'$. The corresponding $g$ rows of $\mathbf{A}'$
form a $g\times g$ full-rank matrix $\mathbf{A}$. Then, \eqref{eq:RecPackets}
becomes 
\begin{equation}
\mathbf{Y}=\mathbf{A}\left(\mathbf{V}_{1},\dots,\mathbf{V}_{n_{\text{v}}},\mathbf{P}\right).\label{eq:RecPacketsFullRank}
\end{equation}
The problem is then to estimate the packets $\mathbf{x}_{1},\ldots,\mathbf{x}_{g}$
from the received packets in $\mathbf{Y}$, without knowing $\mathbf{A}$.

Three situations have to be considered. The first is when the rows
of $\left(\mathbf{V}_{1},\dots,\mathbf{V}_{n_{\text{v}}}\right)$
are linearly independent due to the presence of some $\mathbf{V}_{\ell}$
of full rank $g$. The second is when the rows of $\left(\mathbf{V}_{1},\dots,\mathbf{V}_{n_{\text{v}}}\right)$
are linearly independent but there is no full rank $\mathbf{V}_{\ell}$.
The third is when the rank of $\left(\mathbf{V}_{1},\dots,\mathbf{V}_{n_{\text{v}}}\right)$
is strictly less than $g$, but the rank of $\mathbf{Y}$ is equal
to $g$. These three cases are illustrated in Examples~\ref{exa:NoCollision}-\ref{exa:Cycle}.
In the last two situations, a specific decoding procedure is required,
which is detailed in Section~\ref{sub:Collisions}.
\begin{example}
\label{exa:NoCollision}Consider a scenario where three nodes generate
packets with $n_{\text{v}}=2$ random coding subvectors in $\mathbb{F}_{2}^{L_{\ell}}$
with $L_{1}=L_{2}=3$. When the generated coding vectors are $\left(\left(1,0,0\right),\left(0,1,0\right)\right)$,
$\left(\left(1,0,0\right),\left(0,0,1\right)\right)$, and $\left(\left(0,0,1\right),\left(1,0,0\right)\right)$,
two nodes have selected the same first coding subvector, but all second
coding subvectors are linearly independent, which allows one to recover
the original packets via Gaussian elimination. This situation is illustrated
in Figure~\ref{fig:morpion} ($a$), where each coding vector may
be associated to a point in a $3\times3$ grid, the first coding vector
representing the row index and the second coding subvector the column
index. Three different columns have been selected, decoding can be
performed via Gaussian elimination on the second coding subvectors. 
\end{example}
\begin{figure}
\centering \includegraphics{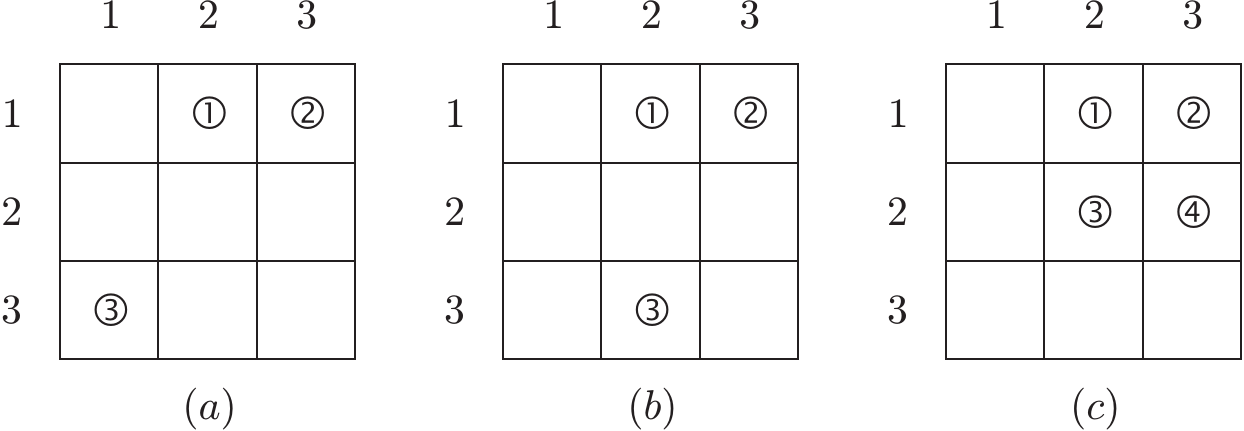}

\caption{($a$) Illustration of Example~\ref{exa:NoCollision}: No collision
in the second subvectors; ($b$) Illustration of Example~\ref{exa:CollisionsFullRank}:
single collisions in both subvectors; ($c$) Illustration of Example~\ref{exa:Cycle}:
single collisions in both subvectors leading to a cycle and a rank
deficiency.\label{fig:morpion}}
\end{figure}
\begin{example}
\label{exa:CollisionsFullRank}Consider the same scenario as in Example~\ref{exa:NoCollision}.
When the generated coding vectors are $\left(\left(1,0,0\right),\left(0,1,0\right)\right)$,
$\left(\left(1,0,0\right),\left(0,0,1\right)\right)$, and $\left(\left(0,0,1\right),\left(0,1,0\right)\right)$,
collisions are observed in the first and second coding subvectors,
but $\left(\mathbf{V}_{1},\mathbf{V}_{2}\right)$ is of full rank
$g=3$. This situation is illustrated in Figure~\ref{fig:morpion}
($b$): three different entries have been chosen randomly, decoding
will be easy. 
\end{example}
\begin{example}
\label{exa:Cycle}Consider now a scenario where four nodes generate
packets with $n_{\text{v}}=2$ random coding subvectors in $\mathbb{F}_{2}^{L_{\ell}}$
with $L_{1}=L_{2}=3$. When the randomly generated coding vectors
are $\left(\left(1,0,0\right),\left(0,1,0\right)\right)$, $\left(\left(1,0,0\right),\left(0,0,1\right)\right)$,
$\left(\left(0,1,0\right),\left(0,1,0\right)\right)$, and $\left(\left(0,1,0\right),\left(0,0,1\right)\right)$,
the rank of $\left(\mathbf{V}_{1},\mathbf{V}_{2}\right)$ is only
three. $\mathbf{Y}$ will be of full rank $g=4$ only if $\left(\mathbf{V}_{1},\mathbf{V}_{2},\mathbf{P}\right)$
is of full rank. This situation is illustrated in Figure~\ref{fig:morpion}
($c$): even if different entries have been chosen, the rank deficiency
comes from the fact that the chosen entries may be indexed in such
a way that they form a cycle. 
\end{example}

\subsection{Decoding via Gaussian elimination}

\label{sub:GaussElim}

When $\mathbf{A}$ is a $g\times g$ full-rank matrix, a necessary
and sufficient condition to have the rank of one of the matrices $\mathbf{AV}_{\ell}$
equal to $g$ is that all mobile sensing nodes have chosen a different
canonical subvector for the component $\mathbf{v}_{\ell}$ of the
NC vector. Decoding may then be performed via usual Gaussian elimination
on $\mathbf{AV}_{\ell}$, as in classical NC. As will be seen in Section~\ref{sec:PerfEval},
this event is unlikely, except for large values of $L_{\ell}$ compared
to the number of packets $g$.

\subsection{Decoding with packet index collisions}

\label{sub:Collisions}

When $\mathbf{A}$ is a $g\times g$ full-rank matrix, the rank of
$\mathbf{AV}_{\ell}$ is strictly less than $g$ when at least two
rows of $\mathbf{V}_{\ell}$ are identical, \emph{i.e.,} two nodes
have chosen the same canonical subvector. This event is called a \emph{collision}
in the $\ell$-th component of the NC vector.

When $\mathbf{A}$ is a $g\times g$ full-rank matrix, the rank of
$\mathbf{A}\left[\mathbf{V}_{1},\dots,\mathbf{V}_{n_{\text{v}}}\right]$
is strictly less than $g$ when the rows of $\left[\mathbf{V}_{1},\dots,\mathbf{V}_{n_{\text{v}}}\right]$
are linearly dependent. This may obviously occur when the NC vector
chosen by two nodes are identical, \emph{i.e.}, there is a collision
in \emph{all} $n_{\text{v}}$ components of their NC vector. This
also occurs when there is no such collision, when nodes have randomly
generated linearly dependent NC subvectors, as illustrated in Example~\ref{exa:Cycle},
see also Section~\ref{sec:Complexity}.

\subsubsection{Main idea}

In both cases, one searches a full rank matrix $\mathbf{W}$ such
that $\mathbf{X=W}\mathbf{Y}$ up to a permutation of the rows of
$\mathbf{X}$. We propose to build this \emph{unmixing} matrix $\mathbf{W}$
row-by-row exploiting the part of the packets containing the $n_{\text{v}}$
NC subvectors $\left(\mathbf{A}\mathbf{V}_{1},\dots,\mathbf{\mathbf{A}V}_{n_{\text{v}}}\right)$,
which helps defining a subspace in which admissible rows of $\mathbf{W}$
have to belong. Additionally, one exploits the content of the packets
and especially the hash $h\left(\boldsymbol{\pi}\right)$ introduced
in Section~\ref{sub:NetworkEncodingFormat} to eliminate candidate
rows of $\mathbf{W}$ leading to inconsistent payloads.

\subsubsection{Exploiting the collided NC vectors}

For all full rank $g$ matrix $\mathbf{Y}$, there exists a matrix
$\mathbf{T}$ such that $\mathbf{TY}$ is in reduced row echelon form
(RREF) 
\begin{equation}
\mathbf{T}\mathbf{Y}=\left(\begin{array}{ccccc}
\mathbf{B}_{11} & \mathbf{B}_{12} &  & \mathbf{B}_{1n_{\text{v}}} & \mathbf{C}_{1}\\
\mathbf{0} & \mathbf{B}_{22}\\
\mathbf{0} & \mathbf{0} & \ddots &  & \vdots\\
\vdots &  & \ddots & \mathbf{B}_{n_{\text{v}}n_{\text{v}}}\\
\mathbf{0} & \cdots & \cdots & \mathbf{0} & \mathbf{C}_{n_{\text{v}}+1}
\end{array}\right),\label{eq:Echelon}
\end{equation}
where $\mathbf{B}_{\ell\ell}$ is a $\rho_{\ell}\times L_{\ell}$
matrix with $\text{rank}\left(\mathbf{B}_{\ell\ell}\right)=\rho_{\ell}$.
Since $\textrm{rank}\left(\mathbf{Y}\right)=g$, $\mathbf{C}_{n_{\text{v}}+1}$
is a $\rho_{n_{\text{v}}+1}\times L_{\text{p}}$ matrix with $\text{rank}\left(\mathbf{C}_{n_{\text{v}}+1}\right)=\rho_{n_{\text{v}}+1}=g-\sum_{\ell=1}^{n_{\text{v}}}\rho_{\ell}$.
The matrix $\mathbf{B}_{11}$ is of rank $\rho_{1}$ and its rows
are $\rho_{1}$ different vectors of $\mathbf{V}_{1}$.

One searches now for generic unmixing row vectors $\mathbf{w}$ of
the form $\mathbf{w}=\left(\mathbf{w}_{1},\mathbf{w}_{2},\dots,\mathbf{w}_{n_{\text{v}}},\mathbf{w}_{n_{\text{v}}+1}\right)$,
with $\mathbf{w}_{1}\in\mathbb{F}_{q}^{\rho_{1}}$, $\mathbf{w}_{2}\in\mathbb{F}_{q}^{\rho_{2}}$,$\dots,\mathbf{w}_{n_{\text{v}}+1}\in\mathbb{F}_{q}^{\rho_{n_{\text{v}}+1}}$,
such that $\mathbf{wT}\mathbf{Y}=\mathbf{x}_{k}$ for some $k\in\left\{ 1,\dots,g\right\} $.
This implies that the structure of the decoded vector $\mathbf{wTY}$
has to match the format introduced in \eqref{eq:PacketFormat} and
imposes some constraints on $\mathbf{w},$ which components have to
satisfy 
\begin{align}
\mathbf{w}_{1}\mathbf{B}_{11} & =\mathbf{e}_{j_{1}}\label{eq:Constr1}\\
\mathbf{w}_{1}\mathbf{B}_{12}+\mathbf{w}_{2}\mathbf{B}_{22} & =\mathbf{e}_{j_{2}}\label{eq:Constr2}\\
\vdots\nonumber \\
\mathbf{w}_{1}\mathbf{B}_{1,\ell-1}+\mathbf{w}_{2}\mathbf{B}_{2,\ell-1}+\dots+\mathbf{w}_{\ell-1}\mathbf{B}_{\ell-1,\ell-1} & =\mathbf{e}_{j_{\ell-1}}\label{eq:Constrl-1}\\
\mathbf{w}_{1}\mathbf{B}_{1\ell}+\mathbf{w}_{2}\mathbf{B}_{2\ell}+\dots+\mathbf{w}_{\ell}\mathbf{B}_{\ell\ell} & =\mathbf{e}_{j_{\ell}}\label{eq:Constrl1}\\
\mathbf{w}_{1}\mathbf{B}_{1,\ell+1}+\mathbf{w}_{2}\mathbf{B}_{2,\ell+1}+\dots+\mathbf{w}_{\ell+1}\mathbf{B}_{\ell+1,\ell+1} & =\mathbf{e}_{j_{\ell+1}}\label{eq:Constrl11}\\
\vdots\nonumber \\
\mathbf{w}_{1}\mathbf{B}_{1n_{\text{v}}}+\mathbf{w}_{2}\mathbf{B}_{2n_{\text{v}}}+\dots+\mathbf{w}_{n_{\text{v}}}\mathbf{B}_{n_{\text{v}}n_{\text{v}}} & =\mathbf{e}_{j_{n_{\text{v}}}}\label{eq:ConstrL}\\
c\left(\mathbf{w}_{1}\mathbf{C}_{1}+\mathbf{w}_{2}\mathbf{C}_{2}+\dots+\mathbf{w}_{n_{\text{v}}+1}\mathbf{C}_{n_{\text{v}}+1}\right) & =0,\label{eq:ConstrLast}
\end{align}
where $\mathbf{e}_{j_{\ell}}$ is the $j_{\ell}$-th canonical vector
of $\mathbb{F}_{q}^{L_{\ell}}$ and $c$ is a hash-consistency verification
function such that 
\begin{equation}
c\left(\boldsymbol{\pi},\mathbf{h}\right)=\begin{cases}
0 & \text{if }\mathbf{h}=h\left(\boldsymbol{\pi}\right),\\
1 & \text{else,}
\end{cases}\label{eq:HashVerification}
\end{equation}
where both components $\boldsymbol{\pi}$ and $\mathbf{h}$ are extracted
from $\mathbf{w}_{1}\mathbf{C}_{1}+\mathbf{w}_{2}\mathbf{C}_{2}+\dots+\mathbf{w}_{n_{\text{v}}+1}\mathbf{C}_{n_{\text{v}}+1}$.

The constraint \eqref{eq:Constrl1} can be rewritten as 
\begin{equation}
\mathbf{w}_{\ell}\mathbf{B}_{\ell\ell}=\mathbf{e}_{j_{\ell}}-\left(\mathbf{w}_{1}\mathbf{B}_{1,\ell}+\dots+\mathbf{w}_{\ell-1}\mathbf{B}_{\ell-1,\ell}\right).\label{eq:ConstrlPivot-1}
\end{equation}
This is a system of linear equations in $\mathbf{w}_{\ell}$. Since
the $\rho_{\ell}\times L_{\ell}$ matrix $\mathbf{B}_{\ell,\ell}$
has full row rank $\rho_{\ell}$, for every $\mathbf{e}_{j_{\ell}}\in\mathbb{F}_{q}^{L_{\ell}}$
there is at most one solution for $\mathbf{w}_{\ell}$. This property
allows building all candidate decoding vectors $\mathbf{w}$ using
a branch-and-prune approach described in the following algorithm,
which takes $\mathbf{T}\mathbf{Y}$ as input.

\textbf{Algorithm 1}. DeRPIA (Decoding from Random Packet Index Assignment) 
\begin{itemize}
\item Initialization: Initialize the root of the decoding tree with an empty
unmixing vector $\mathbf{w}$. 
\item Level 1: From the root node, find branches corresponding to all possible
values of $\mathbf{e}_{j_{1}}$, $j_{1}=1,\dots,L_{1}$, for which
there exists a value of $\mathbf{w}_{1}$ satisfying \eqref{eq:Constr1}. 
\item Level 2: 
\begin{itemize}
\item Expand each branch at Level~$1$ with branches corresponding to all
possible values of $\mathbf{e}_{j_{2}}$, $j_{2}=1,\dots,L_{2}$,
for which there exists a value of $\mathbf{w}_{2}$ satisfying \eqref{eq:Constr2}. 
\item Prune all branches corresponding to a pair $\left(\mathbf{w}_{1},\mathbf{w}_{2}\right)$
for which there is no $j_{2}=1,\dots,L_{2}$ such that \eqref{eq:Constr2}
is satisfied. 
\end{itemize}
\item Level $\ell$: Expand all remaining branches at Level~$\ell-1$ in
the same way. 
\begin{itemize}
\item Expand each branch at Level~$\ell-1$ for a given tuple $\left(\mathbf{w}_{1},\dots,\mathbf{w}_{\ell-1}\right)$
with branches corresponding to all possible values of $\mathbf{e}_{j_{\ell}}$,
$j_{\ell}=1,\dots,L_{\ell}$, for which there exists a value of $\mathbf{w}_{\ell}$
satisfying \eqref{eq:Constrl1}. 
\item Prune all branches corresponding to tuples $\left(\mathbf{w}_{1},\dots,\mathbf{w}_{\ell-1}\right)$
for which there is no $j_{\ell}=1,\dots,L_{\ell}$ such that \eqref{eq:Constrl1}
is satisfied. 
\end{itemize}
\item Level $n_{\text{v}}+1$:
\begin{itemize}
\item If $\rho_{n_{\text{v}}+1}=0$, all tuples $\left(\mathbf{w}_{1},\dots,\mathbf{w}_{n_{\text{V}}}\right)$
found at Level~$n_{\text{v}}$ are unmixing vectors. 
\item If $\rho_{n_{\text{v}}+1}>0$, each branch of the tree corresponding
to a vector $\left(\mathbf{w}_{1},\dots,\mathbf{w}_{n_{\text{v}}}\right)$
satisfying all constraints \eqref{eq:Constr1}-\eqref{eq:ConstrL},
is expanded with all values of $\mathbf{w}_{n_{\text{v}}+1}\in\mathbb{F}_{q}^{\rho_{n_{\text{v}}+1}}$
such that \eqref{eq:ConstrLast} is satisfied. Note that \eqref{eq:ConstrLast}
is not a linear equation.
\end{itemize}
\end{itemize}
\begin{figure}
\centering \includegraphics[width=1\textwidth]{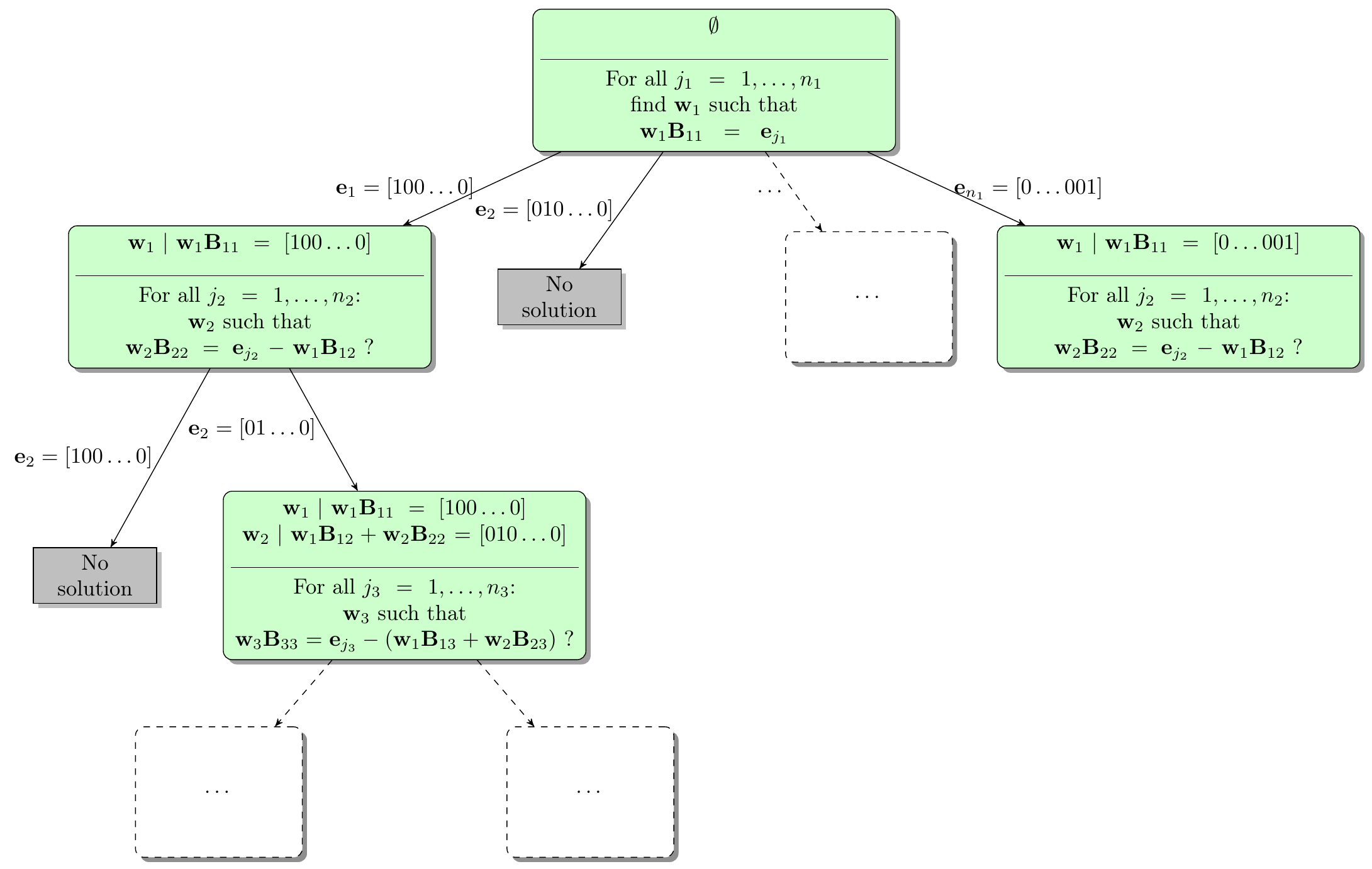} \caption{\label{fig:TreeRoot}First steps of DERPIA starting from the root
of the decoding tree.}
\end{figure}

The first steps of DeRPIA are illustrated in Figure~\ref{fig:TreeRoot}.
From the root node, several hypotheses are considered for $\mathbf{w}_{1}$.
Only those satisfying \eqref{eq:Constr1} are kept at Level~1. The
nodes at Level~1 are then expanded with candidates for $\mathbf{w}_{2}$.
Figure~\ref{fig:TreeBottom} illustrates the behavior of DeRPIA at
Level~$n_{\text{v}}+1$. Several hypotheses for $\mathbf{w}_{n_{\text{v}}+1}$
are considered. Only those such that \eqref{eq:ConstrLast} is satisfied
are kept to form the final unmixing vectors $\mathbf{w}.$

\begin{figure}
\centering \includegraphics[width=1\textwidth]{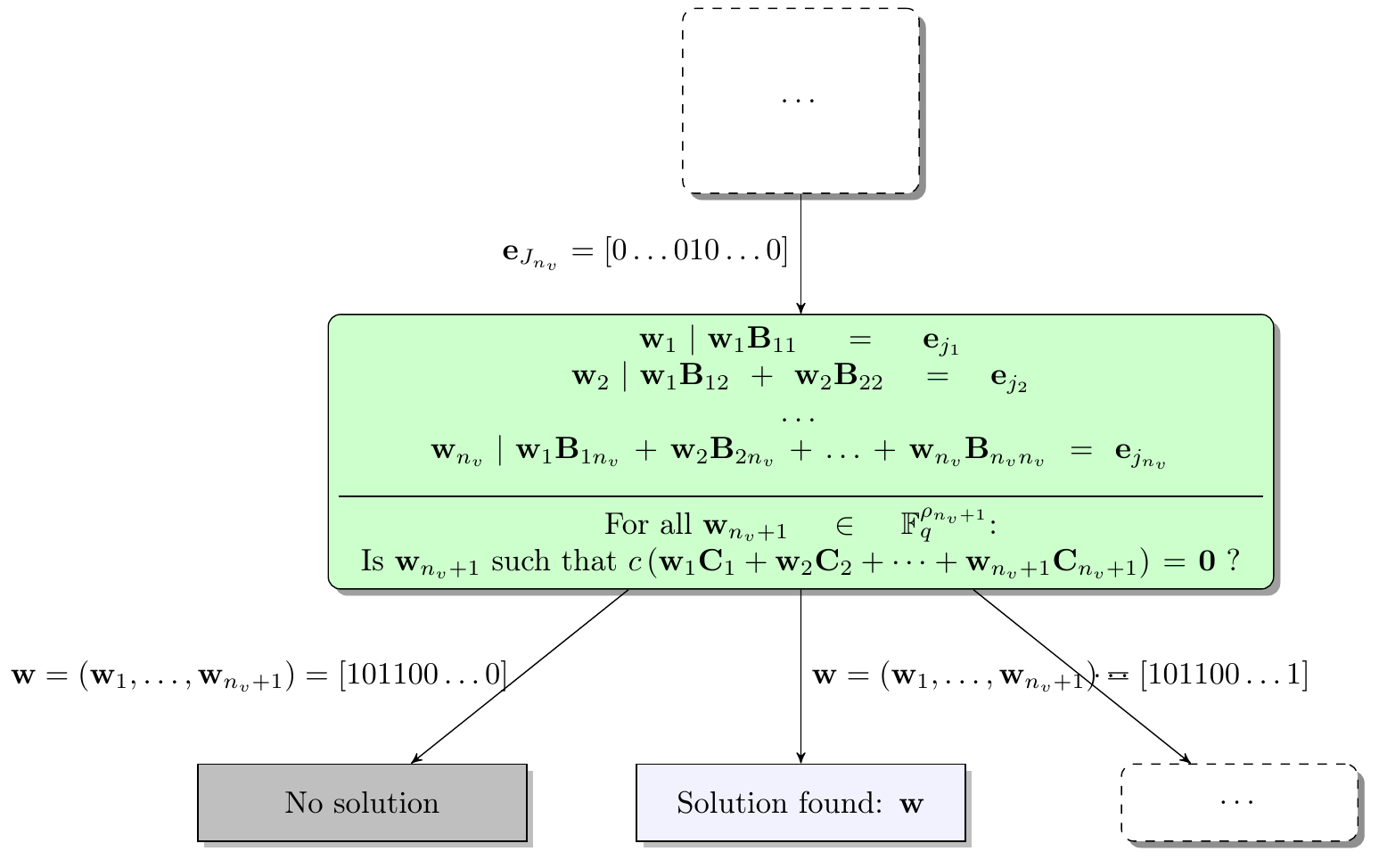}
\caption{\label{fig:TreeBottom}Last steps of DERPIA at the leaves at Level~$n_{\text{v}}+1$
of the decoding tree.}
\end{figure}

\subsubsection{Complexity reduction}

\label{sub:ComplexityRed}

The aim of this section is to show that at each level of the tree
built by DeRPIA, the solution of \eqref{eq:ConstrlPivot-1} does not
need solving a system of linear equations but may be performed by
the search in a look-up table. From Theorem~\ref{thm:OneElementOnly},
one sees that $\mathbf{w}_{\ell}$ can take at most $\rho_{\ell}+1$
different values, which are the null vector and the canonical base
vectors $\mathbf{e}_{i}$ of $\mathbb{F}^{\rho_{\ell}}$ corresponding
to the pivots of $\mathbf{B}_{\ell\ell}$. 
\begin{thm}
\emph{\label{thm:OneElementOnly}Each $\mathbf{w}_{\ell}$ satisfying
\eqref{eq:Constr1}-\eqref{eq:Constrl1} contains at most one non-zero
component. In addition, $\mathbf{w}_{1}$ contains exactly one non-zero
component.} 
\end{thm}
\begin{proof} The proof is by contradiction. Consider first $\ell=1$.
$\mathbf{B}_{11}$ is a pivot matrix of $\rho{}_{1}$ lines. If $\mathbf{w}_{1}$
contains more than one non-zero component, then $\mathbf{w}_{1}\mathbf{B}_{11}$
will contain more than one non-zero entries corresponding to the pivots
of $\mathbf{B}_{11}$ associated to the non-zero components of $\mathbf{w}_{1}$
and \eqref{eq:Constr1} cannot be satisfied. Moreover, with $\mathbf{w}_{1}=\mathbf{0}$,
\eqref{eq:Constr1} cannot be satisfied too.

Consider now some $\ell>1$. In \eqref{eq:ConstrlPivot-1}, since
$\mathbf{TY}$ is in RREF, $\mathbf{B}_{\ell\ell}$ is a pivot matrix
of $\rho{}_{\ell}$ lines. Moreover all columns of the matrices $\mathbf{B}_{1,\ell},\dots,\mathbf{B}_{\ell-1,\ell}$
which correspond to the columns of the pivots of $\mathbf{B}_{\ell\ell}$
are zero. This property is shared by the linear combination $\mathbf{w}_{1}\mathbf{B}_{1,\ell}+\dots+\mathbf{w}_{\ell-1}\mathbf{B}_{\ell-1,\ell}$.
Thus $\mathbf{e}_{j_{\ell}}-\left(\mathbf{w}_{1}\mathbf{B}_{1,\ell}+\dots+\mathbf{w}_{\ell-1}\mathbf{B}_{\ell-1,\ell}\right)$
is either the null vector, in which case $\mathbf{w}_{\ell}=\mathbf{0}$,
or contains at most one non-zero entry due to $\mathbf{e}_{j_{\ell}}$
at the columns corresponding to the pivots of $\mathbf{B}_{\ell\ell}$.
In the latter case, if $\mathbf{w}_{\ell}$ contains more than one
non-zero component, then $\mathbf{w}_{\ell}\mathbf{B}_{\ell\ell}$
will contain more than one non-zero entry corresponding to the columns
of the pivots of $\mathbf{B}_{\ell\ell}$ associated to the non-zero
components of $\mathbf{w}_{\ell}$ and \eqref{eq:Constrl1} cannot
be satisfied. \end{proof}

Note that for a given branch, when a set of vectors $\mathbf{w}_{1},\dots,\mathbf{w}_{\ell-1}$
has been found, a vector \emph{$\mathbf{w}_{\ell}$ }satisfying\emph{
}\eqref{eq:Constr1}-\eqref{eq:ConstrLast}\emph{ }does not necessarily
exist. In such case the corresponding branch is pruned, see the last
case of Example~\ref{exa:Decod} in what follows.

Using Theorem~\ref{thm:OneElementOnly}, there is no linear system
of equations to be solved any more. In practice, the search for $\mathbf{w}_{\ell}$
can be even further simplified using the following corollary. 
\begin{cor}
The search for $\mathbf{w}_{\ell}$ satisfying \eqref{eq:Constrl1}
reduces to a simple look-up in a table. 
\end{cor}
\begin{proof} Assume that there exists \emph{$\mathbf{w}_{\ell}$
}satisfying \eqref{eq:Constr1}-\eqref{eq:ConstrLast}. Then \eqref{eq:ConstrlPivot-1}
can be rewritten as: $\mathbf{w}_{\ell}\mathbf{B}_{\ell\ell}-\mathbf{e}_{j_{\ell}}=-\left(\mathbf{w}_{1}\mathbf{B}_{1,\ell}+\dots+\mathbf{w}_{\ell-1}\mathbf{B}_{\ell-1,\ell}\right)$.
Here we assume that $\mathbf{w}_{\ell}\neq\mathbf{0}$ and then further
analyze properties established in Theorem~\ref{thm:OneElementOnly}: 
\begin{itemize}
\item In the linear combination $-\left(\mathbf{w}_{1}\mathbf{B}_{1,\ell}+\dots+\mathbf{w}_{\ell-1}\mathbf{B}_{\ell-1,\ell}\right)$,
all components that correspond to the columns of the pivots of $\mathbf{B}_{\ell\ell}$
are zero. 
\item Both sides of \eqref{eq:ConstrlPivot-1} can have at most one non-zero
entry in the columns of the pivots of $\mathbf{B}_{\ell\ell}$. If
there is one, that non-zero entry must correspond to $\mathbf{e}_{j_{\ell}}$.
\end{itemize}
It follows that $\mathbf{e}_{j_{\ell}}$ must exactly correspond to
the component at the column of the unique pivot of $\mathbf{B}_{\ell\ell}$
found in the vector $\mathbf{w}_{\ell}\mathbf{B}_{\ell\ell}$ and
must cancel it in the expression $\mathbf{w}_{\ell}\mathbf{B}_{\ell\ell}-\mathbf{e}_{j_{\ell}}$.
Since $\mathbf{w}_{\ell}$ (assumed non-zero) has only one non-zero
entry and since coefficients corresponding to pivots are equal to
$1$, the non-zero component of $\mathbf{w}_{\ell}$ must be $1$
(as it is the case for $\mathbf{e}_{j_{\ell}}$). As a result $\mathbf{w}_{\ell}\mathbf{B}_{\ell\ell}$
is actually one of the row vectors of $\mathbf{B}_{\ell\ell}$, and
the expression $\mathbf{w}_{\ell}\mathbf{B}_{\ell\ell}-\mathbf{e}_{j_{\ell}}$
is that row vector with a zero at the place of the component of the
associated pivot. Finding one non-zero $\mathbf{w}_{\ell}$ satisfying
\eqref{eq:Constr1}-\eqref{eq:Constrl1} is then equivalent to identifying
all row vectors $\mathbf{u}$ of $\mathbf{B}_{\ell\ell}$ such that
$\overline{\mathbf{u}}=$$-\left(\mathbf{w}_{1}\mathbf{B}_{1,\ell}+\dots+\mathbf{w}_{\ell-1}\mathbf{B}_{\ell-1,\ell}\right)$,
where $\overline{\mathbf{u}}=\mathbf{u}-\mathbf{e}_{\gamma\left(\mathbf{u}\right)}$
and $\gamma\left(\mathbf{u}\right)$ is the index of the pivot column
of $\mathbf{u}$. \end{proof}

One deduces the following look-up table-based algorithm, which consists
in two parts. Algorithm~2a is run once and builds a set of look-up
tables from $\mathbf{B}_{1,1},\dots,\mathbf{B}_{n_{\text{v}},n_{\text{v}}}$.
Algorithm~2b uses then these look-up tables, takes as input $\left(\mathbf{w}_{1},\dots,\mathbf{w}_{\ell-1}\right)$
satisfying \eqref{eq:Constrl-1} and provides the set of vectors $\left(\mathbf{w}_{1},\dots,\mathbf{w}_{\ell-1},\mathbf{w}_{\ell}\right)$
satisfying \eqref{eq:Constrl1}.

\textbf{Algorithm 2a}. Construction of the look-up tables from $\mathbf{B}_{1,1},\dots,\mathbf{B}_{n_{\text{v}},n_{\text{v}}}$ 
\begin{enumerate}
\item For $\ell=1,\dots,n_{\text{v}}$
\begin{enumerate}
\item Initialization: $\Pi_{\ell}=\emptyset$. 
\item For each row vectors $\mathbf{u}$ of $\mathbf{B}_{\ell\ell}$,
\begin{enumerate}
\item identify the index $\gamma\left(\mathbf{u}\right)$ of its pivot column,
and denote by $\overline{\mathbf{u}}$ the row vector with a zero
at the place of its pivot: $\overline{\mathbf{u}}=\mathbf{u}-\mathbf{e}_{\gamma\left(\mathbf{u}\right)}$ 
\item if $\overline{\mathbf{u}}\notin\Pi_{\ell}$, then $\Pi_{\ell}=\Pi_{\ell}\cup\left\{ \overline{\mathbf{u}}\right\} $. 
\end{enumerate}
\item For each $\mathbf{v}\in\Pi_{\ell}$, evaluate 
\[
\mathcal{C_{\ell}}\left(\mathbf{v}\right)=\{\mathbf{e}_{i}\in\mathbb{F}_{q}^{\rho_{\ell}},i=1,\dots,\rho_{\ell}|\mathbf{u}\mathrm{\mathbf{\mathbf{=}}\mbox{\ensuremath{\mathbf{e}_{i}}}\mathbf{B}_{\ell\ell}}\mathrm{~and~}\mathbf{u}-\mathbf{e}_{\gamma\left(\mathbf{u}\right)}=\mathbf{v}\}.
\]
\end{enumerate}
\end{enumerate}
The sets $\mathcal{C_{\ell}}\left(\mathbf{v}\right)$ contain the
candidate $\mathbf{w}_{\ell}$ that may satisfy \eqref{eq:Constrl1}. 
\begin{example}
\label{exa:LUT}Assume for example that 
\[
\mathbf{B}_{\ell,\ell}=\left(\begin{array}{cccccc}
1 & 0 & 0 & 0 & 1 & 0\\
0 & 0 & 1 & 0 & 1 & 0\\
0 & 0 & 0 & 1 & 1 & 0\\
0 & 0 & 0 & 0 & 0 & 1
\end{array}\right).
\]
Using Algorithm~2a, one obtains 
\[
\Pi_{\ell}=\left\{ \left(\begin{array}{cccccc}
0 & 0 & 0 & 0 & 1 & 0\end{array}\right),\left(\begin{array}{cccccc}
0 & 0 & 0 & 0 & 0 & 0\end{array}\right)\right\} 
\]
and 
\begin{align*}
\mathcal{C_{\ell}}\left(\left(\begin{array}{cccccc}
0 & 0 & 0 & 0 & 1 & 0\end{array}\right)\right) & =\left\{ \left(\begin{array}{cccc}
1 & 0 & 0 & 0\end{array}\right),\left(\begin{array}{cccc}
0 & 1 & 0 & 0\end{array}\right),\left(\begin{array}{cccc}
0 & 0 & 1 & 0\end{array}\right)\right\} ,\\
\mathcal{C_{\ell}}\left(\left(\begin{array}{cccccc}
0 & 0 & 0 & 0 & 0 & 0\end{array}\right)\right) & =\left\{ \left(\begin{array}{cccc}
0 & 0 & 0 & 1\end{array}\right)\right\} .
\end{align*}
\end{example}
\textbf{Algorithm 2b}. Obtain the set $\mathcal{W}_{\ell}$ of all
$\mathbf{w}_{\ell}$ satisfying \eqref{eq:Constrl1}, from $\left(\mathbf{w}_{1},\dots,\mathbf{w}_{\ell-1}\right)$
satisfying \eqref{eq:Constrl-1}. 
\begin{enumerate}
\item Initialization: $\mathcal{W}_{\ell}=\emptyset$.
\item Compute $\mathbf{v}=-\left(\mathbf{w}_{1}\mathbf{B}_{1,\ell}+\dots+\mathbf{w}_{\ell-1}\mathbf{B}_{\ell-1,\ell}\right)$.
\item If $\mathbf{w}_{\ell}=\mathbf{0}$ satisfies \eqref{eq:Constrl1},
\emph{i.e.}, if $\mathbf{v}+\mathbf{e}_{j_{\ell}}=\mathbf{0}$ for
some canonical vector $\mathbf{e}_{j_{\ell}}\in\mathbb{F}_{q}^{L_{\ell}}$,
then $\mathcal{W}_{\ell}=\left\{ \mathbf{0}\right\} $. 
\item If\textbf{ }$\mathbf{v}\in\Pi_{\ell}$, then $\mathcal{W}_{\ell}=\mathcal{W}_{\ell}\cup\mathcal{C_{\ell}}(\mathbf{v})$.
\end{enumerate}
\begin{example}
\label{exa:Decod}Consider the results of Example~\ref{exa:LUT}
and a branch at level $\ell-1$ of the decoding tree associated to
$\left(\mathbf{w}_{1},\dots,\mathbf{w}_{\ell-1}\right)$. One searches
the set of $\mathbf{w}_{\ell}$s satisfying 
\begin{align}
\mathbf{w}_{\ell}\mathbf{B}_{\ell\ell} & =\mathbf{e}_{j_{\ell}}-\left(\mathbf{w}_{1}\mathbf{B}_{1,\ell}+\dots+\mathbf{w}_{\ell-1}\mathbf{B}_{\ell-1,\ell}\right).\label{eq:Ex5Eq}
\end{align}
Assume first that $\left(\mathbf{w}_{1},\dots,\mathbf{w}_{\ell-1}\right)$
is such that 
\begin{align*}
\mbox{\ensuremath{\mathbf{v}_{1}\mathbf{=}}} & -\left(\mathbf{w}_{1}\mathbf{B}_{1,\ell}+\dots+\mathbf{w}_{\ell-1}\mathbf{B}_{\ell-1,\ell}\right)\\
= & \left(\begin{array}{cccccc}
0 & 0 & 0 & 0 & 1 & 0\end{array}\right)
\end{align*}
then $\mathbf{w=0}$ is a solution of \eqref{eq:Ex5Eq} associated
to $\mathbf{e}_{j_{\ell}}=\mathbf{e}_{5}=\left(\begin{array}{cccccc}
0 & 0 & 0 & 0 & 1 & 0\end{array}\right)$. The other solutions are given by $\mathcal{C_{\ell}}\left(\mathbf{v}_{1}\right)=\left\{ \left(\begin{array}{cccc}
1 & 0 & 0 & 0\end{array}\right),\left(\begin{array}{cccc}
0 & 1 & 0 & 0\end{array}\right),\left(\begin{array}{cccc}
0 & 0 & 1 & 0\end{array}\right)\right\} $ and $\mathcal{W_{\ell}}=\left\{ \left(\begin{array}{cccc}
0 & 0 & 0 & 0\end{array}\right),\left(\begin{array}{cccc}
1 & 0 & 0 & 0\end{array}\right),\left(\begin{array}{cccc}
0 & 1 & 0 & 0\end{array}\right),\left(\begin{array}{cccc}
0 & 0 & 1 & 0\end{array}\right)\right\} $.

Assume now that $\left(\mathbf{w}_{1},\dots,\mathbf{w}_{\ell-1}\right)$
is such that 
\begin{align*}
\mbox{\ensuremath{\mathbf{v}_{2}\mathbf{=}}} & -\left(\mathbf{w}_{1}\mathbf{B}_{1,\ell}+\dots+\mathbf{w}_{\ell-1}\mathbf{B}_{\ell-1,\ell}\right)\\
= & \left(\begin{array}{cccccc}
0 & 1 & 0 & 0 & 0 & 0\end{array}\right),
\end{align*}
then $\mathbf{w=0}$ is a solution of \eqref{eq:Ex5Eq} associated
to $\mathbf{e}_{j_{\ell}}=\mathbf{e}_{2}=\left(\begin{array}{cccccc}
0 & 1 & 0 & 0 & 0 & 0\end{array}\right)$. There is no other solution, since $\mathbf{v}_{2}\notin\Pi_{\ell}$.

Assume finally that $\left(\mathbf{w}_{1},\dots,\mathbf{w}_{\ell-1}\right)$
is such that 
\begin{align*}
\mbox{\ensuremath{\mathbf{v}_{3}\mathbf{=}}} & -\left(\mathbf{w}_{1}\mathbf{B}_{1,\ell}+\dots+\mathbf{w}_{\ell-1}\mathbf{B}_{\ell-1,\ell}\right)\\
= & \left(\begin{array}{cccccc}
1 & 0 & 0 & 1 & 0 & 0\end{array}\right),
\end{align*}
then $\mathcal{W}_{\ell}=\emptyset$, since $\mathbf{w=0}$ is not
a solution of \eqref{eq:Ex5Eq} and $\mathbf{v}_{3}\notin\Pi_{\ell}$. 
\end{example}

\section{Complexity evaluation}

\label{sec:Complexity}

To evaluate the arithmetic complexity of the NeCoRPIA decoding algorithms,
one assumes that the complexity of the product of an $n\times m$
matrix and an $m\times p$ matrix is at most $K_{\text{m}}nmp$ operations
where $K_{\text{m}}$ is some constant. The complexity of the evaluation
of the checksum of a vector in $\mathbb{F}_{q}^{n}$ is $K_{c}n$,
where $K_{\text{c}}$ is also a constant. Finally determining whether
two vectors of $n$ entries in $\mathbb{F}_{q}$ are equal requires
at most $n$ operations.

The decoding complexity of NeCoRPIA depends on several parameters.
First, the $g'$ received packets of length $L_{\text{x}}$ collected
in $\mathbf{Y}'$ have to be put in RREF to get $\mathbf{TY}$, introduced
in \eqref{eq:Echelon}. The arithmetic complexity of this operation
is $K_{\text{R}}\left(g'\right)^{2}L_{\text{x}}$, with $K_{\text{R}}$
a constant.

Algorithm~1 is a tree traversal algorithm. Each branch at Level~$\ell$
of the tree corresponds to a partial decoding vector $\left(\mathbf{w}_{1},\dots,\mathbf{w}_{\ell}\right)$
satisfying\emph{ }\eqref{eq:Constr1}-\eqref{eq:Constrl1}. The complexity
depends on the number of subbranches stemming from this branch and
the cost related to the partial verification of \eqref{eq:Constr1}-\eqref{eq:Constrl11}
for each of these subbranches. At Level $n_{\text{v}}+1$, for each
branch corresponding to a partial decoding vector $\left(\mathbf{w}_{1},\dots,\mathbf{w}_{n_{\text{v}}}\right)$
satisfying\emph{ }\eqref{eq:Constr1}-\eqref{eq:ConstrL}, an exhaustive
search has to be performed for $\mathbf{w}_{n_{\text{v}}+1}$ such
that $\left(\mathbf{w}_{1},\dots,\mathbf{w}_{n_{\text{v}}+1}\right)$
satisfies \eqref{eq:Constr1}-\eqref{eq:ConstrLast}\emph{.}

Section~\ref{sub:NbBranchTree} evaluates an upper bound for the
number of branches at each level of the decoding tree. Section~\ref{sub:CplAlg1}
determines the complexity of Algorithm~1 when the satisfying $\mathbf{w}_{\ell}$s
are obtained by the solution of a system of linear equations. This
version of Algorithm~1 is called DeRPIA-SLE in what follows. Section~\ref{sub:Alg2}
describes the complexity of Algorithm~1 when the satisfying $\mathbf{w}_{\ell}$s
are obtained from look-up tables as described in Algorithms~2a and
2b. This version of Algorithm~1 is called DeRPIA-LUT in what follows.
As will be seen, the complexities depend on the ranks $\rho_{1},\dots,\rho_{n_{\text{v}}+1}$,
which distribution is evaluated in Section~\ref{sub:Rankdist} in
the case $n_{\text{v}}=1$ and $n_{\text{v}}=2$.

\subsection{Number of branches in the tree}

\label{sub:NbBranchTree}

The following corollary of Theorem~\ref{thm:OneElementOnly} provides
an evaluation of the number of branches that needs to be explored
in Algorithm~1. 
\begin{cor}
\label{cor:Complexity}At Level~$\ell$, with $1\leqslant\ell\leqslant n_{\text{v}}$,
the maximum number of vectors $\left(\mathbf{w}_{1},\dots,\mathbf{w}_{\ell}\right)$
satisfying \eqref{eq:Constrl1} is 
\begin{equation}
N_{\text{b}}(\ell)=\rho_{1}\left(\rho_{2}+1\right) ... \left(\rho_\ell+1\right).\label{eq:Nnl}
\end{equation}
The maximum number of branches to be considered by Algorithm~1
at Level~$\ell=n_{\text{v}}+1$ is 
\begin{equation}
N_{\text{b}}\left(n_{\text{v}}+1\right)=\rho_{1}\left(\rho_{2}+1\right) ... \left(\rho_{n_{\text{v}}}+1\right)q^{\rho_{n_{\text{v}}+1}}\label{eq:Nb}
\end{equation}
and the total number of branches in the decoding tree is upper bounded
by 
\begin{equation}
N_{\text{b}}=\rho_{1}+\rho_{1}\left(\rho_{2}+1\right)+ ... +\rho_{1}\left(\rho_{2}+1\right) ... \left(\rho_{n_{\text{v}}}+1\right)+\rho_{1}\left(\rho_{2}+1\right) ... \left(\rho_{n_{\text{v}}}+1\right)q^{\rho_{n_{\text{v}}+1}}.\label{eq:Nn}
\end{equation}
\end{cor}
\begin{proof} From Theorem~\ref{thm:OneElementOnly}, one deduces
that $\mathbf{w}_{1}$, of size $\rho_{1}$ can take at most $\rho_{1}$
different values. For $1<\ell\leqslant n_{\text{v}}$, $\mathbf{w}_{\ell}$
of size $\rho_{\ell}$ can either be the null vector or take $\rho_{\ell}$
different non-zero values. For the vector $\mathbf{w}_{n_{\text{v}}+1}$,
all possible values $\mathbf{w}_{n_{\text{v}}+1}\in\mathbb{F}_{q}^{\rho_{n_{\text{v}}+1}}$
have to be considered to check whether \eqref{eq:ConstrLast} is verified.
The number of vectors\emph{ $\left(\mathbf{w}_{1},\dots,\mathbf{w}_{\ell}\right)$
}satisfying\emph{ }\eqref{eq:Constrl1} at level $\ell$ is thus upper
bounded by the product of the number of possible values of $\mathbf{w}_{k}$,
$k=1,\dots,\ell$ which is \eqref{eq:Nnl}. Similarly, the number
of branches that have to be considered at Level~$n_{\text{v}}+1$
of the search tree of Algorithm~1 is the product of the number of
all possible values of the $\mathbf{w}_{\ell}$, $\ell=1,\dots,n_{\text{v}}+1$
and thus upper-bounded by \eqref{eq:Nb}. An upper bound of the total
number of branches to consider is then 
\[
N_{\text{b}}=\sum_{\ell=1}^{n_{\text{v}}+1}N_{\text{b}}(\ell),
\]
which is given by \eqref{eq:Nn}. \end{proof}

\subsection{Arithmetic complexity of DeRPIA-SLE}

\label{sub:CplAlg1}

An upper bound of the arithmetic complexity of the tree traversal
algorithm, when the number of encoding vectors is $n_{\text{v}}$,
is provided by Proposition~\ref{prop:ArithComplexitySLE} 
\begin{prop}
\emph{\label{prop:ArithComplexitySLE}Assume that NeCoRPIA has been
performed with $n_{\text{v}}$ random coding subvectors. Then an upper
bound of the total arithmetic complexity of DeRPIA-SLE is 
\begin{align}
K_{\text{SLE}}\left(n_{\text{v}}\right) & =\sum_{\ell=1}^{n_{\text{v}}+1}N_{\text{b}}\left(\ell-1\right)K\left(\ell\right)\label{eq:KSLE}
\end{align}
with 
\begin{align}
K\left(1\right) & =L_{1}\rho_{1}L_{1},\nonumber \\
K\left(\ell\right) & =K_{\text{m}}L_{\ell}\left(\rho_{1}+...+\rho_{\ell-1}\right)+L_{\ell}\left(\ell+\left(\rho_{\ell}+1\right)L_{\ell}\right),\nonumber \\
K\left(n_{\text{v}}+1\right) & =K_{\text{m}}gL_{\text{p}}+\left(n_{\text{v}}-1\right)L_{\text{p}}+q^{\rho_{n_{\text{v}}+1}}\left(K_{\text{m}}\rho_{n_{\text{v}}+1}L_{\text{p}}+L_{\text{p}}+K_{\text{c}}L_{\text{p}}\right).\label{eq:Knv1}
\end{align}
and $N_{\text{b}}\left(0\right)=1$. In the case $n_{\text{v}}=1$,
\eqref{eq:KSLE} boils down to 
\begin{equation}
K_{\text{SLE}}\left(1\right)=\rho_{1}L_{1}^{2}+\rho_{1}L_{\text{p}}\left(K_{\text{m}}g+q^{\rho_{2}}\left(K_{\text{m}}\rho_{2}+1+K_{\text{c}}\right)\right).\label{eq:Kt1}
\end{equation}
}
\end{prop}
The proof of Proposition~\ref{prop:ArithComplexitySLE} is given
in Appendix~\ref{sub:ComplexityDERPIASLE}.

One sees that \eqref{eq:KSLE} is expressed in terms of $\rho_{1},\dots,\rho_{n_{\text{v}}+1}$,
which are the ranks of the matrices $\mathbf{B}_{\ell\ell}$. As expected,
the complexity is exponential in $\rho_{n_{\text{v}}+1}$, which has
to be made as small as possible. The values of $\rho_{1},\dots,\rho_{n_{\text{v}}+1}$
depend on those of $L_{\ell}$ and $g$, as will be seen in Section~\ref{sub:Rankdist}.

\subsection{Arithmetic complexity of DeRPIA-LUT}

\label{sub:Alg2}

The tree obtained with DeRPIA-LUT is the same as that obtained with
DeRPIA-SLE. The main difference comes from the arithmetic complexity
when expanding one branch of the tree. Algorithm~2a is run only once.
Algorithm~2b is run for each branch expansion. An upper bound of
the total arithmetic complexity of DeRPIA-LUT is given by the following
proposition. 
\begin{prop}
\emph{\label{prop:ArithComplexityDERPIALUT}Assume that NeCoRPIA has
been performed with $n_{\text{v}}$ random coding subvectors. Then
an upper bound of the total arithmetic complexity of DeRPIA-LUT is
\begin{align}
K_{\text{LUT}}\left(n_{\text{v}}\right) & =\sum_{\ell=1}^{n_{\text{v}}}\left(K_{\text{LU},\text{1}}\left(\ell\right)+K_{\text{LU},\text{2}}\left(\ell\right)\right)+\sum_{\ell=1}^{n_{\text{v}}}N_{\text{b}}\left(\ell-1\right)K_{\text{LU},3}\left(\ell\right)\nonumber \\
 & +N_{\text{b}}\left(n_{\text{v}}\right)K\left(n_{\text{v}}+1\right)\label{eq:KLUT}
\end{align}
with 
\begin{align*}
K_{\text{LU},\text{1}}\left(\ell\right) & =\rho_{\ell}+L_{\ell}\frac{\rho_{\ell}\left(\rho_{\ell}+1\right)}{2},\\
K_{\text{LU},2}\left(\ell\right) & =\rho_{\ell}\left(L_{\ell}+1+\rho_{\ell}L_{\ell}\right),\\
K_{\text{LU},3}\left(\ell\right) & =K_{m}L_{\ell}(\rho_{1}+...+\rho_{\ell-1})+L_{\ell}\left(\ell+\rho_{\ell}\right),
\end{align*}
and $K\left(n_{\text{v}}+1\right)$ given by \eqref{eq:Knv1}. }
\end{prop}
The proof of Proposition~\ref{prop:ArithComplexityDERPIALUT} is
provided in Appendix~\ref{sec:ArithCompDERPIALUT}.

Again, as in \eqref{eq:KSLE}, the complexity \eqref{eq:KLUT} depends
on the ranks $\rho_{1},\dots,\rho_{n_{\text{v}}+1}$ of the matrices
$\mathbf{B}_{\ell\ell}$.

\subsection{Complexity comparison}

\label{sub:ComplexityComp}

When comparing $K_{\text{SLE}}\left(n_{\text{v}}\right)$ and $K_{\text{LUT}}\left(n_{\text{v}}\right)$,
one observes that there is a (small) price to be paid for building
the look-up tables corresponding to the first sum in \eqref{eq:KLUT}.
Then, the complexity gain provided by the look-up tables appears in
the expression of $K_{\text{LU},3}\left(\ell\right)$, which is linear
in $L_{\ell}$, whereas $K\left(\ell\right)$ is quadratic in $L_{\ell}$.
Nevertheless, the look-up procedure is less useful when there are
many terminal branches to consider, \emph{i.e.}, when $\rho_{n_{\text{v}}+1}$
is large, since in this case, the term $N_{\text{b}}\left(n_{\text{v}}\right)K\left(n_{\text{v}}+1\right)$
dominates in both expressions of $K_{\text{SLE}}\left(n_{\text{v}}\right)$
and $K_{\text{LUT}}\left(n_{\text{v}}\right)$.

\subsection{Distribution of the ranks $\rho_{1},\dots,\rho_{n_{\text{v}}+1}$}

\label{sub:Rankdist}

Determining the distributions of $\rho_{1},\dots,\rho_{n_{\text{v}}+1}$
in the general case is relatively complex. In what follows, one focuses
on the cases $n_{\text{v}}=1$ and $n_{\text{v}}=2$. Experimental
results for other values of $n_{\text{v}}$ are provided in Section~\ref{sec:PerfEval}.

\subsubsection{NC vector with one component, $n_{\text{v}}=1$}

\label{sub:Nvv1}

In this case, the random NC vectors are gathered in the matrix $\mathbf{V}_{1}$.
Once $\mathbf{Y}$ has been put in RREF, the rows of the matrix $\mathbf{B}_{11}$
of rank $\rho_{1}$ are the $\rho_{1}$ linearly independent vectors
of $\mathbf{V}_{1}$. The distribution of $\rho_{1}$ may be analyzed
considering the classical\emph{ }urn problem described in \cite{Holst1986}.
This problem may be formulated as: assume that $g$ indistinguishable
balls are randomly dropped in $L_{1}$ distinguishable boxes. The
distribution of the number $X_{g}^{L_{1}}$ of boxes that contain
at least one ball is described in \cite{Holst1986} citing De Moivre
as

\begin{align}
P(X_{g}^{L_{1}} & =k)=f\left(g,L_{1},k\right),\label{eq:DeMoivre}
\end{align}
with 
\[
f\left(g,L,k\right)=\frac{L(L\text{\textminus}1)\text{\ensuremath{\dots}}(L\text{\textminus}k+1)}{L^{g}}S\left(g,k\right),
\]
where $S(g,k)$ denotes the Stirling numbers of the second kind \cite{Abramowitz1972}.
In the context of NeCorPIA, a \emph{collision }corresponds to at least
two balls dropped in the same box. The probability mass function (pmf)
of the rank $\rho_{1}$ of $\mathbf{B}_{11}$ is then that of the
number of boxes containing at least one ball and is given by \eqref{eq:DeMoivre}.
The pmf of $\rho_{2}$ is deduced from \eqref{eq:DeMoivre} as 
\begin{align}
P\left(\rho_{2}=g-k\right) & =1-P(\rho_{1}=k)\nonumber \\
 & =1-f\left(g,L_{1},k\right).\label{eq:Prho2Nv1}
\end{align}

\begin{figure}
\begin{centering}
\includegraphics[width=0.5\columnwidth]{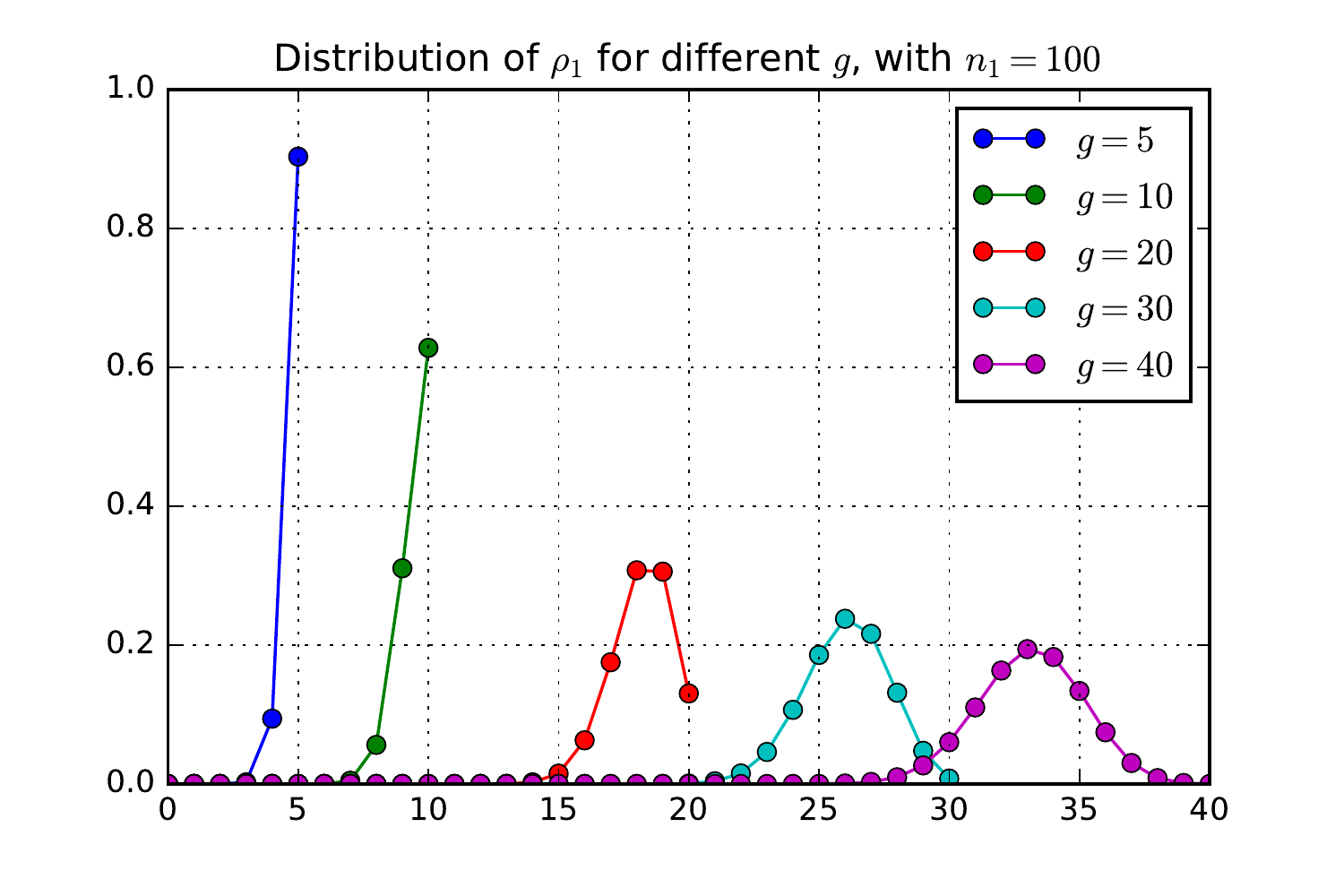}\includegraphics[width=0.5\columnwidth]{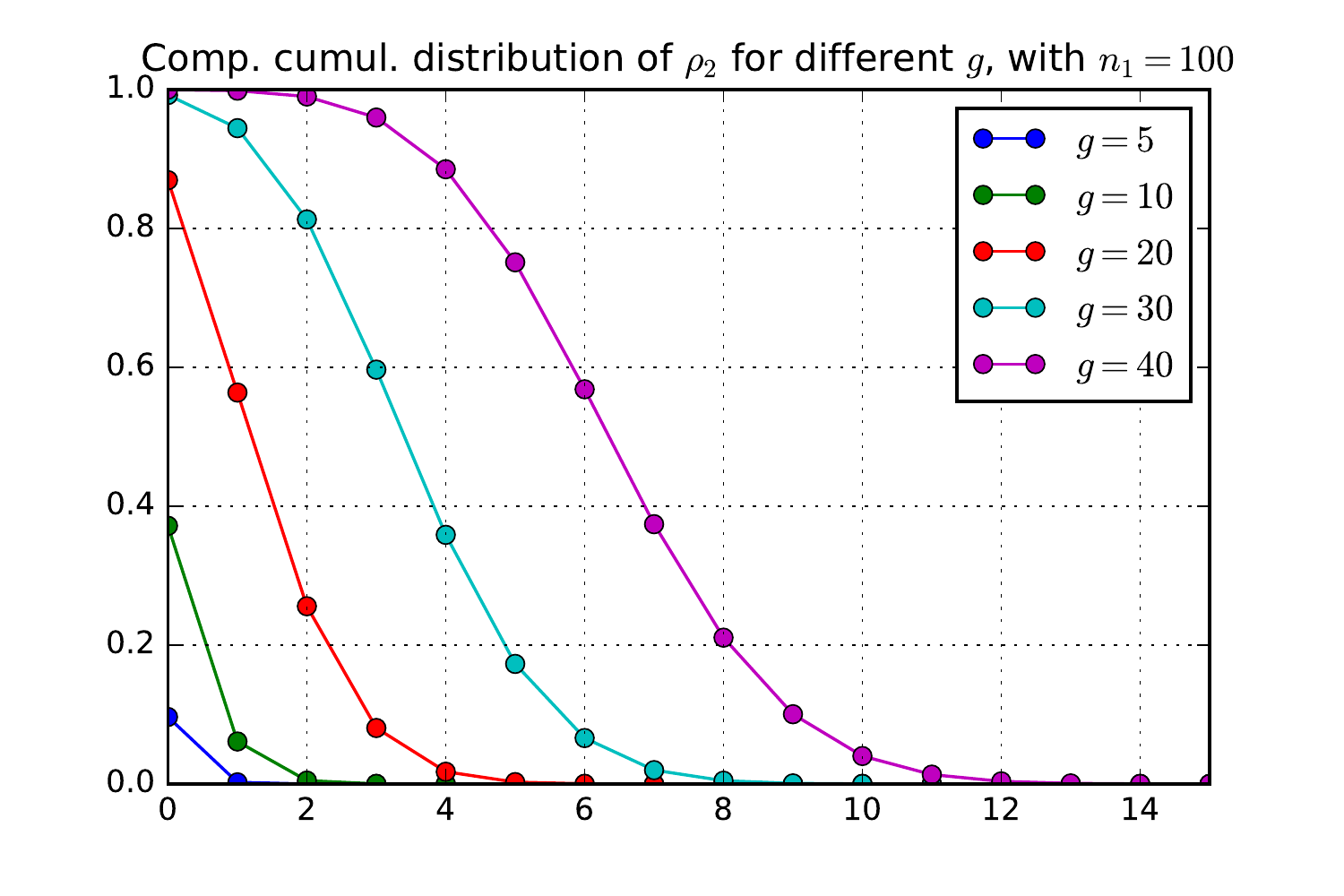} 
\par\end{centering}
\protect\caption{\label{fig:L1RhoComplexity}Case $n_{\text{v}}=1$, theoretical and
experimental distributions of $\rho_{1}$ (left) and theoretical and
practical complementary cumulative distribution function of $\rho_{2}$
(right) for different values of $g$ for $L_{1}=100$.}
\end{figure}
Figure~\ref{fig:L1RhoComplexity} shows the distribution of $\rho_{1}$
(left) and the complementary cumulative distribution function $\rho_{2}$
(right) for different values of $g$ when $L_{1}=100$. One sees that
$\Pr\left(\rho_{2}>1\right)<0.0025$ when $g=5$ and $\Pr\left(\rho_{2}>1\right)<0.062$
when $g=10$. In the decoding complexity, the exponential term of
$K\left(n_{\text{v}}+1\right)$ in both \eqref{eq:KSLE} and \eqref{eq:KLUT}
will thus be of limited impact. When $g=30$, $\Pr\left(\rho_{2}>1\right)$
is larger than 94\%. The exponential term in the complexity becomes
overwhelming in that case.

\subsubsection{NC vector with two components, $n_{\text{v}}=2$}

\label{sub:NVv2}

In this case, the random NC vectors form the matrix $\left(\mathbf{V}_{1},\mathbf{V}_{2}\right)$.
The pmf of $\rho_{1}$ is still given by \eqref{eq:DeMoivre}. Determining
the pmf of the rank $\rho_{2}$ of $\mathbf{B}_{22}$ is much more
complicated. Hence, we will first evaluate an upper bound on the probability
that $\rho_{3}=0$ and an approximation of the pmf of $\rho_{3}$.
Using these results, one will derive an estimate of the pmf of $\rho_{2}$.

In general, to have $\left(\mathbf{V}_{1},\dots,\mathbf{V}_{n_{\text{v}}}\right)$
of full rank, \emph{i.e.}, $\rho_{n_{\text{v}}+1}=0$, it is necessary
that no pair of nodes has generated packets with the same random coding
subvectors. In the corresponding urn problem, one has now to consider
$g$ balls thrown into $L_{1}L_{2}\dots L_{n_{\text{v}}}$ boxes.
In the case $n_{\text{v}}=2$, a node selects two random NC subvectors
$\mathbf{e}_{i}\in\mathbb{F}_{q}^{L_{1}}$ and $\mathbf{e}_{j}\in\mathbb{F}_{q}^{L_{2}}$.
The pair of indices $\left(i,j\right)$ may be interpreted as the
row and column index of the box in which a ball has been dropped,
when $L_{1}L_{2}$ boxes are arranged in a rectangle with $L_{1}$
rows and $L_{2}$ columns. The probability of having $g$ balls thrown
in $L_{1}L_{2}\dots L_{n_{\text{v}}}$ boxes reaching $g$ different
boxes, \emph{i.e.}, of having coding subvectors different for all
$g$ nodes is $P(X_{g}^{L_{1}\dots L_{n_{\text{v}}}}=g)$ and can
again be evaluated with \eqref{eq:DeMoivre}.

A rank deficiency happens when the previous necessary condition is
not satisfied, but it may also happen in other cases, see Example~\ref{exa:Cycle}.
As a consequence, one only gets the following upper bound 
\begin{equation}
P\left(\rho_{1}+\dots+\rho_{n_{\text{v}}}=g\right)\leqslant f\left(g,L_{1}L_{2}\dots L_{n_{\text{v}}},g\right).\label{eq:UpperBoundFullRank}
\end{equation}

If one assumes that the rank deficiency is only due to nodes having
selected the same random coding subvectors, similarly, one may apply
\eqref{eq:DeMoivre} as in the case $n_{\text{v}}=1$ and get the
following approximation

\begin{equation}
P(\rho_{1}+\dots+\rho_{n_{\text{v}}}=k)\approx f\left(g,L_{1}L_{2}\dots L_{n_{\text{v}}},k\right)\label{eq:nv2}
\end{equation}
which leads to
\begin{equation}
P(\rho_{n_{\text{v}}+1}=g-k)\approx f\left(g,L_{1}L_{2}\dots L_{n_{\text{v}}},k\right).\label{eq:approxrhonv1}
\end{equation}
This is an approximation, since even if all nodes have selected different
coding subvectors, we might have a rank deficiency, as illustrated
by Example~\ref{exa:Cycle}.

In the case $n_{\text{v}}=2$, one will now build an approximation
of 
\begin{align}
P\left(\rho_{1}=k_{1},\rho_{2}=k_{2},\rho_{3}=k_{3}\right) & =P\left(\rho_{2}=k_{2}|\rho_{1}=k_{1},\rho_{3}=k_{3}\right)\nonumber \\
 & P\left(\rho_{1}=k_{1}|\rho_{3}=k_{3}\right)P\left(\rho_{3}=k_{3}\right).\label{eq:r1r2r3}
\end{align}
Using the fact $k_{1}+k_{2}+k_{3}=g$, one has 
\[
P\left(\rho_{2}=k_{2}|\rho_{1}=k_{1},\rho_{3}=k_{3}\right)=\begin{cases}
1 & \text{if }k_{2}=g-k_{1}-k_{3}\\
0 & \text{else.}
\end{cases}
\]
One will assume that the only dependency between $\rho_{1}$, $\rho_{2}$,
and $\rho_{3}$ that has to be taken into account is that $\rho_{1}+\rho_{2}+\rho_{3}=g$
and that $P\left(\rho_{3}=k_{3}\right)$ is given by \eqref{eq:approxrhonv1}.
Then \eqref{eq:r1r2r3} becomes 
\begin{align}
P\left(\rho_{1}=k_{1},\rho_{2}=g-k_{1}-k_{3},\rho_{3}=k_{3}\right)= & P\left(\rho_{1}=k_{1}\right)f\left(g,L_{1}L_{2},g-k_{3}\right).\label{eq:r1r2r3-1}
\end{align}
Combining \eqref{eq:DeMoivre} and \eqref{eq:r1r2r3-1}, one gets
\begin{align}
P\left(\rho_{1}=k_{1},\rho_{2}=g-k_{1}-k_{3},\rho_{3}=k_{3}\right)= & f\left(g,L_{1},k_{1}\right)f\left(g,L_{1}L_{2},g-k_{3}\right),\label{eq:approxr1r2r3}
\end{align}
which may also be written as

\begin{align}
P\left(\rho_{1}=k_{1},\rho_{2}=k_{2},\rho_{3}=g-k_{1}-k_{2}\right)= & f\left(g,L_{1},k_{1}\right)f\left(g,L_{1}L_{2},k_{1}+k_{2}\right).\label{eq:approxr1r2r3-1}
\end{align}

Figure~\ref{fig:Casenv2} (left) shows the joint pmf $P\left(\rho_{1}=k_{1},\rho_{2}=k_{2}\right)$
deduced from \eqref{eq:r1r2r3} as a function of $k_{1}$ and $k_{2}$
for different values of $g$ with $L_{1}=50$ and $L_{2}=50.$ Figure~\ref{fig:Casenv2}
(right) shows the complementary CDF of $\rho_{3}$ again deduced from
\eqref{eq:r1r2r3} for different values of $g$ with $L_{1}=50$ and
$L_{2}=50.$ Now, when $g=30$, $\Pr\left(\rho_{3}>1\right)$ is about
1.3\%. When $g=40$, $\Pr\left(\rho_{3}>1\right)$ is about 3.8\%.
In both cases, the exponential term of $K\left(n_{\text{v}}+1\right)$
in \eqref{eq:KSLE} and \eqref{eq:KLUT} will thus be of limited impact.
Considering $n_{\text{v}}=2$ allows one to consider much larger generations
than with $n_{\text{v}}=1$.

\begin{figure}
\begin{centering}
\includegraphics[width=0.5\columnwidth]{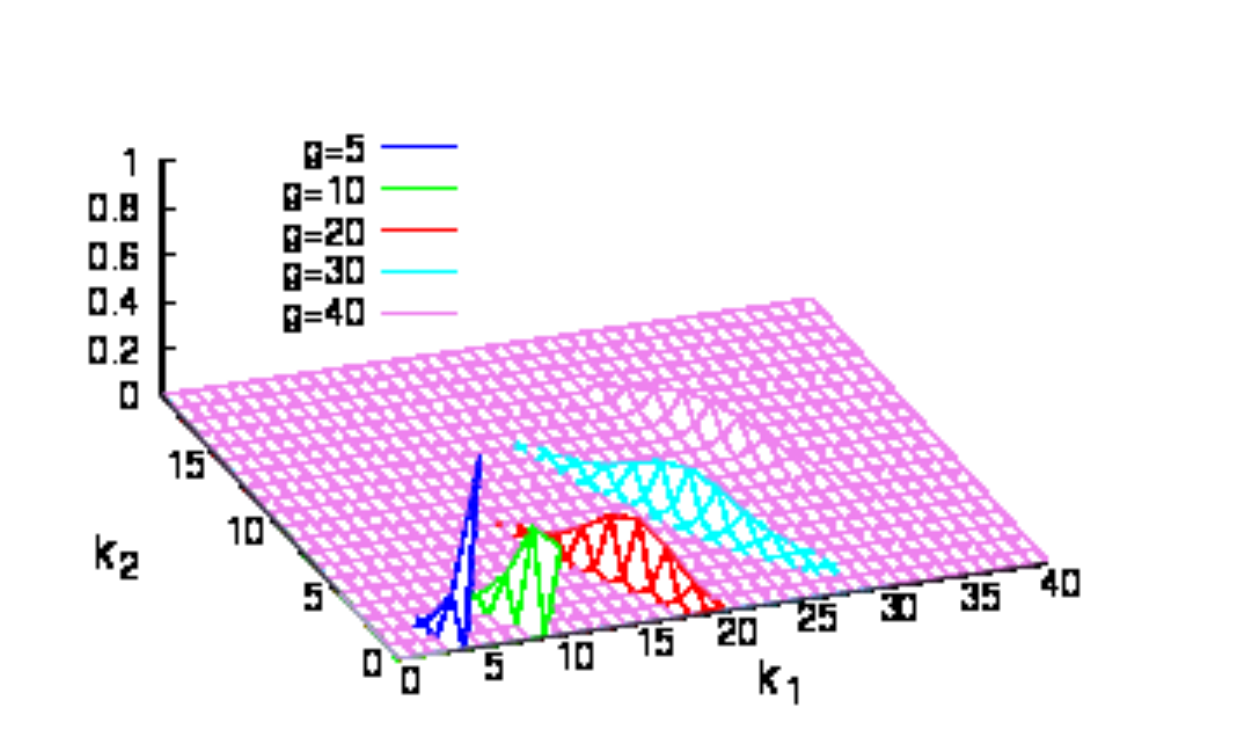}\includegraphics[width=0.5\columnwidth]{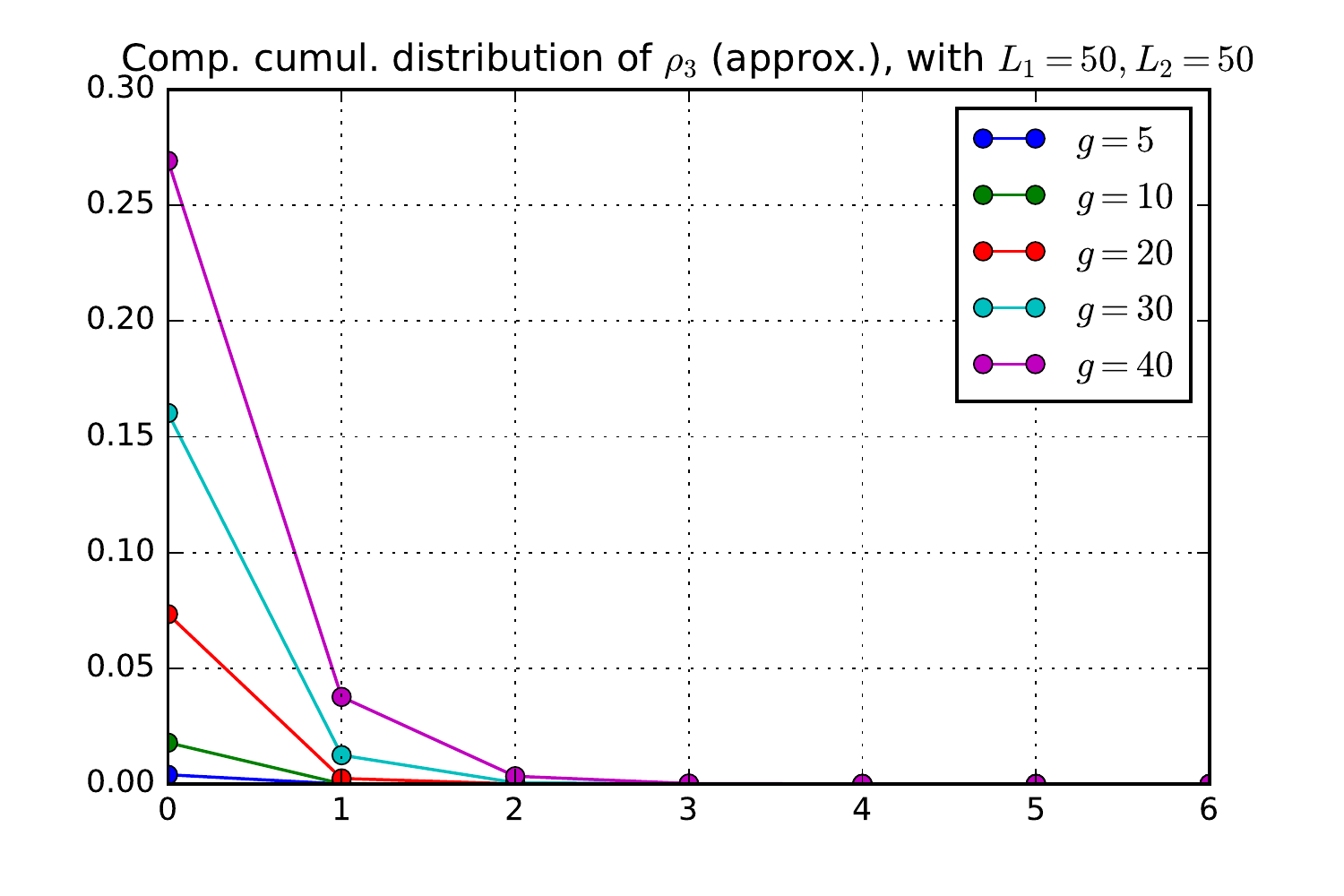} 
\par\end{centering}
\protect\caption{Case $n_{\text{v}}=2$, joint distribution of $\rho_{1}$ and $\rho_{2}$
(left) and approximated complementary cumulative distribution function
of $\rho_{3}$ deduced from \eqref{eq:approxrhonv1} (right) for different
values of $g$ for $L_{1}=50$ and $L_{2}=50$.\label{fig:Casenv2}}
\end{figure}
Finally, Figure~\ref{fig:rho3L1L2} shows the complementary CDF of
$\rho_{3}$ for $g=40$ and different values of the pair $\left(L_{1},L_{2}\right)$
such that $L_{1}+L_{2}=100.$ Choosing $L_{1}=L_{2}$ provides the
smallest probability of rank deficiency, which is consistent with
the hypothesis that the rank deficiency is mainly due to nodes having
selected the same random coding subvectors. The maximum of the product
$L_{1}\dots L_{n_{\text{v}}}$ with a constraint on $L_{1}+\dots+L_{n_{\text{v}}}$
is obtained taking $L_{1}=L_{2}=\dots=L_{n_{\text{v}}}$.

\begin{figure}
\begin{centering}
\includegraphics[width=0.5\textwidth]{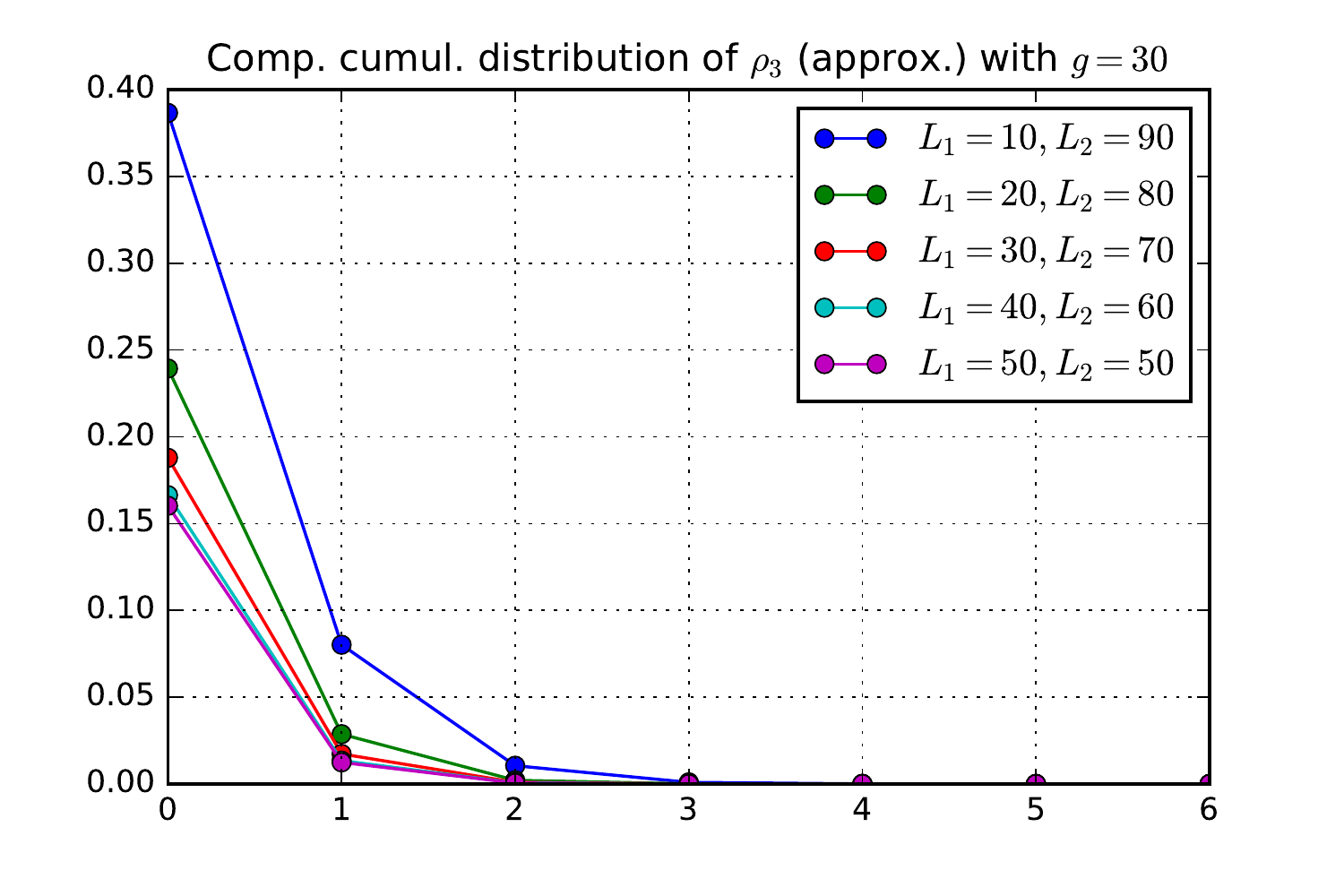} 
\par\end{centering}
\protect\caption{Case $n_{\text{v}}=2$, complementary cumulative distribution function
of $\rho_{3}$ deduced from \eqref{eq:approxrhonv1} for $g=30$ and
different values of $L_{1}$ and $L_{2}$ such that $L_{1}+L_{2}=100$.\label{fig:rho3L1L2}}
\end{figure}

\section{Performance Evaluation}

\label{sec:PerfEval}

Several simulation scenarios have been considered to evaluate the
performance of NeCorPIA in terms of complexity and decoding error
probability. In each simulation run, a set of $g$ packets $\mathbf{x}_{1},\ldots,\mathbf{x}_{g}$
containing $n_{\text{v}}$ random coding subvectors with elements
in $\mathbb{F}_{2}$ and with the same STS information are generated.
The payload $\mathbf{\boldsymbol{\pi}}=\left(\boldsymbol{\xi},\mathbf{d}\right)$
is replaced by a unique sufficiently long packet identifier to ensure
that the $g$ randomly generated packets are linearly independent.
Hash functions producing hash of different lengths are considered.
In this section, the NC operations are simulated by the generation
of a full-rank $g\times g$ random coding matrix $\mathbf{A}$ with
elements in $\mathbb{F}_{2}$. The received packets are stored in
a matrix \textbf{$\mathbf{Y=AX}$} of full rank. Writing $\mathbf{Y}$
in RREF, one obtains \eqref{eq:Echelon}.

In all simulations, the total length $\sum_{\ell=1}^{n_{\text{v}}}L_{\ell}$
of the NC subvectors is fixed at $100$.

Upon completion of Algorithm~1, with the final list of candidate
unmixing vectors satisfying \eqref{eq:Constr1}-\eqref{eq:ConstrLast},
one is always able to get $\mathbf{x}_{1},\ldots,\mathbf{x}_{g}$,
since $\mathbf{A}$ is of full rank. Nevertheless, other unmixed packets
may be obtained, even if they satisfy all the previous constraints
when the rank of the NC header is not sufficient and the hash was
inefficient. To evaluate the probability of such event, one considers
first the number $n_{\text{w}}$ of different unmixing vectors provided
by Algorithm~$1$. Among these $n_{\text{w}}$ vectors, $g$ of them
lead to the generated packets, the $n_{\text{w}}-g$ others to erroneous
packets. Considering a hash of $L_{\text{h}}$ elements of $\mathbb{F}_{q}$,
the probability of getting a given hash for a randomly generated payload
is uniform and equal to $1/q^{L_{\text{h}}}$. The probability that
one of the $n_{\text{w}}-g$ erroneous packets has a satisfying hash
is thus $1/q^{L_{\text{h}}}$ and the probability that none of them
has a satisfying hash is $\left(1-1/q^{L_{\text{h}}}\right)^{n_{\text{w}}-g}$.
The probability of decoding error is then 
\begin{equation}
P_{\text{e}}=1-\left(1-1/q^{L_{\text{h}}}\right)^{n_{\text{w}}-g}.\label{eq:DecError}
\end{equation}
An upper bound for $n_{\text{w}}$ is provided by $N_{\text{b}}\left(n_{\text{v}}+1\right)$.
In what follows, this upper bound is used to get an upper bound for
$P_{\text{e}}$ evaluated as 
\begin{equation}
\overline{P}_{\text{e}}=E\left(1-\left(1-1/q^{L_{\text{h}}}\right)^{N_{\text{b}}\left(n_{\text{v}}+1\right)-g}\right),\label{eq:approxPe}
\end{equation}
where the expectation is taken either considering the pmf of $\left(\rho_{1},\dots,\rho_{n_{\text{V}}+1}\right)$
or an estimated pmf obtained from experiments.

The other metrics used to evaluate the performance of NeCoRPIA are: 
\begin{itemize}
\item the number of branches explored in the decoding tree before being
able to decode all the packets $\mathbf{x}_{1},\ldots,\mathbf{x}_{g}$, 
\item the arithmetic complexity of the decoding process. 
\end{itemize}
In what follows, to evaluate the complexity, one assumes that $L_{\text{h}}=2$~bytes
or $L_{\text{h}}=4$~bytes and that $L_{\text{x}}=256$~bytes. For
the various constants in the arithmetic complexity, one chooses $K_{\text{m}}=2$,
$K_{\text{R}}=3$, see, \emph{e.g.}, \cite{Cormen2009}, and $K_{\text{c}}=3$,
see \cite{Touch1995}.

Averages are taken over 1000 realizations.

\subsection{Case $n_{\text{v}}=1$}

The theoretical and experimental distributions of $\rho_{1}$, evaluated
using \eqref{eq:DeMoivre}, for $L_{1}=100$ and $g=20$, are represented
in Figure~\ref{fig:L1RhoComplexity} (left). A very good match between
theoretical and experimental distributions is observed.

\begin{figure}
\begin{centering}
\includegraphics[width=0.5\columnwidth]{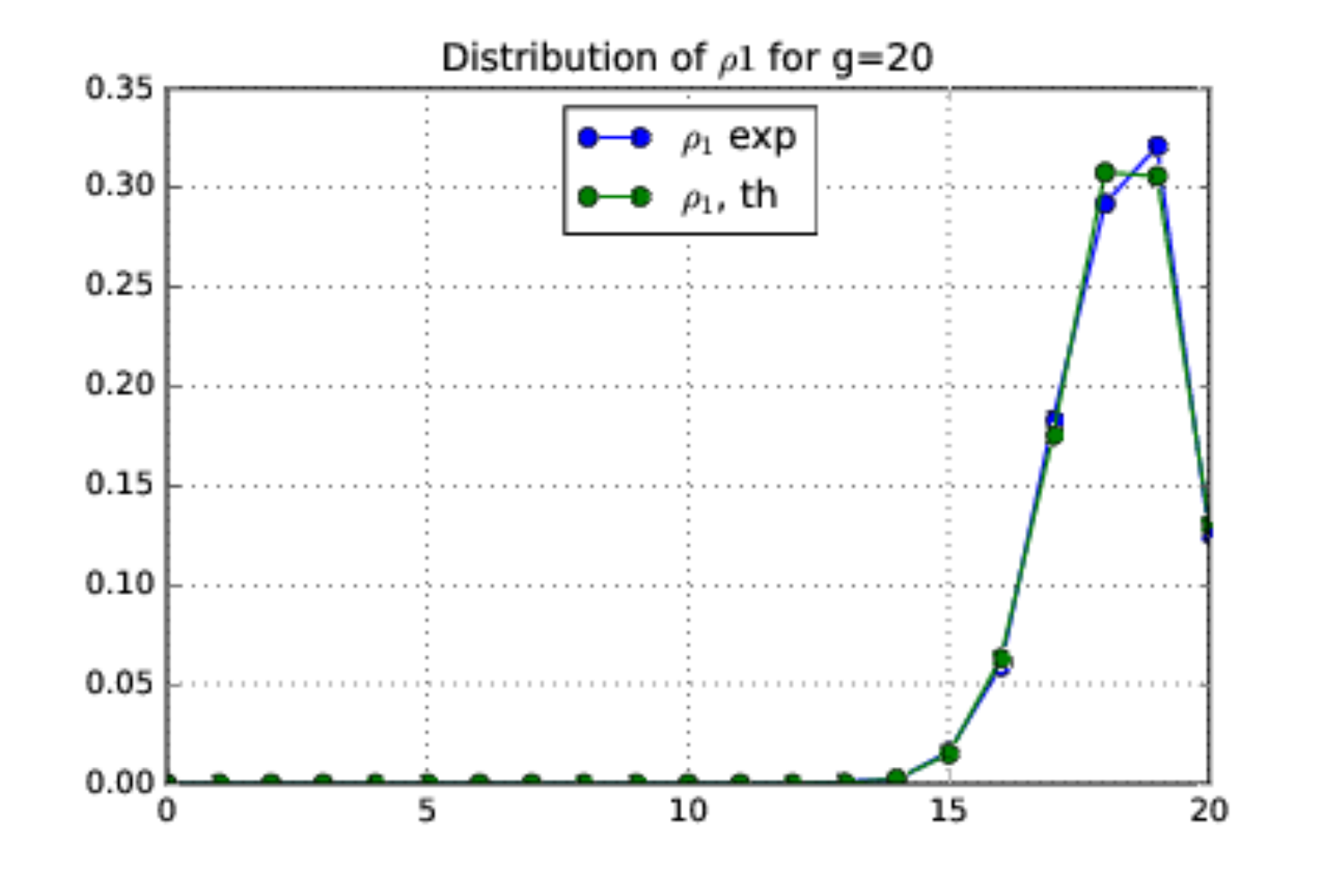}\includegraphics[width=0.5\columnwidth]{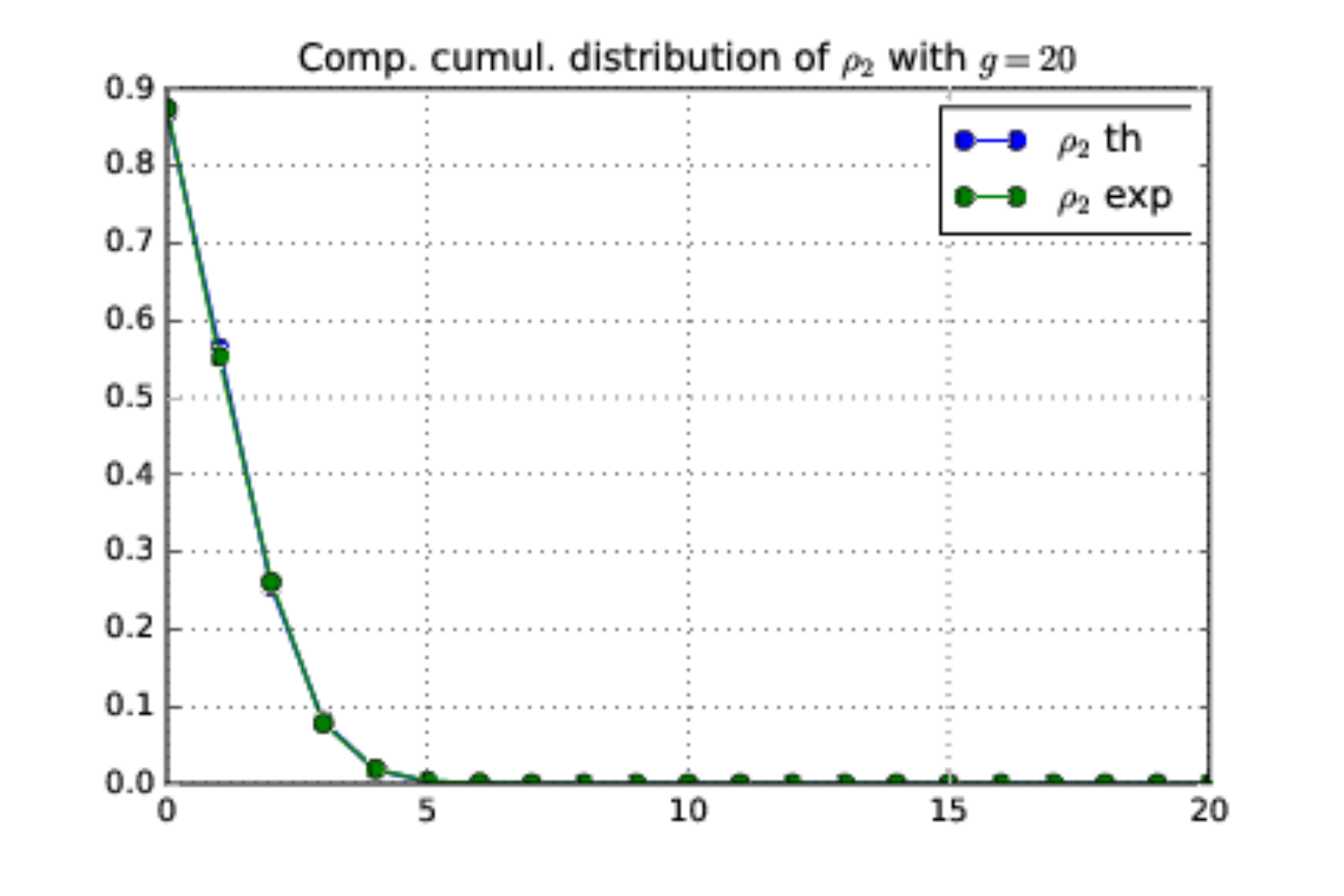} 
\par\end{centering}
\protect\caption{Case $n_{\text{v}}=1$, theoretical and experimental distributions
of $\rho_{1}$ (left) when $g=20$ and theoretical and practical complementary
cumulative distribution function of $\rho_{2}$ (right) when $g=20$
for $L_{1}=100$.}
\end{figure}
For a fixed value of $\rho_{1}$, the upper-bound \eqref{eq:Nb} for
the number of branches in the decoding tree boils down to 
\[
N_{\text{b}}=\rho_{1}+\rho_{1}q^{g-\rho_{1}}.
\]
One is then able to evaluate the average value of $N_{\text{b}}$
\begin{equation}
E\left(N_{\text{b}}\right)=\sum_{\rho_{1}=0}^{g}f\left(g,L_{1},\rho_{1}\right)\left(\rho_{1}+\rho_{1}q^{g-\rho_{1}}\right).\label{eq:ENb}
\end{equation}

\begin{figure}
\begin{centering}
\includegraphics[width=0.5\columnwidth]{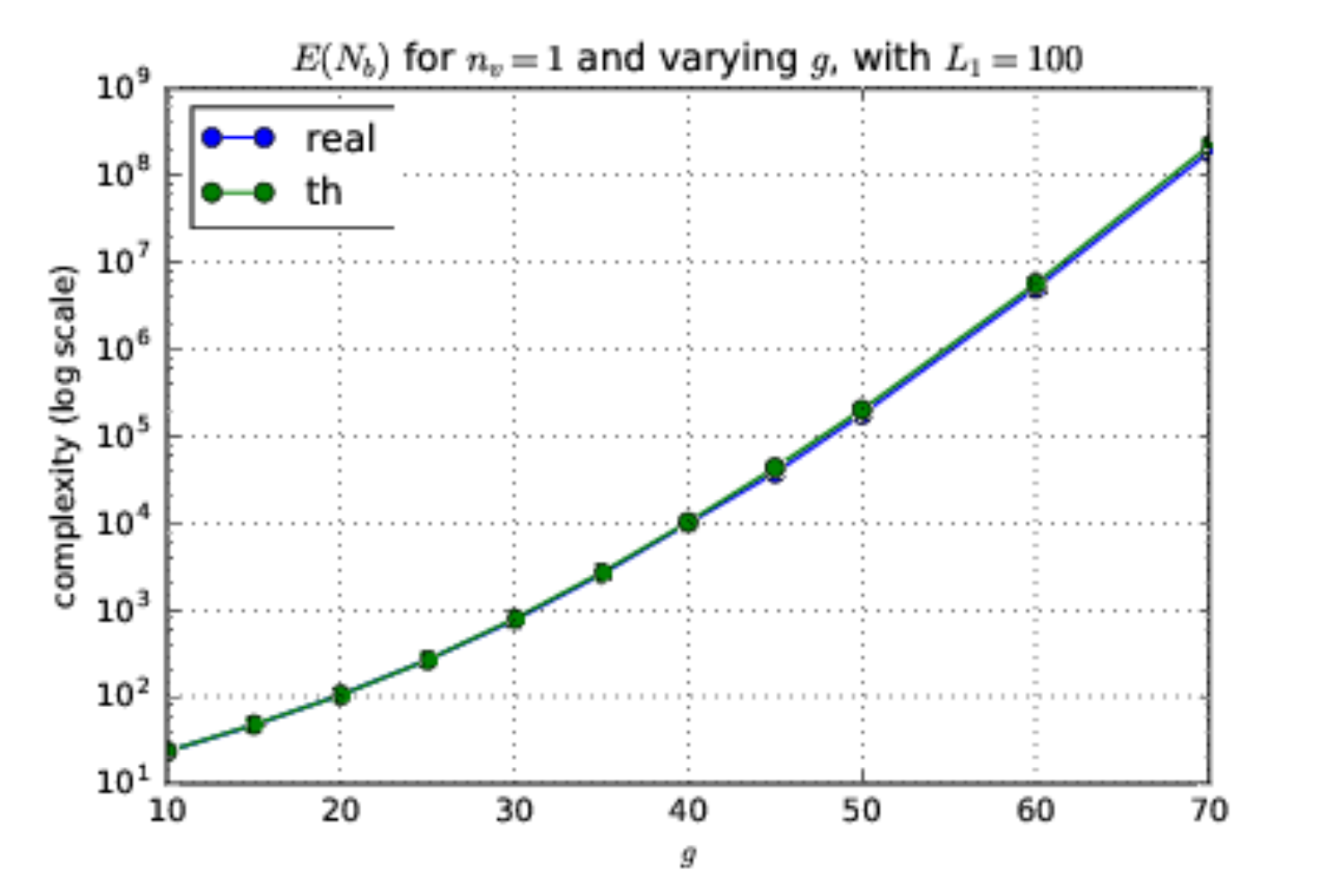}\includegraphics[width=0.5\columnwidth]{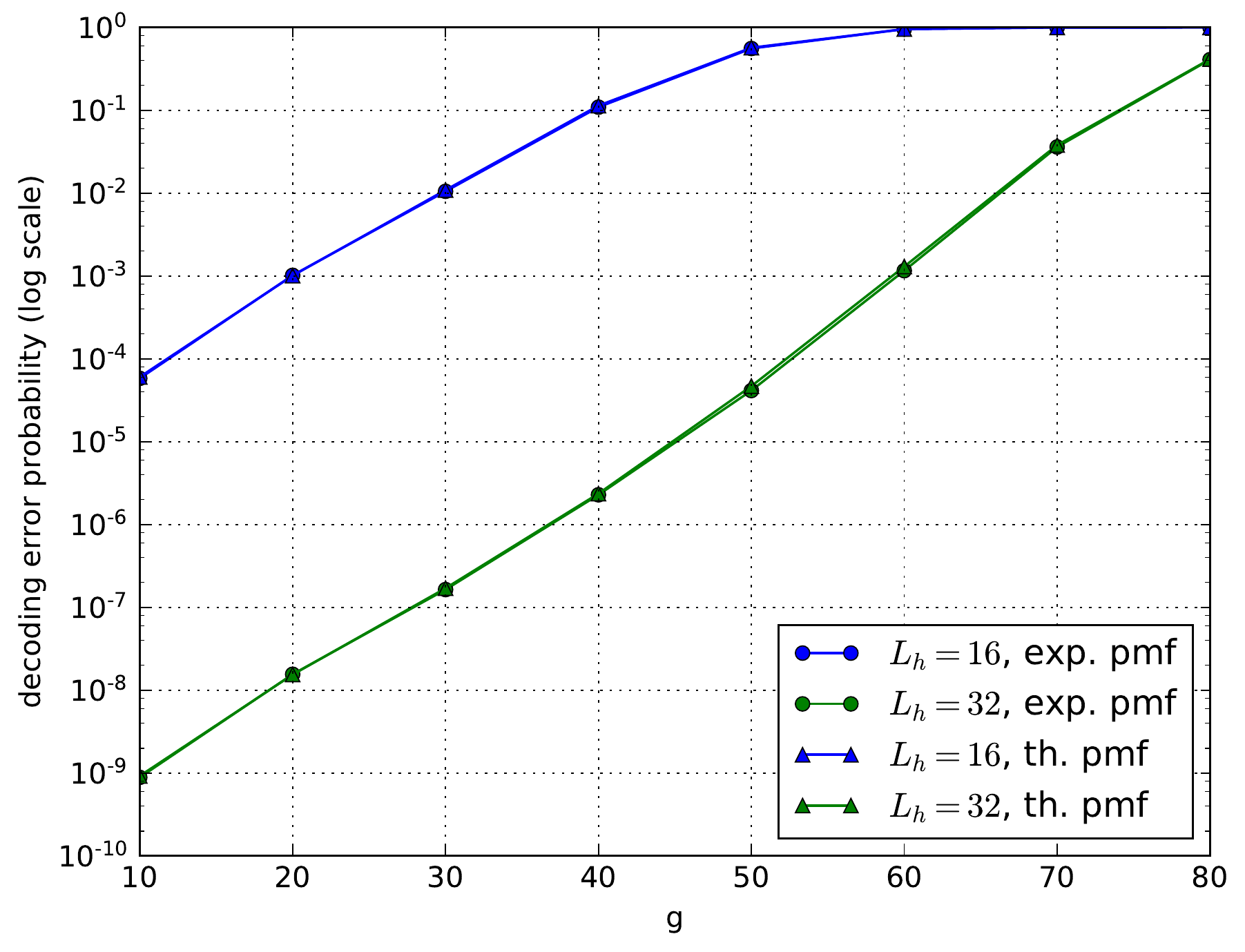} 
\par\end{centering}
\protect\caption{\label{fig:ENBErrornv1}Case $n_{\text{v}}=1$, theoretical and practical
values of $E\left(N_{\text{b}}\right)$ (left) of the decoding error
probability $\overline{P}_{\text{e}}$ with $L_{\text{h}}=16$ and
$L_{\text{h}}=32$ (right) for different values of $g$ when $L_{1}=100$.}
\end{figure}
Figure~\ref{fig:ENBErrornv1} represents the theoretical and experimental
values of $E\left(N_{\text{b}}\right)$ (left) and of $\overline{P}_{\text{e}}$
(right) as a function of $g$ when $L=100.$ The expression of $E\left(N_{\text{b}}\right)$
provided by \eqref{eq:ENb} as well as that of $\overline{P}_{\text{e}}$
given by \eqref{eq:approxPe} match well the experimental results.
As expected, for a given value of $g$, the decoding error probability
is much less with $L_{\text{\text{h}}}=32$ than with $L_{\text{\text{h}}}=16$.
When $g=20$ and $L_{\text{h}}=16$, one gets $\overline{P}_{\text{e}}=10^{-3}$,
which may be sufficiently small for some applications.

The arithmetic complexity of both decoding algorithms is then compared
to that of a plain NC decoding. The latter requires only a single
RREF (once the nodes have agreed on the NC vector they should select).
Since both DeRPIA decoding algorithms also require an initial RREF,
the complexity ratio is always larger than one.

In the case $n_{\text{v}}=1$, the arithmetic complexity of decoding
of plain NC-encoded packets is thus

\[
A_{\text{NC}}=3g^{2}L_{\text{x}},
\]
where the length $L_{\text{x}}$ is expressed as the number of elements
of $\mathbb{F}_{q}$ in which the NC operations are performed. The
arithmetic complexities of DeRPIA-SLE and DeRPIA-LUT depend on $\rho_{1}$
and $\rho_{2}$. Their expected values are 
\[
A_{\text{SLE}}\left(n_{\text{v}}=1\right)\simeq A_{\text{NC}}+E\left[\rho_{1}L_{1}^{2}+\rho_{1}L_{\text{p}}\left(K_{\text{m}}g+q^{\rho_{2}}\left(K_{\text{m}}\rho_{2}+1+K_{\text{c}}\right)\right)\right],
\]
and 
\begin{align*}
A_{\text{LUT}}\left(n_{\text{v}}=1\right) & \simeq A_{\text{NC}}+\sum_{\ell=1}^{2}E\left[\rho_{1}+L_{1}\frac{\rho_{1}\left(\rho_{1}+1\right)}{2}+\rho_{1}\left(L_{1}+1+\rho_{1}L_{1}\right)\right.\\
 & \left.+L_{1}\left(1+\rho_{1}\right)+\rho_{1}L_{\text{p}}\left(K_{\text{m}}g+q^{\rho_{2}}\left(K_{\text{m}}\rho_{2}+1+K_{\text{c}}\right)\right)\right]
\end{align*}
where the expectation is evaluated using \eqref{eq:DeMoivre} and
\eqref{eq:Prho2Nv1}.

\begin{figure}
\begin{centering}
\includegraphics[width=0.5\textwidth]{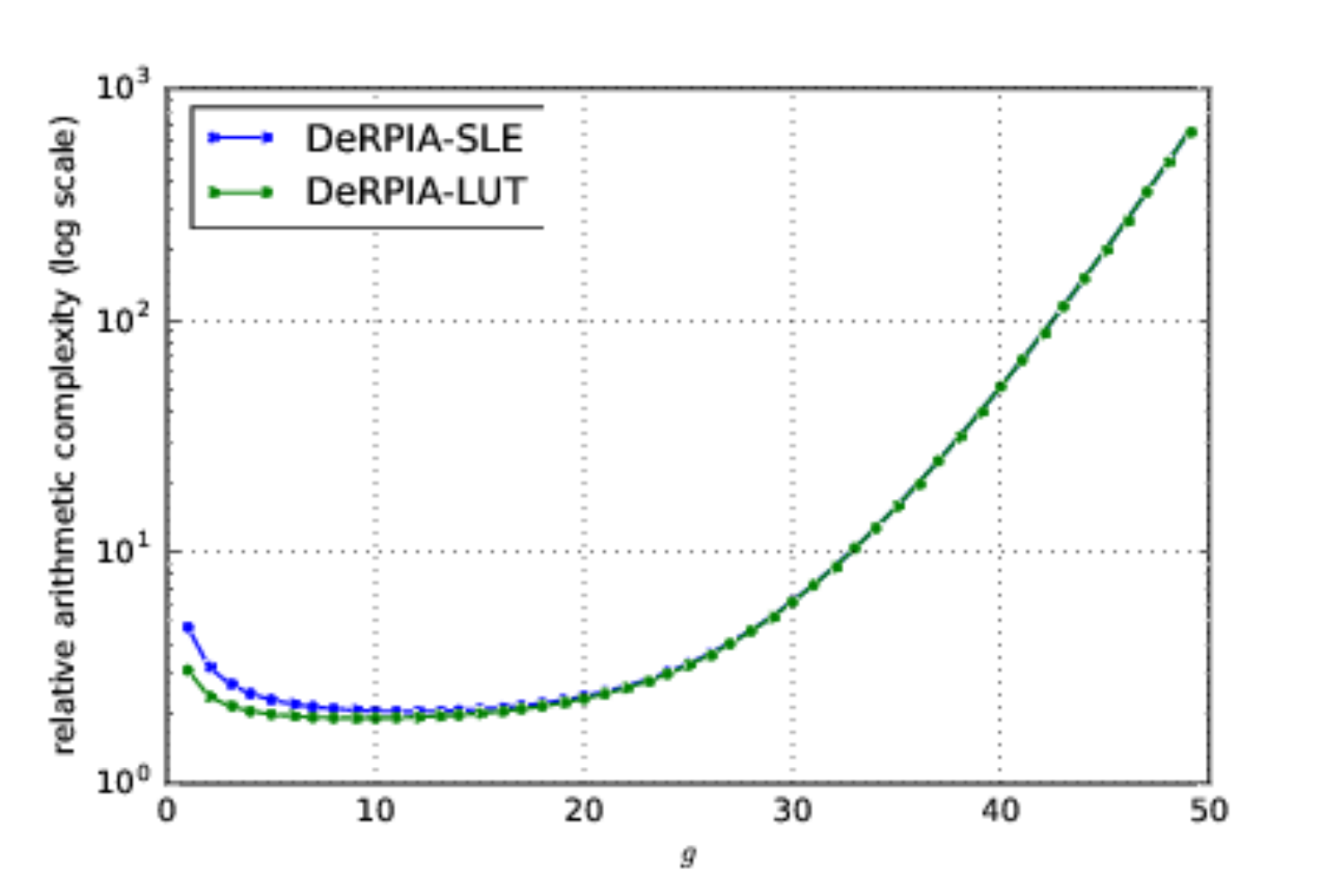} 
\par\end{centering}
\protect\caption{\label{fig:ComplexityCompnv1}Case $n_{\text{v}}=1$, evolution of
the ratio of the expected decoding complexity of DeRPIA-SLE and DeRPIA-LUT
with respect to the complexity of a simple RREF transformation for
different value of $g$ when $L_{1}=100$.}
\end{figure}
Figure~\ref{fig:ComplexityCompnv1} compares the relative arithmetic
complexity of the two variants of DeRPIA with respect to $A_{\text{NC}}$.
One sees that $A_{\text{SLE}}$ and $A_{\text{LUT}}$ are about twice
that of $A_{\text{NC}}$ when $g\leqslant25$. When $g\geqslant30$,
the complexity increases exponentially. $A_{\text{LUT}}$ is only
slightly less than $A_{\text{SLE}}$ for small values of $g.$ Again,
for values of $g$ larger than $25$, the exponential term dominates.

\subsection{Case $n_{\text{v}}=2$}

In this case, an expression is only available for the pmf of $\rho_{1}$
as a function of $g$ and $L_{1}$, see Section~\ref{sub:NVv2}.
Figure~\ref{fig:L2Rho} represents the histograms of $\rho_{1}$,
$\rho_{2}$, and $\rho_{3}$ for different values of $g$, as well
as the theoretical pmf of $\rho_{1}$.

\begin{figure}
\begin{centering}
\includegraphics[clip,width=1\columnwidth]{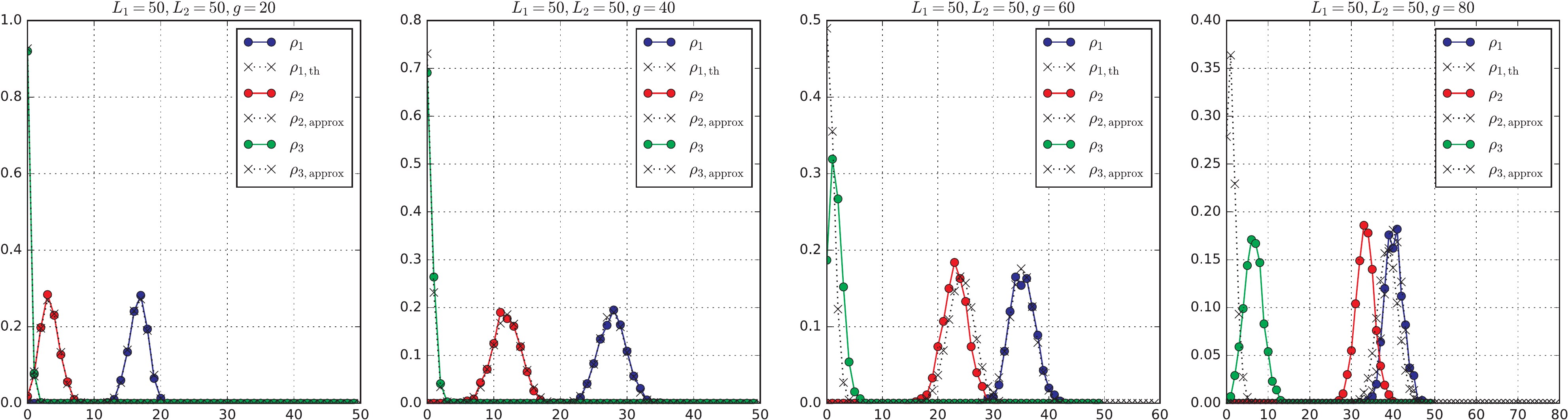} 
\par\end{centering}
\protect\caption{Case $n_{\text{v}}=2$, theoretical pmf of $\rho_{1}$ and histograms
of $\rho_{1}$, $\rho_{2}$, and $\rho_{3}$ for different values
of $g$ when $L_{1}=L_{2}=50$. \label{fig:L2Rho}}
\end{figure}
Figure~\ref{fig:L2Rho-zoom} shows that the approximation of $P\left(\rho_{3}=k\right)$
provided by \eqref{eq:approxrhonv1} matches well the histogram of
$\rho_{3}$ when $g\leqslant40$. When $g=80$, the approximation
is no more valid: the effect of cycles, illustrated in Example~\ref{exa:Cycle},
becomes significant.

\begin{figure}
\begin{centering}
\includegraphics[clip,width=1\columnwidth]{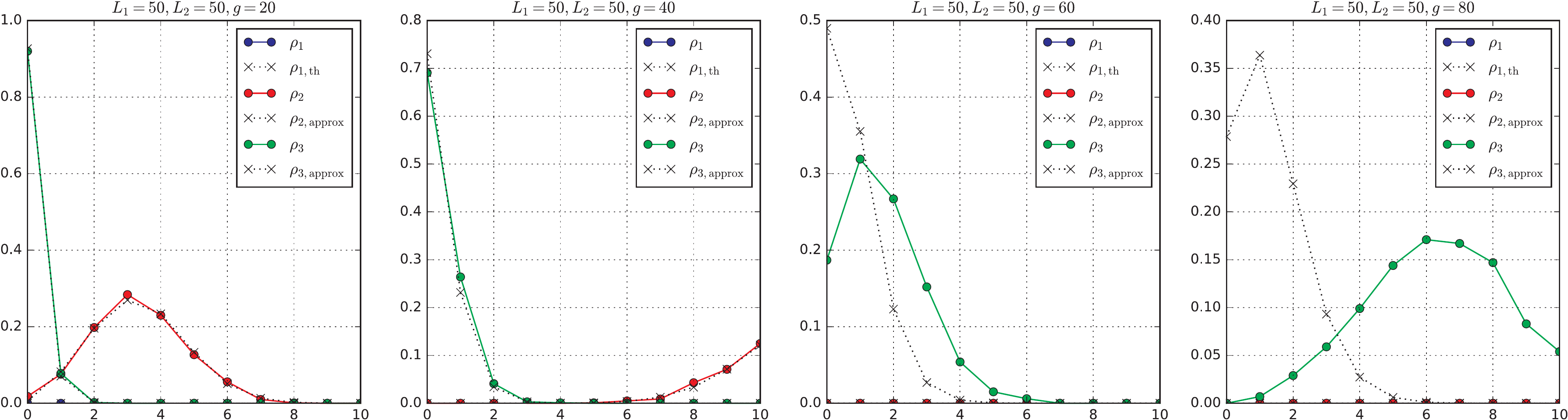} 
\par\end{centering}
\protect\caption{Case $n_{\text{v}}=2$, distribution of $\rho_{1}$ and histograms
of $\rho_{1}$, $\rho_{2}$, and $\rho_{3}$ for different values
of $g$ for $L_{1}=L_{2}=50$ (zoom of Fig~\ref{fig:L2Rho}). \label{fig:L2Rho-zoom}}
\end{figure}
For the complexity evaluation, in the case $N_{\text{v}}=2$, for
a given value of $\rho_{1}$ and $\rho_{2}$, the upper bound \eqref{eq:Nb}
becomes\emph{ 
\begin{equation}
N_{\text{b}}=\rho_{1}+\rho_{1}\left(\rho_{2}+1\right)+\rho_{1}\left(\rho_{2}+1\right)q^{g-\rho_{1}-\rho_{2}}.\label{eq:Nn-1}
\end{equation}
}The average value of $N_{\text{b}}$ is then evaluated using the
approximation \eqref{eq:approxr1r2r3-1} of the joint pmf of $\left(\rho_{1},\rho_{2},\rho_{3}\right)$
as 
\begin{equation}
E\left(N_{\text{b}}\right)=\sum_{\rho_{1}+\rho_{2}\leqslant g}f\left(g,L_{1},\rho_{1}\right)f\left(g,L_{1}L_{2},\rho_{1}+\rho_{2}\right)\left(\rho_{1}+\rho_{1}\left(\rho_{2}+1\right)+\rho_{1}\left(\rho_{2}+1\right)q^{g-\rho_{1}-\rho_{2}}\right).\label{eq:ENb-1}
\end{equation}

\begin{figure}
\begin{centering}
\includegraphics[width=0.5\columnwidth]{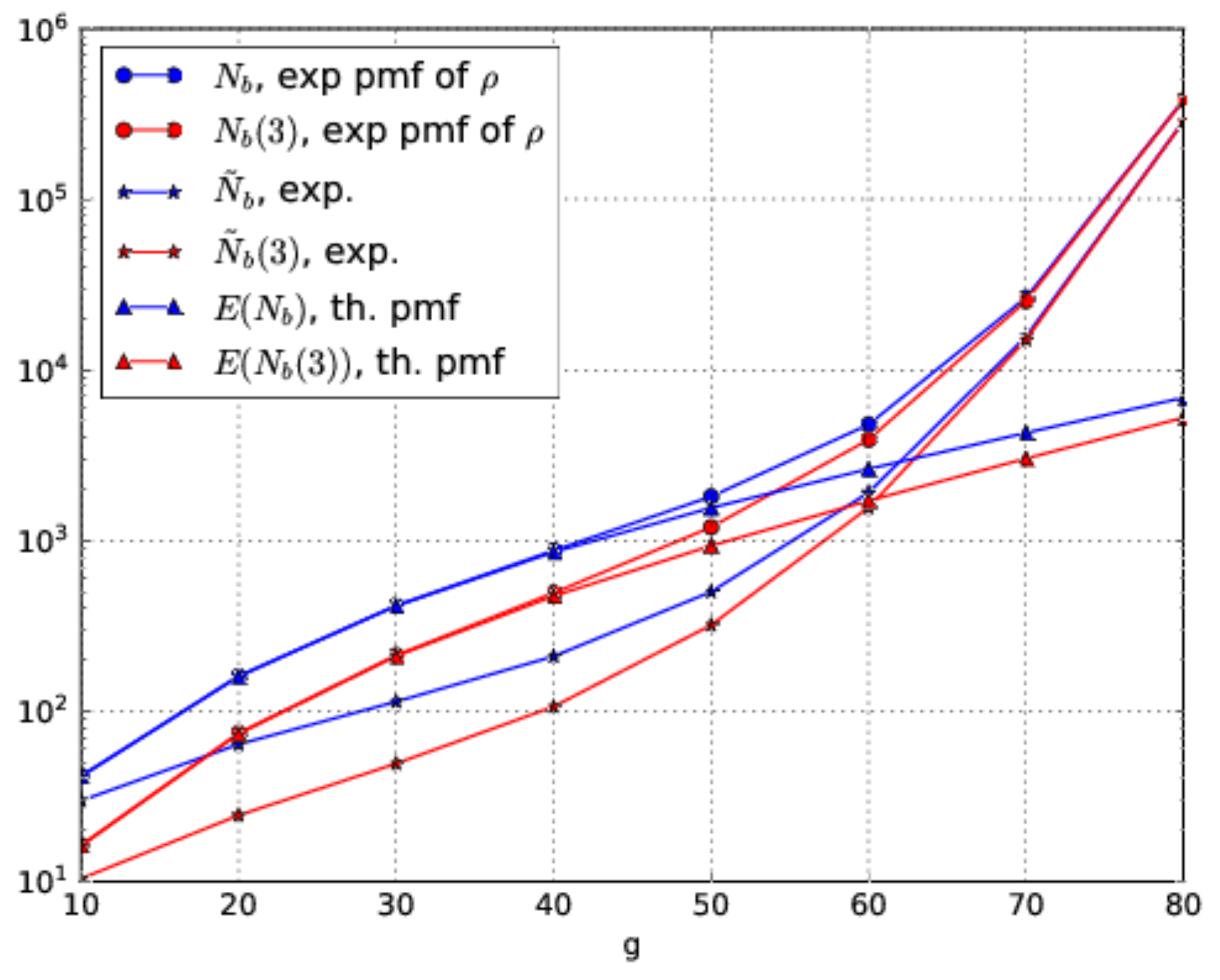}\includegraphics[width=0.5\columnwidth]{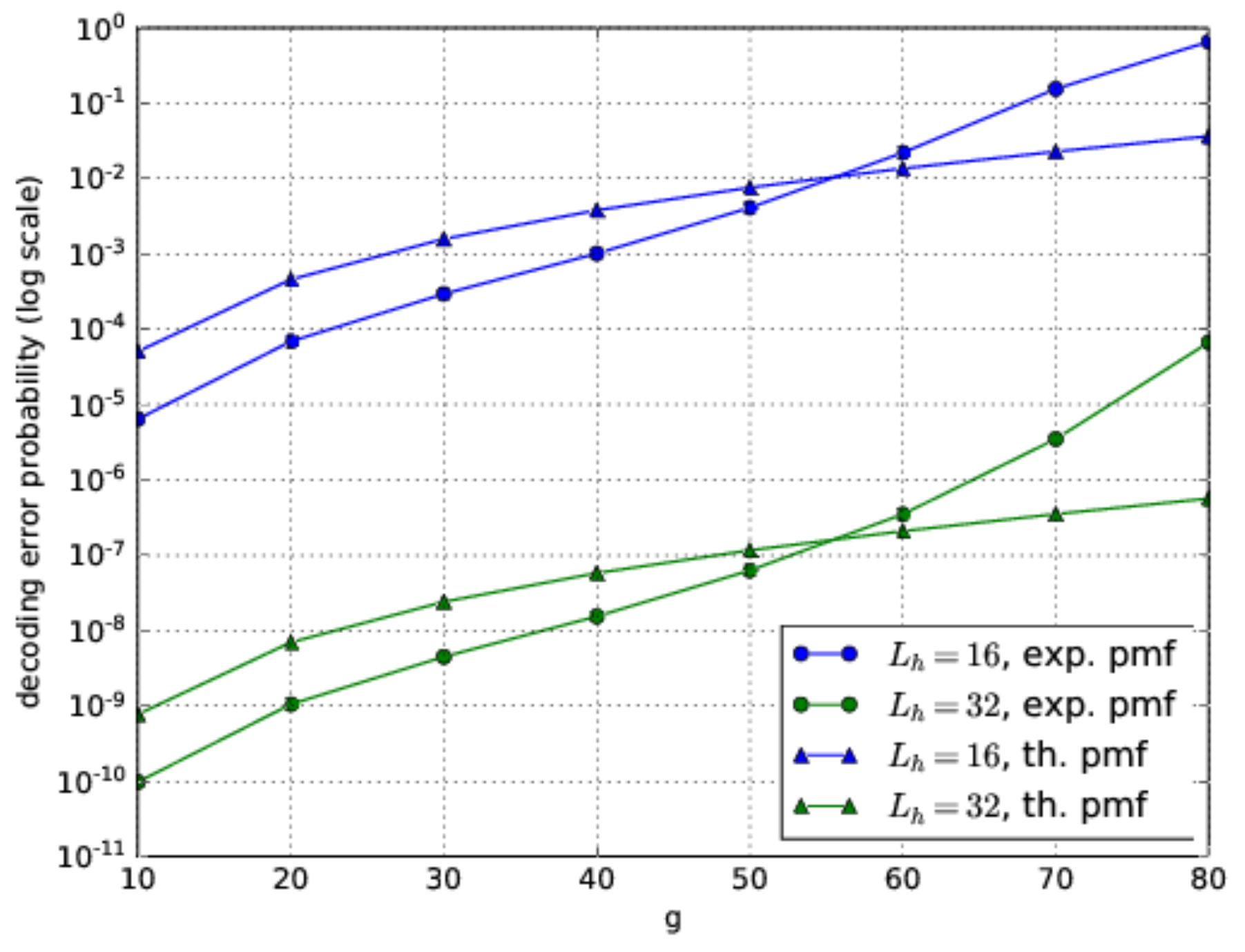} 
\par\end{centering}
\protect\caption{\label{fig:ENBErrornv2}Case $n_{\text{v}}=2$, theoretical and experimental
values of $E\left(N_{\text{b}}\left(3\right)\right)$ and $E\left(N_{\text{b}}\right)$
(left) and of the decoding error probability $\overline{P}_{\text{e}}$
with $L_{\text{h}}=16$ and $L_{\text{h}}=32$ (right) for different
values of $g$ when $L_{1}=50$ and $L_{2}=50$.}
\end{figure}
Figure~\ref{fig:ENBErrornv2} (left) represents the theoretical and
experimental values of $E\left(N_{\text{b}}\left(3\right)\right)$
and $E\left(N_{\text{b}}\right)$ as a function of $g$ when $L_{1}=50$
and $L_{2}=50$. The theoretical values of $E\left(N_{\text{b}}\left(3\right)\right)$
and $E\left(N_{\text{b}}\right)$ are evaluated in two ways: First,
with the estimated pmf of $\left(\rho_{1},\rho_{2},\rho_{3}\right)$
given by \eqref{eq:approxr1r2r3-1} and second with the estimate of
this pmf obtained from experiments. The average number $\widetilde{N}_{\text{b}}$
of branches in the tree and in the last level of the tree $\widetilde{N}_{\text{b}}\left(3\right)$
obtained from experiments are also provided. The value of $E\left(N_{\text{b}}\left(3\right)\right)$
and $E\left(N_{\text{b}}\right)$ evaluated from the experimental
pmf are good upper-bounds for $\widetilde{N}_{\text{b}}\left(3\right)$
and $\widetilde{N}_{\text{b}}$. The values of $E\left(N_{\text{b}}\left(3\right)\right)$
and $E\left(N_{\text{b}}\right)$ obtained from \eqref{eq:approxr1r2r3-1}
match well those obtained from the estimated pmf of $\left(\rho_{1},\rho_{2},\rho_{3}\right)$
only for $g\leqslant50$. When $g\geqslant60$, the lack of accuracy
of the theoretical pmf of $\rho_{3}$ becomes significant: $\widetilde{N}_{\text{b}}\left(3\right)$
and $\widetilde{N}_{\text{b}}$ are underestimated. One observes that
considering two encoding subvectors significantly reduces the number
of branches that have to be explored in the decoding tree. For example,
when $g=70$, $E\left(N_{\text{b}}\right)\simeq1.6\times10^{4}$ with
$n_{\text{v}}=2$, whereas $E\left(N_{\text{b}}\right)\simeq1.8\times10^{8}$
with $n_{\text{v}}=1$.

Figure~\ref{fig:ENBErrornv2} (right) represents $\overline{P}_{\text{e}}$
obtained from the approximate pmf of $\left(\rho_{1},\rho_{2},\rho_{3}\right)$
given by \eqref{eq:approxr1r2r3-1} and from the estimate of this
pmf obtained from experiments. Again, both evaluations match well
when $g\leqslant50$. Compared to the case $n_{\text{v}}=1,$ the
probability of decoding error reduces significantly thanks to the
reduction of the number of branches in the decoding tree at level
$n_{\text{v}}+1$.

In the case $n_{\text{v}}=2$, the arithmetic complexity of decoding
of plain NC-encoded packets is still $A_{\text{NC}}$. The arithmetic
complexities of DeRPIA-SLE and DeRPIA-LUT depend now on $\rho_{1}$,
$\rho_{2}$, and $\rho_{3}$. Their expected values are 
\begin{align*}
A_{\text{SLE}}\left(n_{\text{v}}=2\right) & \simeq A_{\text{NC}}+E\left[\rho_{1}L_{1}^{2}+\rho_{1}L_{2}\left(K_{\text{m}}\rho_{1}+2+\left(\rho_{2}+1\right)L_{2}\right)\right]\\
 & +E\left[\rho_{1}\left(\rho_{2}+1\right)L_{\text{p}}\left(K_{\text{m}}g+1+q^{\rho_{3}}\left(K_{\text{m}}\rho_{3}+1+K_{\text{c}}\right)\right)\right]
\end{align*}
and 
\begin{align*}
A_{\text{LUT}}\left(n_{\text{v}}=2\right) & \simeq A_{\text{NC}}+\sum_{\ell=1}^{2}E\left[\rho_{\ell}+L_{\ell}\frac{\rho_{\ell}\left(\rho_{\ell}+1\right)}{2}+\rho_{\ell}\left(L_{\ell}+1+\rho_{\ell}L_{\ell}\right)\right]\\
 & +E\left[L_{1}\left(1+\rho_{1}\right)+\rho_{1}L_{2}\left(K_{m}\rho_{1}+2+\rho_{2}\right)\right]\\
 & +E\left[\rho_{1}\left(\rho_{2}+1\right)L_{\text{p}}\left(K_{\text{m}}g+1+q^{\rho_{3}}\left(K_{\text{m}}\rho_{3}+1+K_{\text{c}}\right)\right)\right]
\end{align*}
where the expectations are evaluated using \eqref{eq:approxr1r2r3-1}.

\begin{figure}
\centering \includegraphics[width=0.45\columnwidth]{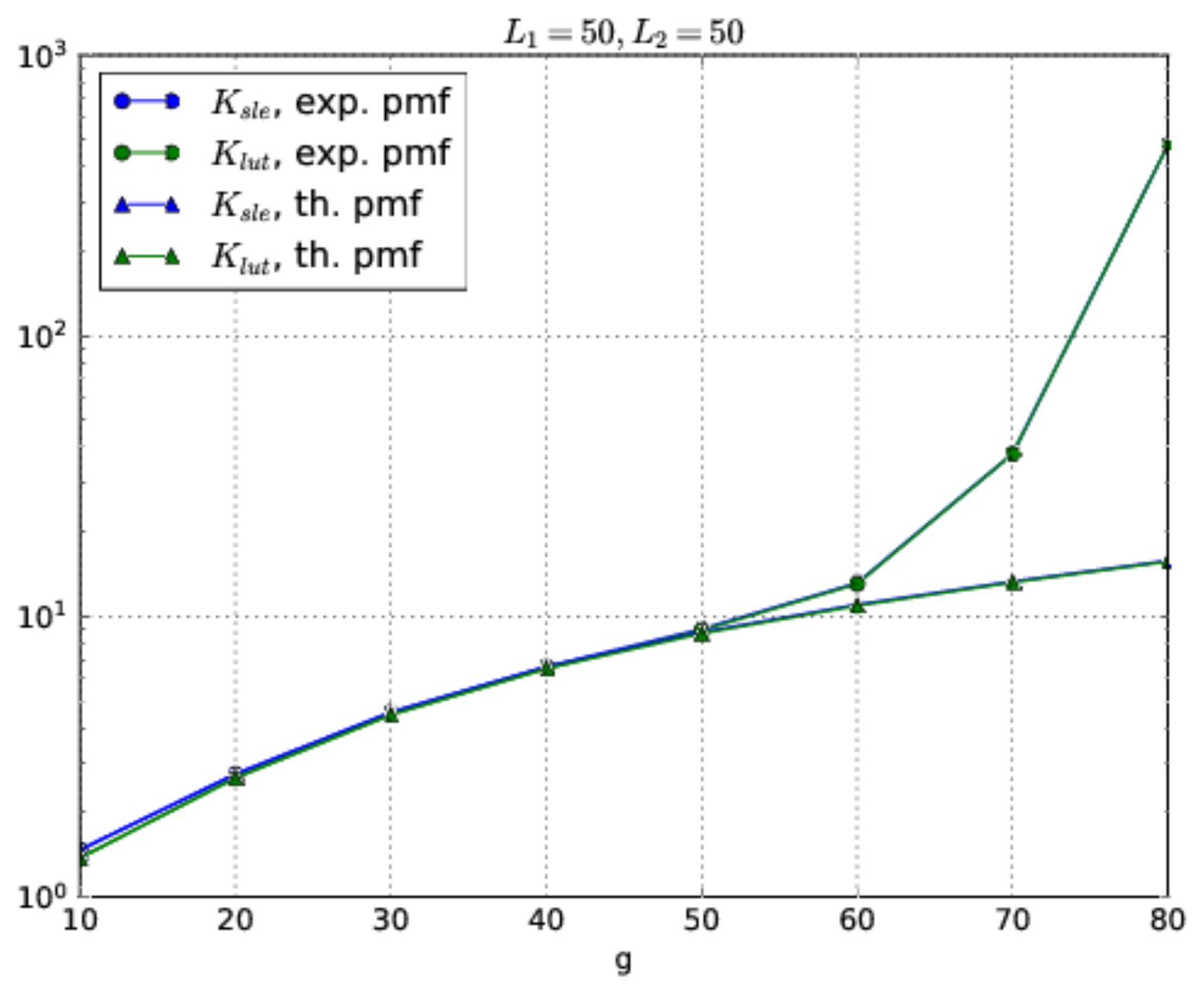}
\caption{\label{fig:ComplexityCompnv1-1}Case $n_{\text{v}}=2$, evolution
of the theoretical and experimental values of the ratio of the expected
decoding complexity of DeRPIA-SLE and DeRPIA-LUT with respect to the
complexity of a simple RREF transformation as a function of $g$ when
$L_{1}=L_{2}=50$.}
\end{figure}
Figure~\ref{fig:ComplexityCompnv1-1} compares the relative arithmetic
complexity of the two variants of DeRPIA with respect to $A_{\text{NC}}$.
One sees that $A_{\text{SLE}}$ and $A_{\text{LUT}}$ are almost equal
and less than ten times $A_{\text{NC}}$ when $g\leqslant50$. When
$g\geqslant60$, the complexity increases exponentially. This is again
not well predicted by the theoretical approximation, due to the degraded
accuracy of the theoretical expression of $\rho_{3}$ when $g\geqslant60$.
$A_{\text{LUT}}$ is again only slightly less than $A_{\text{SLE}}$
for small values of $g.$ Now, the exponential term dominates in the
complexity for values of $g$ larger than $70$.

\subsection{Case $n_{\text{v}}>2$}

In this case, since even an approximate expression of the joint pmf
of $\left(\rho_{1},\dots,\rho_{n_{\text{v}}+1}\right)$ is difficult
to obtain, only the pmf of $\rho_{1}$ and the histograms for $\rho_{2},\dots,\rho_{n_{\text{v}}+1}$
are provided, see Figure~\eqref{fig:L3Rho} for $n_{\text{v}}=3$
and Figure~\eqref{fig:L4Rho} for $n_{\text{v}}=4$.

\begin{figure}
\begin{centering}
\includegraphics[width=1\columnwidth]{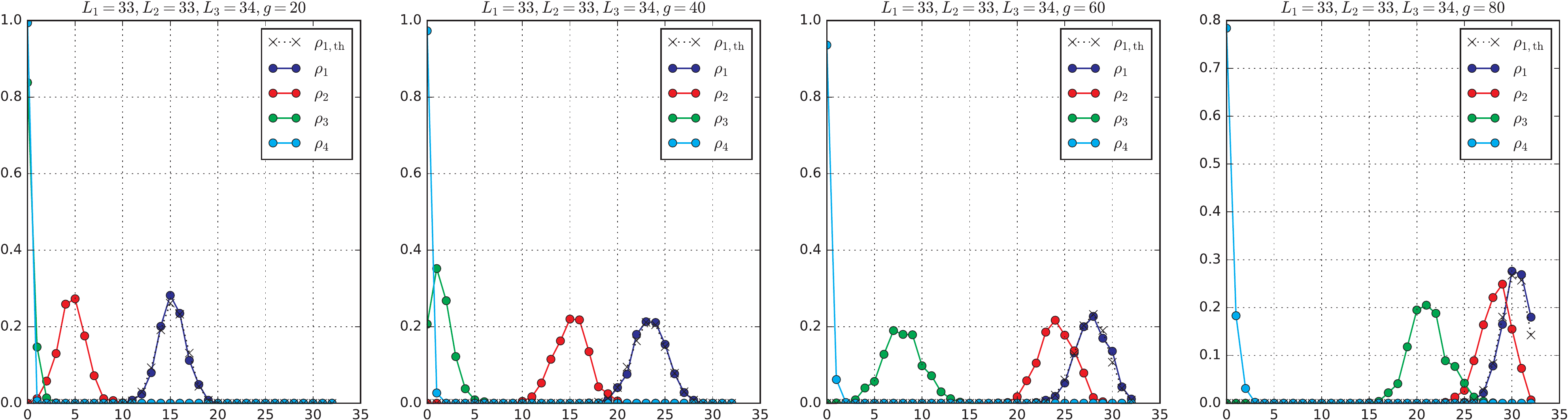} 
\par\end{centering}
\protect\caption{Case $n_{\text{v}}=3$, distribution of $\rho_{1}$ and histograms
of $\rho_{1}$, $\rho_{2}$, $\rho_{3}$, and $\rho_{4}$ for different
values of $g$ for $L_{1}=33$, $L_{2}=33$, $L_{3}=34$. \label{fig:L3Rho}}
\end{figure}
\begin{figure}
\begin{centering}
\includegraphics[width=1\columnwidth]{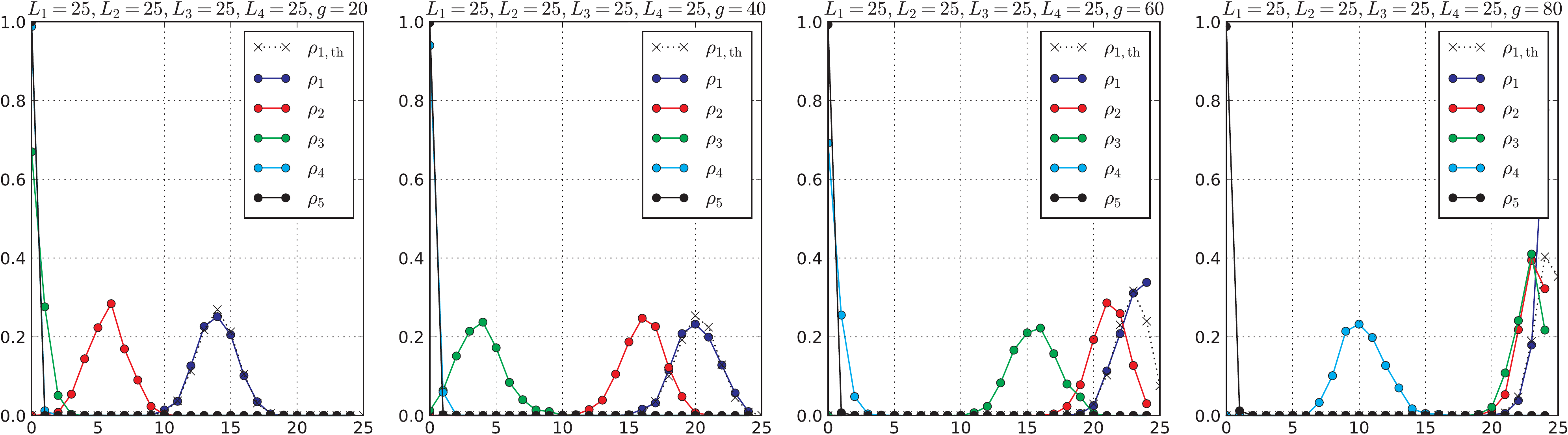} 
\par\end{centering}
\protect\caption{Case $n_{\text{v}}=4$, distribution of $\rho_{1}$ and histograms
of $\rho_{1}\dots\rho_{5}$ for different values of $g$ for $L_{1}=\dots=L_{4}=25$.
\label{fig:L4Rho}}
\end{figure}
One observes that when $n_{\text{v}}=4$, even for $g=80$, $P\left(\rho_{5}\geqslant1\right)\leqslant0.012$.
The contribution of the exponential term in the decoding complexity
will thus remain negligible. This can be observed in Figure~\ref{fig:ComplexComp},
which shows the evolution of the ratio of the decoding complexity
of DeRPIA-SLE and DeRPIA-LUT with respect to the complexity of a simple
RREF transformation as a function of $g$ for different values of
$n_{\text{v}}$. For values of $g$ for which $\rho_{n_{\text{v}}+1}$
remains small, the complexity increases with $n_{\text{v}}$, due
to the increasing number of branches that have to be considered at
intermediate levels of the decoding tree. When $\rho_{n_{\text{v}}+1}$
increases, the complexity is dominated by the exponential term in
the complexity due to the number of branches to consider at level
$n_{\text{v}}+1$ in the decoding tree. Considering a larger value
of $n_{\text{v}}$ becomes then interesting from a complexity point
of view. This phenomenon appears when $g\geqslant30$ for $n_{\text{v}}=1$,
when $g\geqslant80$ for $n_{\text{v}}=2$ and does not appear for
larger values of $n_{\text{v}}$.

The proposed NeCoRPIA scheme is thus able to perform NC without coordination
between agents. With $n_{\text{v}}=2$, compared to classical NC,
generations of 60~packets may be considered with a header overhead
of 66~\%, a vanishing decoding error probability, and a decoding
complexity about 10~times that of Gaussian elimination. With $n_{\text{v}}=3$,
generations of 80~packets may be considered, leading to a header
overhead of 25~\%, but a decoding complexity about 100~times that
of Gaussian elimination.

\begin{figure}
\begin{centering}
\includegraphics[width=0.45\columnwidth]{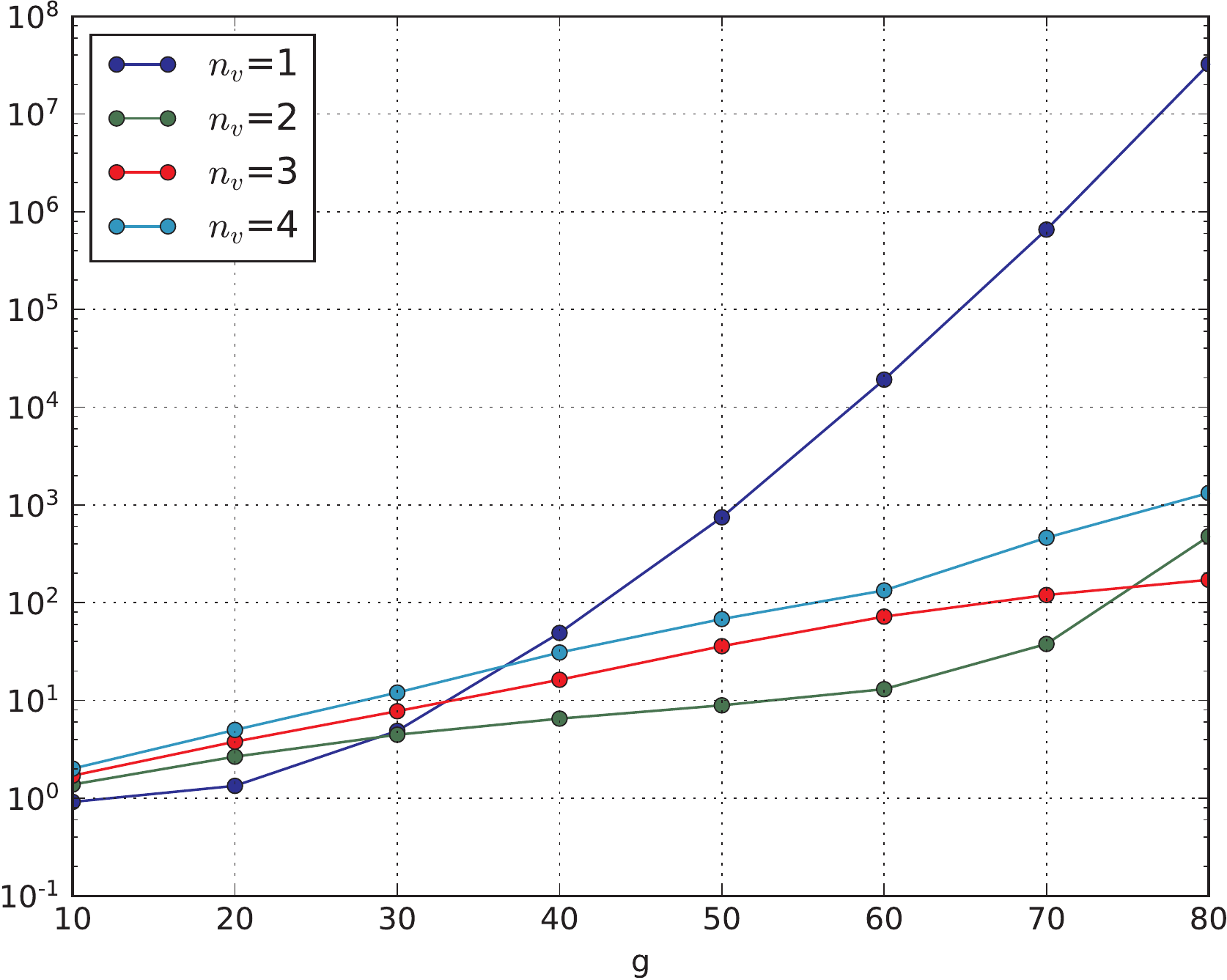} 
\par\end{centering}
\protect\caption{Evolution of the experimental values of the ratio of the decoding
complexity of DeRPIA-SLE and DeRPIA-LUT with respect to the complexity
of a simple RREF transformation as a function of $g$ for different
values of $n_{\text{v}}.$ \label{fig:ComplexComp}}
\end{figure}

\section{Comparison of packet header overhead}

\label{sec:Comparison-of-packet}

This section aims at comparing the NC header overhead when considering
the NC headers of NeCorPIA and a variant of COPE, the variable-length
headers proposed in \cite{Katti+2008}.

For that purpose, one considers a simulation framework with $N+1$
nodes randomly spread over a square of unit area. Nodes are able to
communicate if they are at a distance of less than $r$. One has adjusted
$r$ in such a way that the diameter of the graph associated to the
network is $10$ and the minimum connectivity degree is larger than
$2$. Without loss of generality, one takes Node~$N+1$ is the sink
and one selects $g$ randomly chosen nodes among the $N$ remaining
nodes are source nodes. During one simulation corresponding to a single
STS, each source node generates a single packet. The time in a STS
is further slotted and packets are transmitted at the beginning of
each time slot. The packet forwarding strategy described in \cite{Fragouli+2008}
is implemented with constant forwarding factor $d=1.5$, according
to Algorithm~4 in \cite{Fragouli+2008}. The forwarding factor determines
the average number of linear combinations of already received packets
a node has to broadcast in the next slot, each time it receives an
innovative packet in the current slot. A packet reaching the sink
is no more forwarded. The simulation is ended once the sink is able
to decode the $g$ source packets of the considered STS. In practice,
the source node only controls the duration of a STS and does not know
precisely $g$ during the considered STS. This duration may be adapted
by the sink to have a prescribed number of active source nodes during
an STS. Nevertheless, such algorithm goes beyond the scope of this
paper and to simplify, $g$ is assumed to be known.

\begin{figure}
\begin{centering}
\includegraphics[width=0.45\columnwidth]{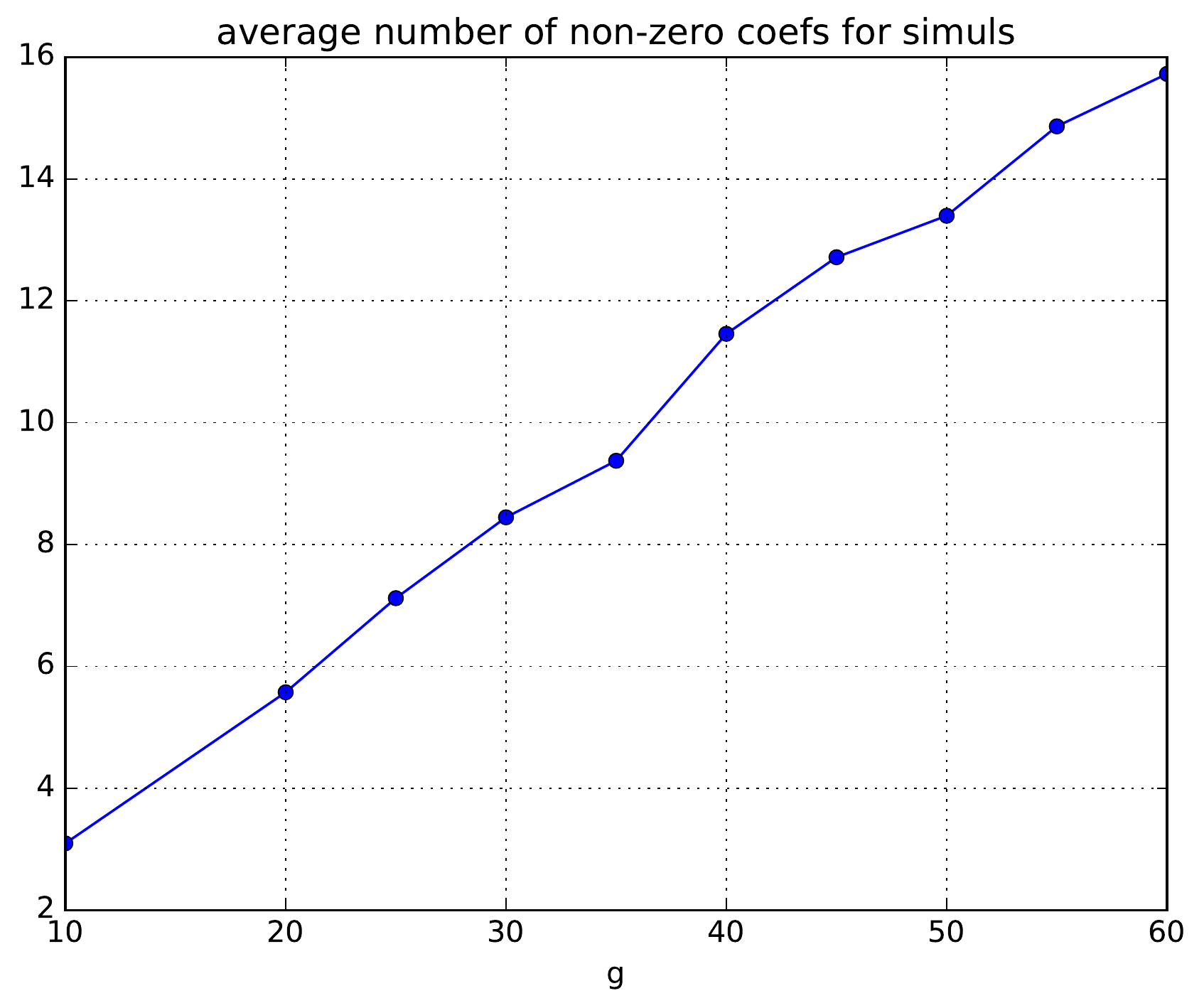} 
\par\end{centering}
\protect\caption{Average number of non-zero coefficients in plain NC headers as a function
of $g$\label{fig:AvgNonZero}}
\end{figure}
For the COPE-inspired variable-length NC protocol (called COPE in
what follows), instead of considering a fixed-length 32-bit packet
identifier as in \cite{Katti+2008}, one assumes that the maximum
number of active nodes in a STS is known. Then the length of the identifier
is adjusted so as to have a probability of collision of two different
packets below some specified threshold $p_{\text{c}}$. The packet
identifier may be a hash of the packet header and content, as in \cite{Katti+2008}.
Here, to simplify evaluations, it is considered as random, with a
uniform distribution over the range of possible identifier values.

For NeCoRPIA, NC headers of $n_{\text{v}}$ blocks of length $L_{1}=\dots=L_{n_{\text{v}}}$
are considered. To limit the decoding complexity, the length of each
block is adjusted in such a way that, in average, the estimate of
the upper bound of the number of branches in the last level of the
decoding tree \eqref{eq:Nb} is less than $10^{3}$. This average
is evaluated combining \eqref{eq:Nb} with \eqref{eq:Prho2Nv1} when
$n_{\text{v}}=1$ and \eqref{eq:Nb} with \eqref{eq:approxr1r2r3}
when $n_{\text{v}}=2$. Then, the size of the hash is adjusted in
such a way that, in case of collision, the probability of being unable
to recover the original packets \eqref{eq:approxPe} is less than
$p_{\text{c}}$.

For plain NC protocol, assuming that Node~$i$ uses $\mathbf{e}_{i}\in\mathbb{F}_{2}^{N}$,
$i=1,\dots,N$, as NC header as in \cite{Chou+2003+A3C}, there is
no collision and the header is of length $N$ elements of $\mathbb{F}_{2}$.
A distribution of the number of non-zero entries in the NC headers
of transmitted packets is estimated, averaging $10$ network realizations.
Figure~\ref{fig:AvgNonZero} describes the average amount of non-zero
coefficients in packets brodcast by the nodes of the network. One
observes that this number increases almost linearly with $g$. The
number of transmitted packets at simulation termination does not depend
on the way headers are represented. The estimated distribution is
used to evaluate the average COPE-inspired NC header length.

\begin{figure}
\begin{centering}
\includegraphics[width=0.45\columnwidth]{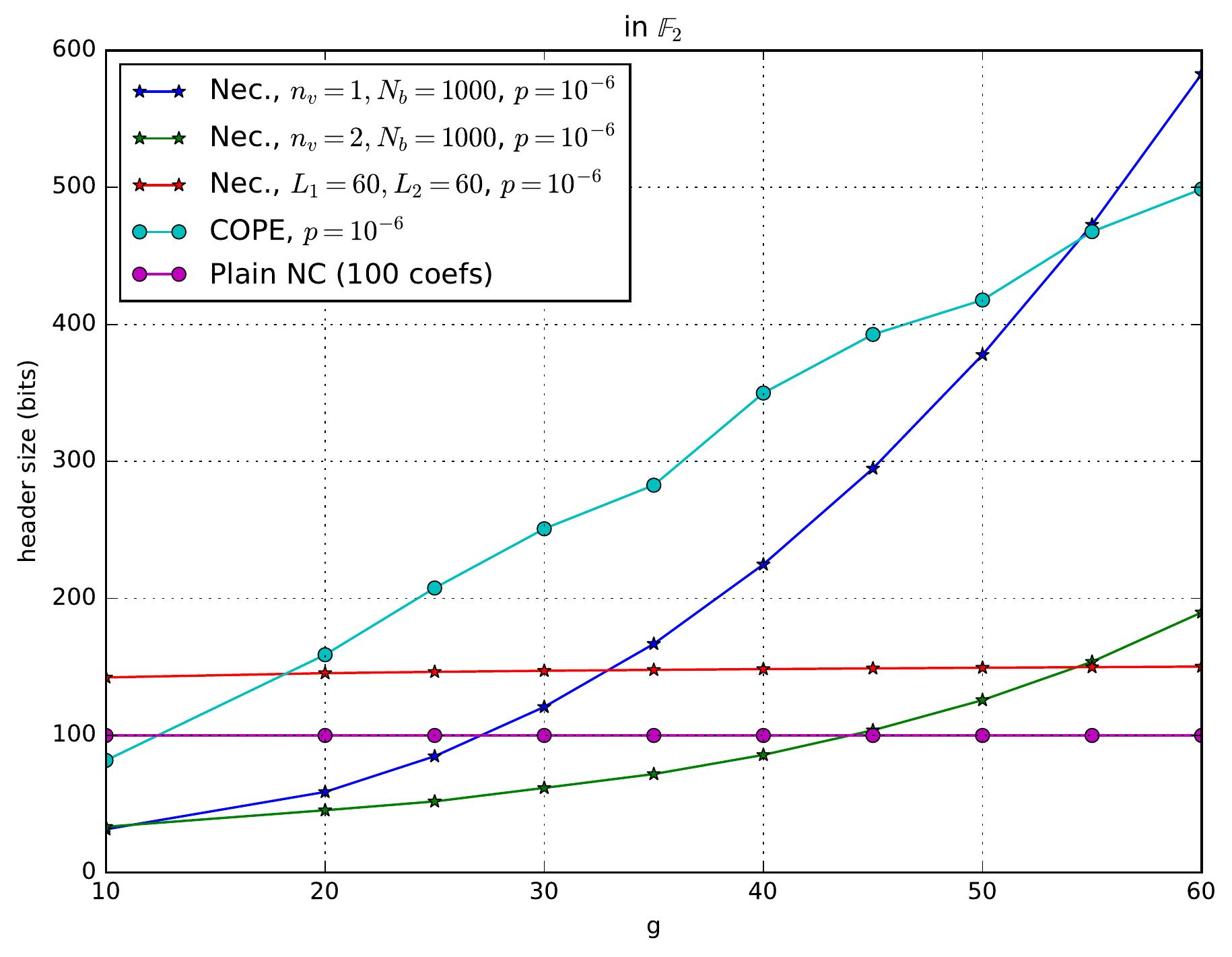}
\includegraphics[width=0.45\columnwidth]{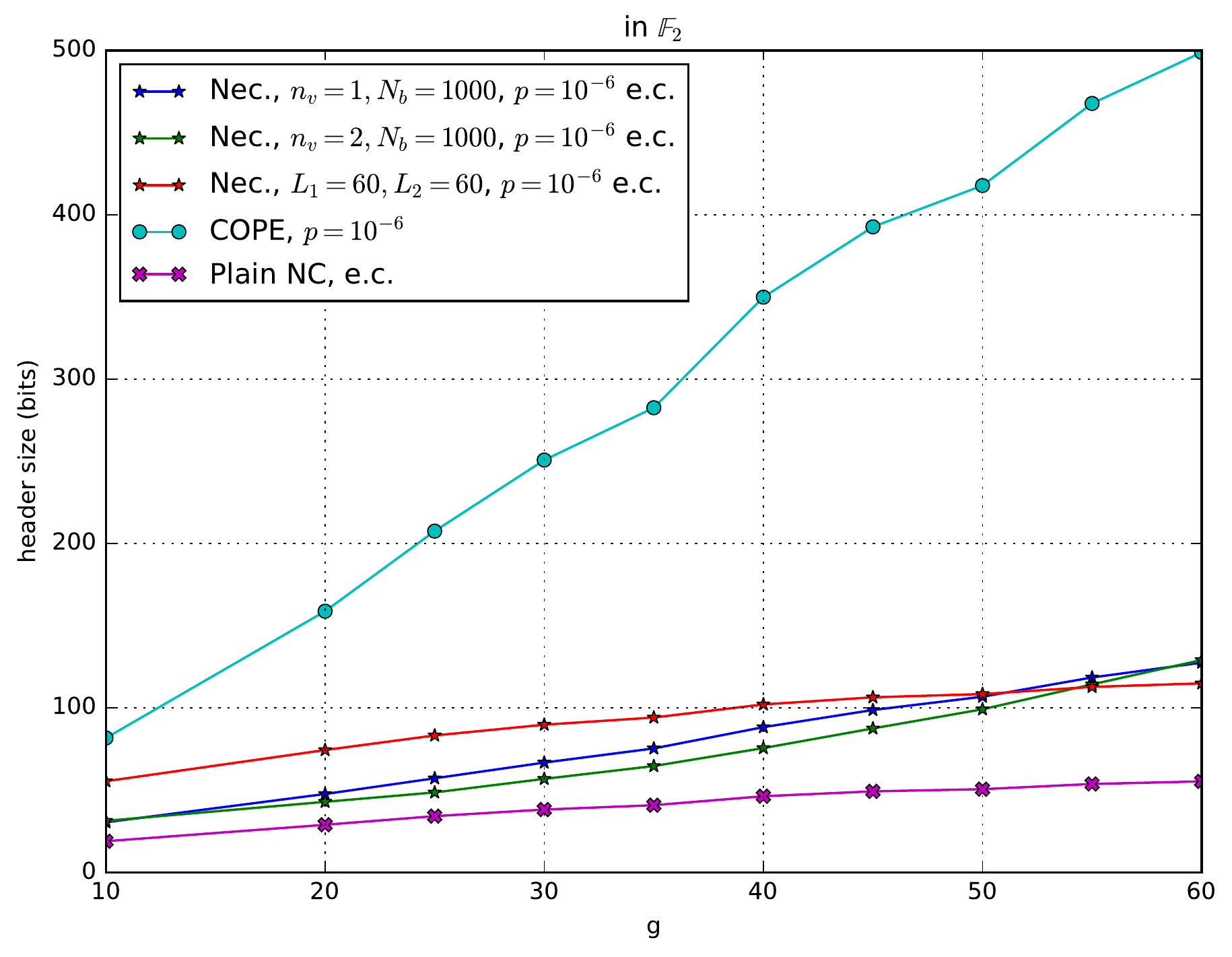}
\par\end{centering}
\protect\caption{Evolution of the header length as a function of the number of active
nodes $g$ in a STS for the plain NC protocol, for COPE, and for NeCoRPIA
with $n_{\text{v}}=1$ and $n_{\text{v}}=2$ and subvectors of variable
lengths ensuring that the number of decoding hypotheses is less than
$N_{b}=1000$, and with $n_{\text{v}}=2$ and fixed subvector lengths
$L_{1}=L_{2}=60$; the decoding error probability is imposed to be
less than $p_{\text{c}}=10^{-6}$, without entropy coding (left) and
with entropy coding of the NC headers (right)\label{fig:HeaderLength}}
\end{figure}
Figure~\ref{fig:HeaderLength} illustrates the evolution of the size
of the NC headers as a function of $g$ for plain NC, COPE, and NeCoRPIA
headers with $n_{\text{v}}=1$ and $n_{\text{v}}=2$. For the two
last approaches, $p_{\text{c}}$ is taken as $10^{-6}$. For NeCoRPIA
headers, either a fixed-size header with $L_{1}=L_{2}=60$ is considered
or a variable size header that ensures that the number of decoding
hypotheses is less than $N_{b}=1000$, limiting thus the decoding
complexity. The header overhead necessary to identify the STS is omitted,
since it is the same in all cases. In Figure~\ref{fig:HeaderLength}
(left), the NC headers are not compressed, wheras in Figure~\ref{fig:HeaderLength}
(right), the plain NC and NeCoRPIA headers are assumed to be entropy
coded and their coded length is provided.

In absence of entropy coding of the NC header, its size is constant
with plain NC. It increases almost linearly with $g$ for COPE and
is larger than the header of plain NC as soon as $g\geqslant12$ and
than the NeCoRPIA header in most of the cases. Considering $n_{\text{v}}=2$
and an adaptive header length provides the best results. When $g=40$,
the header length of NeCoRPIA is only one fifth of that of COPE. The
header length is less than that with plain NC as long as $g<42.$
Thus, without prior agreement of the packets generated in a STS, with
NeCoRPIA, one is able to get header length smaller than those obtained
with plain NC which requires this agreement phase.

When the NC headers are entropy-coded, the average length of entropy-coded
NC headers with plain NC increases almost linearly with $g$. When
$n_{\text{v}}=1$, a significant reduction of the header length is
obtained by entropy coding, since the header has to be very large
to avoid collisions and contains only few non-zero coefficients, see
Figure~\ref{fig:AvgNonZero}. When $n_{\text{v}}=2$, entropy coding
reduces slightly the average header length, which remains less than
that obtained with $n_{\text{v}}=1$. In average, the compressed header
length with NeCoRPIA is only twice that obtained with plain NC.

\section{Conclusions\label{sec:conclusion}}

This paper presents NeCoRPIA, a NC algorithm with random packet index
assignment. This technique is well-suited to data collection using
MCS, as it does not require any prior agreement on the NC vectors,
which are chosen randomly. As a consequence, different packets may
share the same coding vector, leading to a collision, and to the impossibility
to perform decoding with standard Gaussian elimination.  Collisions
are more frequent when the size of the generation increases. A branch-and-prune
approach is adapted to decode in presence of collisions. This approach
is efficient in presence of a low number of collisions. To reduce
the number of collisions, we propose to split the  NC vector into
subvectors. Each packet header consists then of two or more NC subvectors.

A detailed analysis of the decoding complexity and of the probability
of decoding error shows the potential of this approach: when a NC
of $L=100$ elements in $\mathbb{F}_{2}$ is split into two NC subvectors,
generations of about 60 packets may be considered with a decoding
complexity that is about $10$ times that of plain network decoding.
When the NC vector is split into 4 subvectors, generations of about
80 packets may be considered.

A comparison of the average header length required to get a given
probability of network decoding error is provided for a COPE-inspired
variable-length NC protocol, for several variants of NeCoRPIA, and
for plain NC on random network topologies illustrating the data collection
using a network of sensors to a sink. The effect of entropy coding
on the header length has also been analyzed. Without and with entropy
coding, NeCoRPIA provides header sizes which are significantly smaller
than those obtained with COPE. Compared to plain NC, in absence of
entropy coding, headers are smaller with NeCoRPIA when the number
of active nodes remains moderate.

Future research will be devoted to the development of a data collection
protocol based on NeCoRPIA, and more specifically on the adaptation
of the STS as a function of the network activity.


\appendix

\subsection{Arithmetic complexity of DeRPIA-SLE}

\label{sub:ComplexityDERPIASLE}

\subsubsection{Arithmetic complexity for intermediate branches}

Consider Level $\ell=1$. For each of the $L_{1}$ canonical vectors
$\mathbf{e}_{j_{1}}\in\mathbb{F}_{q}^{L_{1}}$, finding $\mathbf{w}_{1}\in\mathbb{F}_{q}^{\rho_{1}}$
satisfying \eqref{eq:ConstrlPivot-1} takes at most $\rho_{1}L_{1}$
operations, since, according to Theorem~\ref{thm:OneElementOnly},
one searches for a row of $\mathbf{B}_{11}\in\mathbb{F}_{q}^{\rho_{1}\times L_{1}}$
equal to $\mathbf{e}_{j_{1}}$. The arithmetic complexity to get all
branches at Level $1$ is thus upper-bounded by 
\begin{equation}
K\left(1\right)=L_{1}\rho_{1}L_{1}.\label{eq:Kl-1}
\end{equation}

Consider Level $\ell$ with $1<\ell\leqslant n_{\text{v}}$. For each
branch associated to a candidate decoding vector $\left(\mathbf{w}_{1},\dots,\mathbf{w}_{\ell-1}\right),$
one first evaluates $-\left(\mathbf{w}_{1}\mathbf{B}_{1,\ell}+\dots+\mathbf{w}_{\ell-1}\mathbf{B}_{\ell-1,\ell}\right)$
which takes $K_{\text{m}}(\rho_{1}L_{\ell}+...+\rho_{\ell-1}L_{\ell})+\left(\ell-1\right)L_{\ell}$
operations, the last term $\left(\ell-1\right)L_{\ell}$ accouting
for the additions of the vector-matrix products and the final sign
change.

Then, for each of the $L_{\ell}$ canonical vectors $\mathbf{e}_{j_{\ell}}\in\mathbb{F}_{q}^{L_{\ell}}$,
evaluating $\mathbf{e}_{j_{\ell}}-\left(\mathbf{w}_{1}\mathbf{B}_{1,\ell}+\dots+\mathbf{w}_{\ell-1}\mathbf{B}_{\ell-1,\ell}\right)$
needs a single addition, and finding $\mathbf{w_{\ell}}$ satisfying
\eqref{eq:ConstrlPivot-1} takes at most $\rho_{\ell}L_{\ell}$ operations,
since one searches for a row of $\mathbf{B}_{\ell\ell}$ equal to
$\mathbf{e}_{j_{\ell}}-\left(\mathbf{w}_{1}\mathbf{B}_{1,\ell}+\dots+\mathbf{w}_{\ell-1}\mathbf{B}_{\ell-1,\ell}\right)$.
Additionally, one has to verify whether $\mathbf{w_{\ell}}=\mathbf{0}$
is satisfying, which requires $L_{\ell}$ operations.

At Level $\ell$, the arithmetic complexity, for a given candidate
$\left(\mathbf{w}_{1},\dots,\mathbf{w}_{\ell-1}\right)$, to find
all candidates $\left(\mathbf{w}_{1},\dots,\mathbf{w}_{\ell}\right)$
satisfying \eqref{eq:Constrl1} is thus upper-bounded by 
\begin{align}
K\left(\ell\right) & =K_{\text{m}}\left(\rho_{1}L_{\ell}+...+\rho_{\ell-1}L_{\ell}\right)+\left(\ell-1\right)L_{\ell}+L_{\ell}\left(1+\rho_{\ell}L_{\ell}+L_{\ell}\right)\label{eq:Kl}\\
 & =K_{\text{m}}L_{\ell}\left(\rho_{1}+...+\rho_{\ell-1}\right)+L_{\ell}\left(\ell+\left(\rho_{\ell}+1\right)L_{\ell}\right)
\end{align}

\subsubsection{Arithmetic complexity for terminal branches}

Consider now Level $n_{\text{v}}+1$. For each candidate $\left(\mathbf{w}_{1},\dots,\mathbf{w}_{n_{\text{v}}}\right)$,
one has first to compute $\mathbf{w}_{1}\mathbf{C}_{1}+...+\mathbf{w}_{n_{\text{v}}}\mathbf{C}_{n_{\text{v}}}$,
which requires $K_{\text{m}}\left(\rho_{1}L_{\text{p}}+\rho_{2}L_{\text{p}}+\ldots+\rho_{n_{\text{v}}}L_{\text{p}}\right)+\left(n_{\text{v}}-1\right)L_{\text{p}}\leqslant K_{\text{m}}gL_{\text{p}}+\left(n_{\text{v}}-1\right)L_{\text{p}}$
operations. Then for each candidate $\mathbf{w}_{n_{\text{v}}+1}\in\mathbb{F}_{q}^{\rho_{n_{\text{v}}+1}}$,
the evaluation of $\mathbf{w}_{1}\mathbf{C}_{1}+...+\mathbf{w}_{n_{\text{v}}}\mathbf{C}_{n_{\text{v}}}+\mathbf{w}_{n_{\text{v}}+1}\mathbf{C}_{n_{\text{v}}+1}$
and the checksum verification cost $K_{\text{m}}\rho_{n_{\text{v}}+1}L_{\text{p}}+L_{\text{p}}+K_{\text{c}}L_{\text{p}}$
operations. The arithmetic complexity, for a given candidate $\left(\mathbf{w}_{1},\dots,\mathbf{w}_{n_{\text{v}}}\right)$,
to find all solutions $\left(\mathbf{w}_{1},\dots,\mathbf{w}_{n_{\text{v}}+1}\right)$
is thus upper-bounded by 
\begin{equation}
K\left(n_{\text{v}}+1\right)=K_{\text{m}}gL_{\text{p}}+\left(n_{\text{v}}-1\right)L_{\text{p}}+q^{\rho_{n_{\text{v}}+1}}\left(K_{\text{m}}\rho_{n_{\text{v}}+1}L_{\text{p}}+L_{\text{p}}+K_{\text{c}}L_{\text{p}}\right).\label{eq:KL1}
\end{equation}

\subsubsection{Total arithmetic complexity of DeRPIA-SLE}

To upper bound the arithmetic complexity of the tree traversal algorithm
when the number of encoding vectors is $n_{\text{v}}$, one combines
\eqref{eq:Nnl}, \eqref{eq:Kl}, and \eqref{eq:KL1} to get \eqref{eq:KSLE}
with $N_{\text{b}}\left(0\right)=1$. In the case $n_{\text{v}}=1$,
\eqref{eq:Kt1} follows directly from \eqref{eq:KSLE}.

\subsection{Arithmetic complexity of DeRPIA-LUT}

\label{sec:ArithCompDERPIALUT}

One first evaluates the complexity of the look-up table construction
with Algorithm~2a. The look-up table is built once, after the RREF
evaluation. The worst-case complexity is evaluated, assuming that
for each row vector $\mathbf{u}$ of $\mathbf{B}_{\ell\ell}$, the
resulting vector $\overline{\mathbf{u}}$ is added to $\Pi_{\ell}$.

For each of the $\rho_{\ell}$ lines $\mathbf{u}$ of $\mathbf{B}_{\ell\ell}$,
the identification of the index of its pivot column takes at most
$L_{\ell}$ operations. The evaluation of $\overline{\mathbf{u}}$
takes one operation. Then determining whether $\overline{\mathbf{u}}\in\Pi_{\ell}$
takes no operation for the first vector ($\Pi_{\ell}$ is empty),
at most $L_{\ell}$ operations for the second vector, at most $2L_{\ell}$
operations for the third vector, and at most $\left(\rho_{\ell}-1\right)L_{\ell}$
operations for the last vector. The number of operations required
in this step is upper bounded by 
\begin{align}
K_{\text{LU},\text{1}}\left(\ell\right) & =\rho_{\ell}\left(L_{\ell}+1\right)+0+L_{\ell}+2L_{\ell}+\dots+\left(\rho_{\ell}-1\right)L_{\ell}\nonumber \\
 & =\rho_{\ell}+L_{\ell}\frac{\rho_{\ell}\left(\rho_{\ell}+1\right)}{2}.\label{eq:KLU1}
\end{align}

Then the canonical vectors $\mathbf{e}_{i}\in\mathbb{F}_{q}^{\rho_{\ell}},i=1,\dots,\rho_{\ell}$
have to be partitionned into the various $\mathcal{S_{\ell}}\left(\mathbf{v}\right)$,
$\mathbf{v}\in\Pi_{\ell}$. This is done by considering again each
of the $\rho_{\ell}$ lines $\mathbf{u}$ of $\mathbf{B}_{\ell\ell}$,
evaluating $\mathbf{u}-\mathbf{e}_{\gamma\left(\mathbf{u}\right)}$,
which needs up to $L_{\ell}+1$ operations. Then determining the vectors
$\mathbf{v}\in\Pi_{\ell}$ such that $\mathbf{u}-\mathbf{e}_{\gamma\left(\mathbf{u}\right)}=\mathbf{v}$
requires at most $\rho_{\ell}L_{\ell}$ operations, since $\Pi_{\ell}$
contains at most $\rho_{\ell}$ vectors of $L_{\ell}$ elements. The
number of operations required in this partitionning is upper bounded
by 
\begin{equation}
K_{\text{LU},2}\left(\ell\right)=\rho_{\ell}\left(L_{\ell}+1+\rho_{\ell}L_{\ell}\right).\label{eq:KLU2}
\end{equation}

Considering a satisfying $\left(\mathbf{w}_{1},\dots,\mathbf{w}_{\ell-1}\right)$
at Level~$\ell-1$, the search complexity to find some satisfying
$\left(\mathbf{w}_{1},\dots,\mathbf{w}_{\ell-1},\mathbf{w}_{\ell}\right)$
at Level~$\ell$ in Algorithm~2b is now evaluated. The evaluation
of $\mathbf{v}=-\left(\mathbf{w}_{1}\mathbf{B}_{1,\ell}+\dots+\mathbf{w}_{\ell-1}\mathbf{B}_{\ell-1,\ell}\right)$
takes $K_{\text{m}}(\rho_{1}L_{\ell}+...+\rho_{\ell-1}L_{\ell})+\left(\ell-1\right)L_{\ell}$
operations. Then, determining whether $\mathbf{w}_{\ell}=\mathbf{0}$
satisfies \eqref{eq:Constrl1}, \emph{i.e.}, whether $\mathbf{v}+\mathbf{e}_{j_{\ell}}=\mathbf{0}$
for some canonical vector $\mathbf{e}_{j_{\ell}}\in\mathbb{F}_{q}^{L_{\ell}}$,
can be made checking whether $\mathbf{v}$ contains a single non-zero
entry in $L_{\ell}$ operations. Finally, the look-up of\textbf{ }$\mathbf{v}\in\Pi_{\ell}$
takes at most $\rho_{\ell}L_{\ell}$ operations. In summary, the number
of operations required for this part of the algorithm is 
\begin{align}
K_{\text{LU},3}\left(\ell\right) & =K_{m}(\rho_{1}L_{\ell}+...+\rho_{\ell-1}L_{\ell})+\left(\ell-1\right)L_{\ell}+L_{\ell}+\rho_{\ell}L_{\ell}\nonumber \\
 & =K_{m}L_{\ell}(\rho_{1}+...+\rho_{\ell-1})+L_{\ell}\left(\ell+\rho_{\ell}\right).\label{eq:KLU3}
\end{align}
Compared to the expression of $K\left(\ell\right)$ given by \eqref{eq:Kl},
which is quadratic in $L_{\ell}$, $K_{\text{LU},3}\left(\ell\right)$
is linear in $L_{\ell}$.

The complexity of DeRPIA-LUT, given by \eqref{eq:KLUT}, when the
number of encoding vectors is $n_{\text{v}}$, is then obtained combining
the results of Corollary~\ref{cor:Complexity} with \eqref{eq:KLU1},
\eqref{eq:KLU3}, \eqref{eq:KLU3}, and \eqref{eq:KL1}, since Algorithm~2b
is not used at Level~$n_{\text{v}}+1$. 
\end{document}